\documentclass[a4paper,onecolumn,11pt]{quantumarticle}
\pdfoutput=1

\usepackage{fullpage}
\usepackage{amsmath}
\usepackage{palatino,amsthm,amsfonts,amscd}
\usepackage{accents}
\usepackage{xcolor}
\usepackage{xspace}
\usepackage{amssymb}

\usepackage{subcaption}
\usepackage{tikz}
\usepackage{mathrsfs}

\usepackage{hyperref}
\usepackage{cleveref}                                            %
\definecolor{webgreen}{rgb}{0,.5,0}
\definecolor{webblue}{rgb}{0,0,.5}
\hypersetup{
    colorlinks,
    linkcolor={green!50!black},
    citecolor={green!50!black},
    urlcolor={blue!50!black},
    filecolor={blue!50!black},
}

\newtheorem{lemma}{Lemma}[section]
\newtheorem{theorem}[lemma]{Theorem}
\newtheorem{corollary}[lemma]{Corollary}
\newtheorem{definition}[lemma]{Definition}

\theoremstyle{definition}
\newtheorem{protocol}[lemma]{Protocol}
\newtheorem{circuit}[lemma]{Circuit}
\newtheorem*{remark}{Remark}
\newtheorem{simulator}{Simulator}
\definecolor{protocol-bg}{RGB}{240,240,240}
\usepackage[framemethod=tikz]{mdframed}
\surroundwithmdframed[outerlinewidth=0pt,
innerlinewidth=0pt,
middlelinewidth=0pt,
backgroundcolor=protocol-bg
]
{protocol}
\surroundwithmdframed[outerlinewidth=0pt,
innerlinewidth=0pt,
middlelinewidth=0pt,
backgroundcolor=protocol-bg
]
{circuit}
\definecolor{simulator-bg}{RGB}{250,250,250}
\definecolor{simulator-line}{RGB}{200,200,200}
\usepackage[framemethod=tikz]{mdframed}
\surroundwithmdframed[outerlinewidth=1pt, outerlinecolor=simulator-line,
innerlinewidth=0pt,
middlelinewidth=0pt,
backgroundcolor=simulator-bg
]
{simulator}
\include{Qcircuit}
\newcommand{\meas}{
\begin{tikzpicture}
\filldraw[fill=white] (0,.25) rectangle (.7,-.25);
\draw (.67,-.1) arc (50:130:.5);
\draw (.35,-.2)--(.525,.2);
\end{tikzpicture}
}
\newcommand{\smeas}[1]{%
	\begin{tikzpicture}[#1]%
	\draw (1.34ex,-0.2ex) arc (50:130:1ex);%
	\draw (0.7ex,-0.4ex) -- (1.05ex,0.4ex);%
	\end{tikzpicture}%
}

\newcommand\Tr{{\mathop\textup{Tr}}}
\newcommand{\norm}[1]{\left\|\,#1\,\right\|}

\newcommand{\ket}[1]{\left|#1\right\rangle}
\newcommand{\bra}[1]{\left\langle #1\right|}

\newcommand{\ketbra}[2]{\ket{#1}\!\bra{#2}}
\newcommand{\proj}[1]{\ketbra{#1}{#1}}
\newcommand{\kb}[1]{\ketbra{#1}{#1}}

\newcommand{\Id}{\ensuremath{\mathbb{I}}\xspace}
\newcommand{\X}{\ensuremath{\mathsf{X}}\xspace}
\newcommand{\Z}{\ensuremath{\mathsf{Z}}\xspace}
\newcommand{\I}{\ensuremath{\mathsf{I}}\xspace}
\renewcommand{\H}{\ensuremath{\mathsf{H}}\xspace}

\newcommand{\T}{\ensuremath{\mathsf{T}}\xspace}
\newcommand{\CNOT}{\ensuremath{\mathsf{CNOT}}\xspace}
\newcommand{\Clifford}{\mathscr{C}}
\newcommand{\Pauli}{\mathscr{P}}
\newcommand{\reg}[1]{{\color{gray}#1}}

\newcommand{\F}{\mathbb{F}}

\newcommand{\eps}{\varepsilon}

\renewcommand{\leq}{\leqslant}
\renewcommand{\geq}{\geqslant}
\newcommand{\negl}[1]{\mathsf{negl}\left(#1\right)}
\newcommand{\poly}[1]{\mathsf{poly}\left(#1\right)}
\DeclareMathOperator*{\E}{\mathbb{E}}
\newcommand{\dbtilde}[1]
{\tilde{\raisebox{0pt}[0.85\height]{$\tilde{#1}$}}}

\newcommand{\inreg}{\mathsf{in}}
\newcommand{\outreg}{\mathsf{out}}
\newcommand{\ancillareg}{\mathsf{ancilla}}
\newcommand{\discardreg}{\mathsf{discard}}

\newcommand{\RealProtocol}{\Pi}
\newcommand{\IdealProtocol}{\mathfrak{I}}
\newcommand{\envE}{\mathcal{E}}
\newcommand{\simS}{\mathcal{S}}
\newcommand{\advA}{\mathcal{A}}
\newcommand{\MPC}{\ensuremath{\mathsf{MPC}}\xspace}
\newcommand{\MPQC}{\ensuremath{\mathsf{MPQC}}\xspace}
\newcommand{\Enc}{\ensuremath{\mathsf{Enc}}\xspace}
\newcommand{\Cliff}{\ensuremath{\mathsf{Cliff}}\xspace}
\newcommand{\Dec}{\ensuremath{\mathsf{Dec}}\xspace}

\newcommand{\twirl}[1]{\mathcal{T}_{#1}}
\newcommand{\Twirl}[1]{\mathcal{T}_{#1}}
\newcommand{\KeySet}{\mathcal{K}}
\newcommand{\Attack}{\Lambda}
\newcommand{\acc}{\ensuremath{\mathsf{acc}}\xspace}
\newcommand{\rej}{\ensuremath{\mathsf{rej}}\xspace}

\newcommand{\paulifilter}[2]{\ensuremath{\mathsf{PauliFilter}}_{#1}^{#2}\xspace}
\newcommand{\idfilter}[1]{\ensuremath{\mathsf{IdFilter}}^{#1}\xspace}
\newcommand{\xfilter}[1]{\ensuremath{\mathsf{XFilter}}^{#1}\xspace}
\newcommand{\zerofilter}[1]{\ensuremath{\mathsf{ZeroFilter}}^{#1}\xspace}
\newcommand{\pauliset}{\mathcal{P}}
\newcommand{\pdistill}{p_{\mathrm{distill}}}
\newcommand{\ddistill}{d_{\mathrm{distill}}}

\title{Secure Multi-party Quantum Computation with a Dishonest Majority}

\usepackage{authblk}

\author[1]{Yfke Dulek} \email{yfkedulek@gmail.com} \thanks{This is the full version of an article published in the conference proceedings of EUROCRYPT 2020~\cite{DGJMS20}.}
\author[2]{Alex Grilo} \email{abgrilo@gmail.com}
\author[2]{Stacey Jeffery} \email{smjeffery@gmail.com}
\author[2]{Christian Majenz} \email{c.majenz@cwi.nl}
\author[1]{Christian Schaffner} \email{c.schaffner@uva.nl}

\affil[1]{QuSoft and University of Amsterdam, The Netherlands}
\affil[2]{QuSoft and CWI, Amsterdam, The Netherlands}

\date{\vspace{-1cm}}

\begin{document}
\maketitle

\begin{abstract}
  The cryptographic task of secure multi-party (classical) computation has received a lot of
  attention in the last decades. Even in the extreme case where a computation is
  performed between $k$ mutually distrustful players, and security is
  required even for the single honest player if all other players are colluding
  adversaries, secure protocols are known. For quantum computation, on the other
  hand, protocols allowing arbitrary dishonest majority have only been proven
  for $k=2$. In this work, we generalize the approach taken by Dupuis, Nielsen and
  Salvail (CRYPTO 2012) in the two-party setting to devise a secure, efficient protocol for
  multi-party quantum computation for any number of players $k$, and prove
  security against up to $k-1$ colluding adversaries. The quantum round
  complexity of the protocol for computing a quantum circuit of $\{\CNOT, \T\}$ depth $d$
  is $O(k \cdot (d + \log n))$, where $n$ is the security parameter.  To achieve efficiency, we develop a novel public verification protocol for the Clifford authentication code, and a testing protocol for magic-state inputs, both using classical multi-party computation.
\end{abstract}
\section{Introduction}
In secure multi-party computation (MPC), two or more players want to jointly
compute some publicly known
function on their private data, without revealing their inputs to the other
players.
Since its introduction by Yao~\cite{Yao82b}, MPC
has been extensively developed in different setups, leading to applications of both
theoretical and practical
interest (see, e.g.,~\cite{CramerDN2015} for a detailed overview).

With the emergence of quantum technologies, it becomes necessary to
understand its consequences in the field of MPC. First, classical MPC protocols
have to be secured against quantum attacks. But also, the increasing
number of applications where quantum computational power is desired motivates protocols enabling multi-party {\em quantum}
computation (MPQC) on the players' private (possibly quantum) data. In this
work, we focus on the second task. Informally, we say a MPQC protocol is secure
if the following two properties hold: 1. Dishonest players gain no information
about the honest players' private inputs. 2. If the players do not abort the protocol,
then at the end of the protocol they share a state corresponding to the correct computation applied to the inputs of honest players (those that follow the protocol) and some choice of inputs for the dishonest players. 

\smallskip

MPQC was first studied by Cr{\'e}peau, Gottesman and Smith~\cite{CGS02}, who
proposed a $k$-party protocol based on verifiable secret sharing that is
information-theoretically secure, but requires the assumption that at most $k/6$ players
are dishonest. The fraction $k/6$ was subsequently
improved to $<k/2$~\cite{BCG06} which is optimal for secret-sharing-based
protocols due to no-cloning. The case of a dishonest majority was thus far only
considered for $k=2$ parties, where one of the two players can be dishonest
~\cite{DNS10,DNS12,Kashefi2017}\footnote{In Kashefi and Pappa~\cite{KashefiP17}, they consider a
non-symmetric setting where the protocol is secure only when some specific sets
of $k-1$ players are dishonest.}. These protocols are
based on different cryptographic techniques, in particular quantum
authentication codes in conjunction with classical MPC \cite{DNS10,DNS12} and
quantum-secure bit commitment and oblivious transfer \cite{Kashefi2017}.

\smallskip

In this work, we propose the first secure MPQC protocol for any number $k$ of
players in the dishonest majority setting, i.e., the case with up to $k-1$ colluding adversarial
players.\footnote{In the case where there are $k$ adversaries and no honest
players, there is nobody whose input privacy and output authenticity is worth
protecting.}  We remark that our result achieves {\em composable security},
which is proven according to the standard ideal-vs.-real definition. Like the protocol of \cite{DNS12}, on which our protocol is built, our protocol assumes a classical MPC that is secure against a dishonest majority, and achieves the same security guarantees as this classical MPC. In particular, if we instantiate this classical MPC with an MPC in the \emph{pre-processing model} (see~\cite{Bendlin2011,Damgaard2012,Keller2018,Cramer2018}), our construction yields a MPQC protocol consisting of a classical ``offline'' phase used to produce authenticated shared randomness among the players, and a second ``computation'' phase, consisting of our protocol, combined with the ``computation'' phase of the classical MPC. The security of the ``offline'' phase requires computational assumptions, but assuming no attack was successful in this phase, the second phase has information-theoretic security.

\subsection{Prior work}\label{sec:prelim-dns}

Our protocol builds on the two-party protocol of Dupuis, Nielsen, and Salvail~\cite{DNS12}, which we now describe in brief. 
The protocol uses a classical MPC protocol, and involves two parties, Alice and Bob, of whom at least one is honestly following the protocol. Alice and Bob encode their inputs using a technique called \emph{swaddling}: if Alice has an input qubit $\ket{\psi}$, she first encodes it using the $n$-qubit Clifford code (see \Cref{def:clifford-code}), resulting in $A(\ket{0^n} \otimes \ket{\psi})$, for some random $(n+1)$-qubit Clifford $A$ sampled by Alice, where $n$ is the security parameter. Then, she sends the state to Bob, who puts another encoding on top of Alice's: he creates the ``swaddled'' state $B(A(\ket{0^n} \otimes \ket{\psi}) \otimes \ket{0^n})$ for some random $(2n+1)$-qubit Clifford $B$ sampled by Bob. This encoded state consists of $2n+1$ qubits, and the data qubit $\ket{\psi}$ sits in the middle.

If Bob wants to test the state at some point during the protocol, he simply needs to undo the Clifford $B$, and test that the last $n$ qubits (called traps) are $\ket{0}$. However, if Alice wants to test the state, she needs to work together with Bob to access her traps. Using classical multi-party computation, they jointly sample a random $(n+1)$-qubit Clifford $B'$ which is only revealed to Bob, and compute a Clifford $T := (\Id^{\otimes n} \otimes B')(A^{\dagger} \otimes \Id^{\otimes n})B^{\dagger}$ that is only revealed to Alice. Alice, who will not learn any relevant information about $B$ or $B'$, can use $T$ to ``flip" the swaddle, revealing her $n$ trap qubits for measurement. After checking that the first $n$ qubits are $\ket{0}$, she adds a fresh $(2n+1)$-qubit Clifford on top of the state to re-encode the state, before computation can continue.

Single-qubit Clifford gates are performed simply by classically updating the inner key: if a state is encrypted with Cliffords $BA$, updating the decryption key to $BAG^\dagger$ effectively applies the gate $G$. In order to avoid that the player holding the inner key $B$ skips this step, both players keep track of their keys using a classical commitment scheme. This can be encapsulated in the classical MPC, which we can assume acts as a trusted third party with a memory~\cite{BCG06}.

\CNOT operations and measurements are slightly more involved, and require both players to test the authenticity of the relevant states several times. Hence, the communication complexity scales linearly with the number of \CNOT{}s and measurements in the circuit.

Finally, to perform \T gates, the protocol makes use of so-called magic states. To obtain reliable magic states, Alice generates a large number of them, so that Bob can test a sufficiently large fraction. He decodes them (with Alice's help), and measures whether they are in the expected state. If all measurements succeed, Bob can be sufficiently certain that the untested (but still encoded) magic states are in the correct state as well.

\subsubsection{Extending two-party computation to multi-party computation}
A natural question is how to lift a two-party computation
protocol to a multi-party computation protocol. We
discuss some of the issues that arise from such an approach, making it
either infeasible or inefficient.

\paragraph{Composing ideal functionalities.} The first naive idea would be
trying to split the $k$ players in two groups and make the groups simulate the players
of a two-party protocol, whereas internally, the players run $\frac{k}{2}$-party
computation protocols for all steps in the two-party protocol. Those $\frac{k}{2}$-party protocols are in turn realized by running $\frac{k}{4}$-party protocols, et cetera, until at the lowest level, the players can run actual two-party protocols.

Trying to construct such a composition in a black-box way, using the {\em ideal functionality} of
a two-party protocol, one immediately faces a problem: at the lower levels, players learn intermediate states of the circuit, because they receive plaintext outputs from the ideal two-party functionality. This would immediately break the privacy of the protocol. If, on the other hand, we require the ideal two-party functionality to output encoded states instead of plaintexts, then the size of the ciphertext will grow at each level. The overhead of this approach would be $O(n^{\log k})$,
where $n \geq k$ is the security parameter of the encoding,
which would make this overhead super-polynomial in the number of players.

\paragraph{Naive extension of DNS to multi-party.} One could also try to extend ~\cite{DNS12} to multiple parties by adapting the subprotocols to work for more than two players. While this approach would likely lead to a correct and secure protocol for $k$ parties, the computational costs of such an extension could be high.

First, note that in such an extension, each party would need to append $n$
trap qubits to the encoding of each qubit, causing an overhead in the
ciphertext size that is linear in $k$. Secondly, in this naive extension, the players
would need to create $\Theta(2^k)$ magic states for \T gates (see \Cref{sec:t-magic-states}), since each party would need to sequentially test at least half of the ones approved by all previous players.

Notice that in both this extension and our protocol, a state has to pass by the honest player
(and therefore all players) in order to be able to verify that it has been properly encoded.

\subsection{Our contributions}

Our protocol builds on the work of Dupuis, Nielsen, and
Salvail~\cite{DNS10,DNS12}, and like it, assumes a classical MPC, and achieves
the same security guarantees as this classical MPC. In contrast to a naive extension of \cite{DNS12}, requiring
$\Theta(2^k)$ magic states, the complexity of our protocol, when
considering a quantum circuit that contains, among
other gates,  $g$ gates in $\{\CNOT , \T \}$ and acts  on $w$
qubits, scales as $O((g+w)k)$.

In order to efficiently extend the two-party protocol of \cite{DNS12} to a general $k$-party protocol, we make two major alterations to the protocol:

\vspace{.1cm}

\noindent	\textbf{Public authentication test.} In~\cite{DNS12}, given a security parameter $n$, each party adds $n$ qubits in the state $\ket{0}$ to each input qubit in order to authenticate it. The size of each ciphertext is thus $2n+1$. The extra qubits serve as check qubits (or ``traps") for each party, which can be measured at regular intervals: if they are non-zero, somebody tampered with the state.

	In a straightforward generalization to $k$ parties, the ciphertext size would become $kn+1$ per input qubit, putting a strain on the computing space of each player. In our protocol, the ciphertext size is constant in the number of players: it is usually $n+1$ per input qubit, temporarily increasing to $2n+1$ for qubits that are involved in a computation step. As an additional advantage, our protocol does not require that all players measure their traps every time a state needs to be checked for its authenticity.

	To achieve this smaller ciphertext size, we introduce a \emph{public authentication test}. Our protocol uses a single, shared set of traps for each qubit. If the protocol calls for the authentication to be checked, the player that currently holds the state cannot be trusted to simply measure those traps. Instead, she temporarily adds extra trap qubits, and fills them with an encrypted version of the content of the existing traps. Now she measures only the newly created ones. The encryption ensures that the measuring player does not know the expected measurement outcome. If she is dishonest and has tampered with the state, she would have to guess a random $n$-bit string, or be detected by the other players. We design a similar test that checks whether a player has honestly created the first set of traps for their input at encoding time. %

\vspace{.1cm}

\noindent\textbf{Efficient magic-state preparation.} For the computation of non-Clifford gates, the \cite{DNS12} protocol requires the existence of authenticated ``magic states", auxiliary qubits in a known and fixed state that aid in the computation. In a two-party setting, one of the players can create a large number of such states, and the other player can, if he distrusts the first player, test a random subset of them to check if they were honestly initialized. Those tested states are discarded, and the remaining states are used in the computation.

	In a $k$-party setting, such a ``cut-and-choose" strategy where all players
  want to test a sufficient number of states would require the first party to
  prepare an exponential number (in $k$) of authenticated magic states, which
  quickly gets infeasible as the number of players grows. Instead, we need a
  testing strategy where dishonest players have no control over which states are
  selected for testing. We ask the first player to create a polynomial number
  of authenticated magic states. Subsequently, we use classical MPC to sample random, disjoint subsets of the proposed magic states, one for each player. Each player continues to decrypt and test their subset of states. The random selection process implies that, conditioned on the test of the honest player(s) being successful, the remaining registers indeed contain encrypted states that are reasonably close to magic states. Finally, we use standard magic-state distillation to obtain auxiliary inputs that are exponentially close to magic states.

\subsection{Overview of the protocol}

We describe some details of the $k$-player quantum MPC protocol for circuits
consisting of classically-controlled Clifford operations and measurements. Such
circuits suffice to perform Clifford computation and magic-state distillation,
so that the protocol can be extended to arbitrary circuits using the technique
described above. The protocol consists of several subprotocols, of which we
highlight four here: input encoding, public authentication test, single-qubit gate application, and CNOT application. %
In the following description, the classical MPC is treated as a trusted third party with memory\footnote{The most common way to achieve classical MPC against dishonest majority is in the so called pre-processing model, as suggested by the SPDZ~\cite{Bendlin2011} and MASCOT~\cite{Keller2016} families of protocols. We believe that these protocols can be made post-quantum secure, but that is beyond the scope of this paper.}. The general idea is to first ensure that initially all inputs are properly encoded into the Clifford authentication code, and to test the encoding after each computation step that exposes the encoded qubit to an attack. During the protocol, the encryption keys for the Clifford authentication code are only known to the MPC.

\vspace{.1cm}

\noindent	\textbf{Input encoding.} For an input qubit $\ket\psi$ of player $i$,
the MPC hands each player a circuit for a random $(2n+1)$-qubit Clifford group
element. Now player $i$ appends $2n$ ``trap" qubits initialized in the $\ket
0$-state and applies her Clifford. The state is passed around, and all other
players apply their Clifford one-by-one, resulting in a Clifford-encoded qubit
$F(\ket{\psi}\ket{0^{2n}})$ for which knowledge of the encoding key $F$ is
distributed among all players. The final step is our \emph{public authentication
test}, which is used in several of the other subprotocols as well. Its goal is to ensure that all players, including player $i$, have honestly followed the protocol.

\vspace{.1cm}

\noindent \textbf{The public authentication test (details).}
The player holding the state $F(\ket{\psi}\ket{0^{2n}})$ (player $i$) will measure $n$ out of the $2n$ trap qubits, which should all be 0. To enable player $i$ to measure a random subset of $n$ of the trap qubits, the MPC could instruct her to apply $(E\otimes {\sf X}^r)(\mathbb{I}\otimes U_\pi) F^{\dagger}$ to get $E(\ket{\psi}\ket{0^n})\otimes \ket{r}$, where $U_\pi$ permutes the $2n$ trap qubits by a random permutation $\pi$, and $E$ is a random $(n+1)$ qubit Clifford, and $r\in\{0,1\}^n$ is a random string. Then when player $i$ measures the last $n$ trap qubits, if the encoding was correct, she will obtain $r$ and communicate this to the MPC. However, this only guarantees that the remaining traps are correct up to polynomial error. 

To get a stronger guarantee, we replace the random permutation with an element
from the sufficiently rich yet still efficiently samplable group of invertible
transformations over $\F^{2n}$, $\mathrm{GL}(2n,\F_2)$. An element $g\in
\mathrm{GL}(2n,\F_2)$ maybe be viewed as a unitary $U_g$ acting on computational
basis states as $U_g\ket{x}=\ket{gx}$ where $x \in \{0,1\}^{2n}$. In particular,
$U_g\ket{0^{2n}}=\ket{0^{2n}}$, so if all traps are in the state $\ket{0}$,
applying $U_g$ does not change this, whereas for non-zero $x$,
$U_g\ket{x}=\ket{x'}$ for a \emph{random} $x'\in \{0,1\}^{2n}$. Thus the MPC
instructs player $i$ to apply $(E\otimes {\sf X}^r)(\mathbb{I}\otimes U_g)
F^{\dagger}$ to the state $F(\ket{\psi}\ket{0^{2n}})$, then measure the last $n$
qubits and return the result, aborting if it is not $r$. Crucially, $(E\otimes
{\sf X}^r)(\mathbb{I}\otimes U_g) F^{\dagger}$ is given as an element of the
Clifford group, hiding the structure of the unitary and, more importantly, 
the values of $r$ and $g$. So if player $i$ is dishonest and holds a corrupted state, she can only pass the MPC's test by guessing $r$. If player $i$ correctly returns $r$, we have the guarantee that the remaining state is a Clifford-authenticated qubit with $n$ traps, $E(\ket{\psi}\ket{0^n})$, up to exponentially small error.

\vspace{.1cm}

\noindent \textbf{Single-qubit Clifford gate application.} As in \cite{DNS12}, this is done by simply updating encryption key held by the MPC: If a state is currently encrypted with a Clifford $E$, decrypting with a ``wrong" key $EG^\dagger$ has the effect of applying $G$ to the state.

\vspace{.1cm}

\noindent \textbf{CNOT application.}  Applying a CNOT gate to two qubits is slightly more complicated: as they are encrypted separately, we cannot just implement the CNOT via a key update like in the case of single qubit Clifford gates. Instead, we bring the two encoded qubits together, and then run a protocol that is similar to input encoding using the $(2n+2)$-qubit register as ``input", but using $2n$ additional traps instead of just $n$, and skipping the final authentication-testing step. The joint state now has $4n+2$ qubits and is encrypted with some Clifford $F$ only known to the MPC. Afterwards, \CNOT can be applied via a key update, similarly to single-qubit Cliffords. To split up the qubits again afterwards, the executing player applies $(E_1 \otimes E_2)F^{\dagger}$, where $E_1$ and $E_2$ are freshly sampled by the MPC. The two encoded qubits can then be tested separately using the public authentication test.

\subsection{Open problems}

Our results leave a number of exciting open problems to be addressed in future work. Firstly, the scope of this work was to provide a  protocol that reduces the problem of MPQC to classical MPC in an information-theoretically secure way. Hence we obtain an information-theoretically secure MPQC protocol \emph{in the preprocessing model}, leaving the post-quantum secure instantiation of the latter as an open problem. 

Another class of open problems concerns applications of MPQC. For instance, classically, MPC
can be used to devise zero-knowledge proofs~\cite{IshaiKOS09} and digital signature schemes~\cite{Chase2017}. 

An interesting  open question concerning our protocol more specifically is
whether the CNOT sub-protocol can be replaced by a different one that has round
complexity independent of the total number of players, reducing the quantum round complexity
of the whole protocol.
We also wonder if it is possible to develop more efficient protocols for
narrower classes of quantum computation, instead of arbitrary (polynomial-size) quantum circuits. 

Finally, it would be interesting to investigate whether the public authentication test we use can be leveraged in protocols for specific MPC-related tasks like oblivious transfer.

\subsection{Outline}
In \Cref{sec:prelim}, we outline the necessary preliminaries and tools we will make use of in our protocol. In \Cref{sec:defin}, we give a precise definition of MPQC. In \Cref{sec:setup-encoding}, we describe how players encode their inputs to setup for computation in our protocol. In \Cref{sec:clifford-computation} we describe our protocol for Clifford circuits, and finally, in \Cref{sec:T-gates}, we show how to extend this to universal quantum circuits in Clifford+\T.

\subsection*{Acknowledgments}
We thank Fr\'ed\'eric Dupuis, Florian Speelman, and Serge Fehr for useful discussions, and the anonymous EUROCRYPT referees for helpful comments and suggestions.
CM is supported by an NWO Veni Innovational Research Grant under project number VI.Veni.192.159.
SJ is supported by an NWO WISE Fellowship, an NWO Veni Innovational Research Grant under  project  number  639.021.752,  and  QuantERA  project  QuantAlgo  680-91-03.   SJ  is  a  CIFAR  Fellow  in  the Quantum Information Science Program. CS and CM were supported by a NWO VIDI grant (Project No. 639.022.519). Part of this work was done while YD, AG and CS were visiting the Simons Institute for the Theory of Computing.

\section{Preliminaries}\label{sec:prelim}

\subsection{Notation}\label{sec:notation}
We assume familiarity with standard notation in quantum computation, such as
(pure and mixed) quantum states, the Pauli gates \X and \Z, the Clifford gates
\H and \CNOT, the non-Clifford gate \T, and measurements.

We work in the quantum circuit model, with circuits $C$ composed of elementary
unitary gates (of the set Clifford+\T), plus computational basis measurements.
We consider those measurement gates to be destructive, i.e., to destroy the
post-measurement state immediately, and only a classical wire to remain. Since
subsequent gates in the circuit can still classically control on those measured wires, this point of view is as general as keeping the post-measurement states around.

For a set of quantum gates $\mathcal{G}$, the $\mathcal{G}$-depth of a quantum circuit is defined as the minimal number of layers such that in every layer, gates from $\mathcal{G}$ do not act on the same qubit.

For two circuits $C_1$ and $C_2$, we write $C_2 \circ C_1$ for the circuit that consists of executing $C_1$, followed by $C_2$. Similarly, for two protocols $\RealProtocol_1$ and $\RealProtocol_2$, we write $\RealProtocol_2 \diamond \RealProtocol_1$ for the execution of $\RealProtocol_1$, followed by the execution of $\RealProtocol_2$.

We use capital letters for both quantum registers ($M$, $R$, $S$, $T, \dots$) and unitaries ($A$, $B$, $U$, $V$, $W, \dots$). We write $|R|$ for the dimension of the Hilbert space in a register $R$. The registers in which a certain quantum state exists, or on which some map acts, are written as gray superscripts, whenever it may be unclear otherwise. For example, a unitary $U$ that acts on register $A$, applied to a state $\rho$ in the registers $AB$, is written as $U^{\reg{A}}\rho^{\reg{AB}}U^{\dagger}$, where the registers $U^\dagger$ acts on can be determined by finding the matching $U$ and reading the grey subscripts. Note that we do not explicitly write the operation $\Id^{\reg{B}}$ with which $U$ is in tensor product. The gray superscripts are purely informational, and do not signify any mathematical operation. If we want to denote, for example, a partial trace of the state $\rho^{\reg{AB}}$, we use the conventional notation $\rho_A$.

For an $n$-bit string $s = s_1s_2\cdots s_n$, define $U^s := U^{s_1} \otimes U^{s_2} \otimes \cdots \otimes U^{s_n}$. For an $n$-element permutation $\pi \in S_n$, define $P_{\pi}$ to be the unitary that permutes $n$ qubits according to $\pi$:
\[P_\pi\ket{\psi_1}...\ket{\psi_n} = \ket{\psi_{\pi(1)}}...\ket{\psi_{\pi(n)}}.\]

We use $[k]$ for the set $\{1,2,\dots,k\}$. For a projector $\Pi$, we write $\overline{\Pi}$ for its complement $\Id - \Pi$. We use $\tau^{\reg{R}} := \Id/|R|$ for the fully mixed state on the register $R$.

Write $GL(n,F)$ for the general linear group of degree $n$ over a field $F$. We refer to the Galois field  of two elements as $\mathbb{F}_2$, the $n$-qubit Pauli group as $\Pauli_n$, and the $n$-qubit Clifford group as $\Clifford_n$. Whenever a protocol mandates handing an element from one of these groups, or more generally, a unitary operation, to an agent, we mean that a (classical) description of the group element is given, e.g. as a normal-form circuit.

Finally, for a quantum operation that may take multiple rounds of inputs and outputs, for example an environment $\envE$ interacting with a protocol $\Pi$, we write $\envE \leftrightarrows \Pi$ for the final output of $\envE$ after the entire interaction.

\subsection{Classical multi-party computation}\label{sec:prelim-classical-mpc}
For multi-party computations where possibly more than half of the players are corrupted by the adversary, it is well known that one cannot achieve \emph{fairness} which asks that either all parties receive the protocol output or nobody does. Cleve has shown~\cite{Cleve1986} that in the case of dishonest majority there cannot exist MPC protocols that provide fairness and guaranteed output delivery. In this setting, we cannot prevent a dishonest player from simply aborting the protocol at any point, for example after having learned an unfavorable outcome of the protocol, before the honest player(s) have obtained their output(s). Hence, we have to settle for protocols allowing abort.

Over the last years, the most efficient protocols for classical multi-party computation with abort are in the so-called \emph{pre-processing model}, as introduced by the SPDZ-family of protocols \cite{Bendlin2011,Damgaard2012,Keller2018,Cramer2018}\footnote{
	We refer to~\cite[Section~5]{pqcrypto18} for a recent overview of other models of active corruptions that tolerate a dishonest majority of players, such as identifiable abort, covert security and public auditability.}. These protocols consist of two phases: The first ``offline'' phase is executed independently of the inputs to the actual MPC and produces authenticated shared randomness among the players, for example in the form of authenticated multiplication triples~\cite{Beaver1992}. These triples are used in the second phase to run a very efficient secure-function-evaluation protocol which is \emph{UC information-theoretically} secure against active corruptions of up to $k-1$ players, see, e.g.,~\cite[Section~8.5]{CramerDN2015}.

At this point, we are unaware of any formal analysis of the post-quantum security of these schemes. However, it should follow directly from Unruh's lifting theorem~\cite{Unruh2010} (asserting that classically secure protocols remain (statistically) secure in the quantum world) that the UC security of the second (online) phase can be lifted to the post-quantum setting. As for the pre-processing phase, there are two main types of protocols:
\begin{enumerate}
	\item The SPDZ-family make use of the homomorphic properties of computationally-secure public-key encryption systems to generate the authenticated multiplication triples. The scheme in~\cite{Damgaard2012} uses somewhat homomorphic encryption which is built from lattice assumptions and might already be post-quantum secure. In addition, the players provide non-interactive zero-knowledge proofs that they have performed these operations correctly. Those proofs are typically not post-quantum secure, but should be replaced by post-quantum secure variants like lattice-based zk-SNARGs~\cite{Gennaro2018} or zk-STARKs~\cite{Ben-Sasson2018}.
	\item The authors of MASCOT (Multi-party Arithmetic Secure Computation with Oblivious Transfer)~\cite{Keller2016} suggest a way to avoid the use of expensive public-key cryptography altogether and propose to use oblivious transfer (OT) and consistency checks to generate authenticated multiplication triples. A large number of OTs can be obtained from a few base OTs by OT extension-techniques such as~\cite{Keller2015}. Those techniques have very recently be proven secure in the QROM as well~\cite{BDK+20}. A post-quantum secure base OT can be obtained from lattices~\cite{Peikert2008}. Recent even more efficient MPC schemes~\cite{Cramer2018} have followed this OT-based approach as well.
\end{enumerate}

Establishing full post-quantum security of classical multi-party computation is outside the scope of this paper. For the purpose of this paper, we assume that such a post-quantum secure classical multi-party computation is available. According to the discussion above, it suffices that a preprocessing phase has been successfully run in order to have UC information-theoretically secure function evaluation (SFE). Unlike for general adversary structures~\cite{Hirt2008}, in our case of threshold adversaries, one can obtain (reactive) MPC from SFE by outputting shares of the overall state to the players and by asking the players to input them again in the next phase, see~\cite[Section~5.3.1]{CramerDN2015}.

Throughout this paper, we will utilize the following ideal MPC functionality as a black box:
\begin{definition}[Ideal classical $k$-party stateful computation with abort] \label{def:classicalMPC}
	Let $f_1, ..., f_k$ and $f_S$ be public classical deterministic functions on $k+2$ inputs. Let a string $s$ represent the internal state of the ideal functionality. (The first time the ideal functionality is called, $s$ is empty.) Let $A \subsetneq [k]$ be a set of corrupted players.
	\begin{enumerate}
		\item Every player $i \in [k]$ chooses an input $x_i$ of appropriate size, and sends it (securely) to the trusted third party.
		\item The trusted third party samples a bit string $r$ uniformly at random.
		\item The trusted third party computes $f_i(s,x_1, ..., x_k,r)$ for all $i \in [k] \cup \{S\}$.
		\item For all $i \in A$, the trusted third party sends $f_i(s,x_1, ..., x_k,r)$ to player $i$.
		\item All $i \in A$ respond with a bit $b_i$, which is 1 if they choose to abort, or 0 otherwise.
		\item If $b_j = 0$ for all $j$, the trusted third party sends $f_i(s,x_1, ..., x_k,r)$ to the other players $i \in [k]\backslash A$ and stores $f_S(s,x_1, ..., x_k,r)$ in an internal state register (replacing $s$). Otherwise, he sends an \texttt{abort} message to those players.
	\end{enumerate}
\end{definition}

\subsection{Pauli filter}\label{sec:pauli-filter}
In our protocol, we use a technique which alters a channel that would act
jointly on registers $S$ and $T$, so that its actions on $S$ are replaced by
a flag bit into a separate register. The flag is set to 0 if the actions
on $S$ belong to some set $\pauliset$, or to 1 otherwise. This way, the new
channel ``filters" the allowed actions on $S$.
\begin{definition}[Pauli filter]\label{def:pauli-filter}
	For registers $S$ and $T$ with $|T| > 0$, let $U^{\reg{ST}}$ be a unitary, and let $\pauliset \subseteq \left(\{0,1\}^{\log |S|}\right)^2$ contain pairs of bit strings. The $\pauliset$-filter of $U$ on register $S$, denoted $\paulifilter{\pauliset}{S}(U)$, is the map $T \to TF$ (where $F$ is some single-qubit flag register) that results from the following operations:
	\begin{enumerate}
		\item Initialize two separate registers $S$ and $S'$ in the state $\kb{\Phi}$, where $\ket{\Phi} := \left(\frac{1}{\sqrt{2}}(\ket{00} + \ket{11})\right)^{\otimes \log|S|}$. Half of each pair is stored in $S$, the other in $S'$.
		\item Run $U$ on $ST$.
		\item Measure $SS'$ with the projective measurement $\{\Pi, \Id - \Pi\}$ for 
		\begin{align} \Pi := \sum_{(a,b) \in \pauliset}
                  \left(\X^a\Z^b\right)^{\reg{S}}\kb{\Phi}\left(\Z^b\X^a\right). \label{eq:pauliprojector}
                  \end{align}
		If the outcome is $\Pi$, set the $F$ register to $\kb{0}$. Otherwise, set it to $\kb{1}$.
	\end{enumerate}
\end{definition}

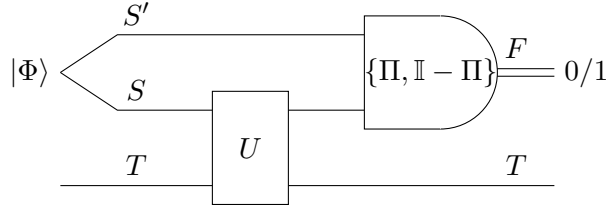
\begin{figure}
	\begin{center}
		\begin{tikzpicture}
		\draw (0,0) -- (6.5,0);
		\draw (0,1.5) -- (0.75,1) -- (5,1);
		\draw (0,1.5) -- (0.75,2) -- (5,2);
		\draw (5,1.45) -- (6.5,1.45);
		\draw (5,1.55) -- (6.5,1.55);
		\node[anchor=east] at (0,1.5) {$\ket{\Phi}$};
		\node[anchor=south] at (1,2) {$S'$};
		\node[anchor=south] at (1,1) {$S$};
		\node[anchor=south] at (1,0) {$T$};
		\draw[fill=white] (2,-0.25) rectangle (3,1.25);
		\node at (2.5,0.5) {$U$};
		\draw[fill=white] (5,0.75) arc (-90:90:0.75);
		\fill[white] (4,0.75) rectangle (5,2.25);
		\draw (5,0.75) -- (4,0.75) -- (4,2.25) -- (5,2.25);
		\node at (4.875,1.5) {$\{\Pi, \Id - \Pi\}$};
		\node[anchor=west] at (6.5,1.5) {$0/1$};
		\node[anchor=south] at (6,1.55) {$F$};
		\node[anchor=south] at (6,0) {$T$};
		\end{tikzpicture}
	\end{center}
	\caption{The circuit for $\paulifilter{\pauliset}{S}(U)$ from \Cref{def:pauli-filter}. By the definition of $\Pi$ in~\Cref{eq:pauliprojector}, the choice of the set $\pauliset$ determines which types of Paulis are filtered out by the measurement $\{\Pi, \Id - \Pi\}$. The $\idfilter{}$ and $\xfilter{}$ are special cases of this filter. Replacing $\ket{\Phi}$ with a different initial state yields a wider array of filters, e.g., the $\zerofilter{}$.} \label{fig:paulifilter}
\end{figure}

\Cref{fig:paulifilter} depicts the circuit for the Pauli filter. Its functionality becomes clear in the following lemma, which we prove in~\Cref{ap:pauli-filter} by straightforward calculation:

\begin{lemma}\label{lem:pauli-filter}
	For registers $S$ and $T$ with $|T| > 0$, let $U^{\reg{ST}}$ be a unitary, and let $\pauliset \subseteq \left(\{0,1\}^{\log |S|}\right)^2$. Write $U = \sum_{x,z} (\X^x\Z^z)^{\reg{S}} \otimes U^{\reg{T}}_{x,z}$. Then running $\paulifilter{\pauliset}{S}(U)$ on register $T$ equals the map $T \to TF$:
	\[
	(\cdot)^{\reg{T}} \mapsto \sum_{(a,b) \in \pauliset} U^{\reg{T}}_{a,b}(\cdot)U^{\dagger}_{a,b} \otimes \kb{0}^{\reg{F}} + \sum_{(a,b) \not\in \pauliset} U^{\reg{T}}_{a,b}(\cdot)U^{\dagger}_{a,b} \otimes \kb{1}^{\reg{F}}
	\]
\end{lemma}
A special case of the Pauli filter for $\pauliset = \{(0^{\log|S|},0^{\log|S|})\}$ is due to Broadbent and Wainewright \cite{BW16}. This choice of $\pauliset$ represents only identity: the operation $\paulifilter{\pauliset}{}$ filters out any components of $U$ that do not act as identity on $S$. We will denote this type of filter with the name $\idfilter{}$.

In this work, we will also use $\xfilter{S}(U)$, which only accepts components of $U$ that act trivially on register $S$ in the computational basis. It is defined by choosing $\pauliset =  \{0^{\log|S|}\} \times \{0,1\}^{\log|S|}$.

Finally, we note that the functionality of the Pauli filter given in~\Cref{def:pauli-filter} can be generalized, or weakened in a sense, by choosing a different state than $\kb{\Phi}$. In this work, we will use the $\zerofilter{S}(U)$, which initializes $SS'$ in the state $\ket{00}^{\log|S|}$, and measures using the projector $\Pi = \kb{00}$. It filters $U$ by allowing only those Pauli operations that leave the computational-zero state (but not necessarily any other computational-basis states) unaltered:
\begin{align*}
(\cdot) \mapsto U_{0}^{\reg{T}} (\cdot)U_{0}^{\dagger} \otimes \kb{0}^{\reg{F}} +  \sum_{a \neq 0}U_a^{\reg{T}} (\cdot)U_a^{\dagger} \otimes \kb{1}^{\reg{F}},
\end{align*}
where we abbreviate $U_a := \sum_b U_{a,b}$. Note that for $\zerofilter{S}(U)$, the extra register $S'$ can also be left out (see \Cref{fig:zeroifilter}). %

\begin{figure}
	\centering
	\resizebox{\textwidth}{!}{
		\begin{tikzpicture}
		\node at (0,0) {
			\begin{tikzpicture}
			\draw (-2,3) rectangle (4.5,-.5);
			\draw (-2.5,0) -- (5,0);
			\draw (0,1.5) -- (0.25,1) -- (3,1);
			\draw (0,1.5) -- (0.25,2) -- (3,2);
			\draw (3.5,1.45) -- (5,1.45);
			\draw (3.5,1.55) -- (5,1.55);
			\node[anchor=east] at (0,1.5) {$\ket{00}^{\otimes \log|S|}$};
			\node[anchor=south] at (0.5,2) {$S'$};
			\node[anchor=south] at (0.5,1) {$S$};
			\node[anchor=south] at (-2.25,0) {$T$};
			\draw[fill=white] (1,-0.25) rectangle (1.75,1.25);
			\node at (1.375,0.5) {$U$};
			\draw[fill=white] (3,0.75) arc (-90:90:0.75);
			\fill[white] (2,0.75) rectangle (3,2.25);
			\draw (3,0.75) -- (2,0.75) -- (2,2.25) -- (3,2.25);
			\node at (2.875,1.5) {$\{\Pi, \Id - \Pi\}$};
			\node[anchor=south] at (4.75,1.55) {$F$};
			\node[anchor=south] at (4.75,0) {$T$};
			\end{tikzpicture}};
		\node at (4.5,0) {$=$};
		\node at (8,-.5) {
			\begin{tikzpicture}
			\draw (-2,2) rectangle (4,-.5);
			\draw (-2.5,0) -- (4.5,0);
			\draw (0,1) -- (3,1);
			\draw (3,0.95) -- (4.5,0.95);
			\draw (3,1.05) -- (4.5,1.05);
			\node[anchor=east] at (0,1) {$\ket{0}^{\otimes \log|S|}$};
			\node[anchor=south] at (0.25,1) {$S$};
			\node[anchor=south] at (-2.25,0) {$T$};
			\draw[fill=white] (0.5,-0.25) rectangle (1.25,1.25);
			\node at (0.875,0.5) {$U$};
			\draw[fill=white] (3,0.5) arc (-90:90:0.5);
			\fill[white] (1.5,0.5) rectangle (3,1.5);
			\draw (3,0.5) -- (1.5,0.5) -- (1.5,1.5) -- (3,1.5);
			\node at (2.5,1) {$\{\Pi', \Id - \Pi'\}$};
			\node[anchor=south] at (4.25,1.05) {$F$};
			\node[anchor=south] at (4.25,0) {$T$};
			\end{tikzpicture}};
	\end{tikzpicture}
}
\caption{The Pauli filter $\zerofilter{}$, where the initial state $\ket{\Phi}$ is replaced by $\ket{00}^{\log|S|}$. This filter, with measurement $\Pi := \kb{00}^{\otimes \log |S|}$, is equivalent to the map which does not prepare the register $S'$, and measures with $\Pi' := \kb{0}^{\otimes \log |S|}$.} \label{fig:zeroifilter}
\end{figure}
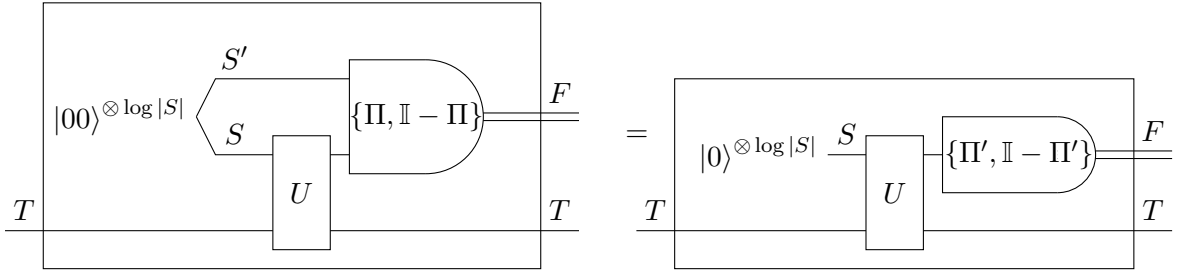

\subsection{Clifford authentication code}
The protocol presented in this paper will rely on quantum authentication. The players will encode their inputs using a quantum authentication code to prevent the other, potentially adversarial, players from making unauthorized alterations to their data. That way, they can ensure that the output of the computation is in the correct logical state.

A quantum authentication code transforms a quantum state (the \emph{logical} state or \emph{plaintext}) into a larger quantum state (the \emph{physical} state or \emph{ciphertext}) in a way that depends on a secret key. An adversarial party that has access to the ciphertext, but does not know the secret key, cannot alter the logical state without being detected at decoding time.

More formally, an authentication code consists of an encoding map $\Enc_{k}^{\reg{M \to MT}}$ and a decoding map $\Dec_{k}^{\reg{MT \to M}}$, for a secret key $k$, which we usually assume that the key is drawn uniformly at random from some key set $\KeySet$. The message register $M$ is expanded with an extra register $T$ to accommodate for the fact that the ciphertext requires more space than the plaintext.

An authentication code is correct if $\Dec_k \circ \Enc_k = \Id$. It is secure if the decoding map rejects (e.g., by replacing the output with a fixed reject symbol $\bot$) whenever an attacker tried to alter an encoded state:

\begin{definition}[Security of authentication codes~\cite{DNS10}]\label{def:auth-security} Let $(\Enc_k^{\reg{M \to MT}}\!\!\!\!,$ $\Dec_k^{\reg{MT \to M}})$ be a quantum authentication scheme for $k$ in a key set $\KeySet$. The scheme is $\eps$-secure if for all CPTP maps $\advA^{\reg{MTR}}$ acting on the ciphertext and a side-information register $R$, there exist CP maps $\Attack_{\acc}$ and $\Attack_{\rej}$ such that $\Attack_{\acc} + \Attack_{\rej}$ is trace-preserving, and for all $\rho^{\reg{MR}}$:

	\begin{align*}
		\left\|
		\E_{k \in \KeySet} \left[\Dec_k \left( \advA \left( \Enc_k \left(\rho \right)\right)\right)\right] \ \ - \ \ \left(\Attack_{\acc}^{\reg{R}}(\rho) + \kb{\bot}^{\reg{M}} \otimes \Tr_M\left[\Attack^{\reg{R}}_{\rej}\left(\rho\right)\right]\right)
		\right\|_1 \leq \eps.
	\end{align*}
\end{definition}

A fairly simple but powerful authentication code is the Clifford code:

\begin{definition}
	[Clifford code~\cite{Aharonov2010}]\label{def:clifford-code}
	The $n$-qubit Clifford code is defined by a key set $\Clifford_{n+1}$, and the encoding and decoding maps for a $C \in \Clifford_{n+1}$:
	\begin{align*}
	\Enc_C(\rho^{\reg{M}}) &:= C(\rho^{\reg{M}} \otimes \kb{0^n}^{\reg{T}})C^{\dagger},\\
	\Dec_C(\sigma^{\reg{MT}}) &:= \bra{0^n}^{\reg{T}}C^{\dagger}\sigma C\ket{0^n} + \kb{\bot}^{\reg{M}} \otimes \Tr_{M}\left[ \sum_{x \neq 0^n} \bra{x}C^{\dagger}\sigma C\ket{x}\right].
	\end{align*}
\end{definition}
Note that, from the point of view of someone who does not know the Clifford key $C$, the encoding of the Clifford code looks like a Clifford twirl (see \Cref{ap:twirling}) of the input state plus some trap states.

\medskip

We prove the security of the Clifford code
in~\Cref{sec:proof-clifford-code}.

\subsection{Universal gate sets}
\label{sec:t-magic-states}
It is well known that if, in addition to Clifford gates, we are able to apply
{\em any} non-Clifford gate $G$, then we are able to achieve universal quantum
computation. In this work, we focus on the non-Clifford \T gate (or $\pi/8$ gate).

In several contexts, however, applying non-Clifford gates is not
straightforward for different reasons: common quantum error-correcting
codes do not allow transversal implementation of
non-Clifford gates, the non-Clifford gates do not commute with the quantum
one-time pad and, more importantly in this work, neither with the Clifford encoding.

In order to concentrate the hardness of non-Clifford gates in an offline
pre-processing phase, we can use techniques from computation by teleportation if
we have so-called \emph{magic states} of the form $\ket{\T} := \T\ket{+}$. Using a single copy of this state as a resource, we are able to implement a \T gate using the circuit in Figure~\ref{fig:magic-state}. The circuit only requires (classically controlled) Clifford gates.

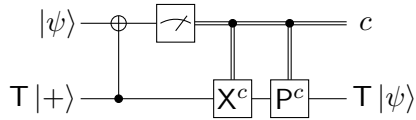
\begin{figure}[h]
\centering
\begin{tikzpicture}
\node at (-.3,1) {$\ket{\psi}$};
\node at (-.5,0) {$\mathsf{T}\ket{+}$};

\draw(0,1)--(1,1); 	\draw (1.5,1.02)--(3.5,1.02);
				\draw (1.5,.98)--(3.5,.98);
\draw(0,0)--(3.5,0);

\draw (.5,0)--(.5,1.1);
\draw (.5,1) circle (.1);
\filldraw (.5,0) circle (.05);

\filldraw[fill=white] (1,.75) rectangle (1.5,1.25);
\draw (1.05,1) arc (120:60:.4);
\draw (1.25,.9) -- (1.4,1.1);

\filldraw (2,1) circle (.05);
\draw (1.98,.25)--(1.98,1);
\draw (2.02,.25)--(2.02,1);
\filldraw[fill=white] (1.75,-.25) rectangle (2.25,.25);
\node at (2,0) {$\mathsf{X}^c$};

\filldraw (2.75,1) circle (.05);
\draw (2.73,.25)--(2.73,1);
\draw (2.77,.25)--(2.77,1);
\filldraw[fill=white] (2.5,-.25) rectangle (3,.25);
\node at (2.75,0) {$\mathsf{P}^c$};

\node at (3.75,1) {$c$};
\node at (4,0) {$\mathsf{T}\ket{\psi}$};
\end{tikzpicture}
\caption{Using a magic state $\ket{\sf T}=\mathsf{T}\ket{+}$ to implement a $\mathsf{T}$ gate.}\label{fig:magic-state}
\end{figure}

The problem is how
to create such magic states in a fault-tolerant way. Bravyi and Kitaev~\cite{BK04} proposed a
distillation protocol that allows to create states that are $\delta$-close to
true magic states, given $\mathrm{poly}(\log(1/{\delta}))$ copies of {\em noisy} magic-states.  
Let $\ket{\mathsf{T}^\bot}=\mathsf{T}\ket{-}$. Then we have:

\begin{theorem}[Magic-state distillation \cite{BK04}]\label{thm:distillation}
There exists a circuit of CNOT-depth $\ddistill(n) \leq O(\log(n))$ consisting of $\pdistill(n) \leq \poly{n}$ many classically controlled Cliffords and
  computational-basis measurements such that for any $\eps <
  \frac{1}{2}\left(1-\sqrt{3/7}\right)$, if $\rho$ is the output on the first
  wire using input
  \begin{align}\label{eq:target-distillation}
    \left( (1-\eps)\ket{\mathsf{T}}\bra{\mathsf{T}}+\eps\ket{\mathsf{T}^\bot}\bra{\mathsf{T}^\bot}\right)^{\otimes n},
  \end{align}
then
${1-\bra{\mathsf{T}}\rho\ket{\mathsf{T}}}\leq O\left( (5\eps)^{n^c}\right)$,
where $c=(\log_2 30)^{-1}\approx 0.2$.
\end{theorem}

As we will see in Section~\ref{sec:T-gates}, our starting point is a bit different from the input state
required by \Cref{thm:distillation}. 
We now present a procedure that will allow us to prepare the states necessary
for  applying Theorem~\ref{thm:distillation} (see Circuit~\ref{protocol:distillation}).
We prove~\Cref{lem:distillation-works} in~\Cref{sec:proof-distillation}.

\begin{lemma}\label{lem:distillation-works}
  Let $V_{LW}=\mathrm{span}\{P_\pi(\ket{\T}^{\otimes
  m-w}\ket{\T^\perp}^{w}):w \leq \ell, \pi \in S_m\}$, and let
  $\Pi_{LW}$ be the orthogonal projector onto $V_{LW}$. Let $\Xi$ denote the
  CPTP map induced by Circuit~\ref{protocol:distillation}. If $\rho$ is an $m$-qubit state such that $\Tr(\Pi_{LW}\rho)\geq 1-\eps$, then
  $$\norm{\Xi(\rho)-(\ket{\T}\bra{\T})^{\otimes t}}_1\leq
  O\left(m\sqrt{t}\left(\frac{\ell}{m}\right)^{O((m/t)^c/2)}+\eps\right),$$
for some constant $c>0$.
\end{lemma}

\begin{circuit}
	[Magic-state distillation]
	\label{protocol:distillation}\ Given an $m$-qubit input state and a parameter $t < m$:
\begin{enumerate}
\item To each qubit, apply $\hat{\mathsf{Z}}:=\mathsf{PX}$ with probability $\frac{1}{2}$.
\item Permute the qubits by a random $\pi\in S_m$.
\item Divide the $m$ qubits into $t$ blocks of size $m/t$, and apply magic-state distillation from Theorem~\ref{thm:distillation} to each block.
\end{enumerate}
\end{circuit}

\begin{remark}
  Circuit \ref{protocol:distillation} can be implemented with (classically
  controlled) Clifford gates and measurements in the computational basis.
\end{remark}

\section{Multi-party Quantum Computation: Definitions}\label{sec:defin}

In this section, we describe the ideal functionality we aim to achieve for multi-party quantum computation (MPQC) with a dishonest majority. As noted in Section~\ref{sec:prelim-classical-mpc}, we cannot hope to achieve fairness: therefore, we consider an ideal functionality with the option for the dishonest players to abort.

\begin{definition}[Ideal quantum $k$-party computation with abort]\label{def:ideal-MPQC}
	Let $C$ be a quantum circuit on $W \in \mathbb{N}_{>0}$ wires. Consider a partition of the wires into the players' input registers plus an ancillary register, as $[W] = R_1^{\inreg} \sqcup \cdots \sqcup R_k^{\inreg} \sqcup R^{\ancillareg}$, and a partition into the players' output registers plus a register that is discarded at the end of the computation, as $[W] = R_1^{\outreg} \sqcup \cdots \sqcup R_k^{\outreg} \sqcup R^{\discardreg}$. Let $I_A \subsetneq [k]$ be a set of corrupted players.
	\begin{enumerate}
		\item Every player $i \in [k]$ sends the content of $R_i^{\inreg}$ to the trusted third party.
		\item The trusted third party populates $R^{\ancillareg}$ with computational-zero states.
		\item The trusted third party applies the quantum circuit $C$ on the wires $[W]$.
		\item For all $i \in I_A$, the trusted third party sends the content of $R_i^{\outreg}$ to player $i$.
		\item All $i \in I_A$ respond with a bit $b_i$, which is 1 if they choose to abort, or 0 otherwise.
		\item If $b_i = 0$ for all $i$, the trusted third party sends the content of $R_i^{\outreg}$ to the other players $i \in [k]\backslash I_A$. Otherwise, he sends an \texttt{abort} message to those players.
	\end{enumerate}
\end{definition}

In Definition~\ref{def:ideal-MPQC}, all corrupted players individually choose whether to abort the protocol (and thereby to prevent the honest players from receiving their respective outputs). In reality, however, one cannot prevent several corrupted players from actively working together and sharing all information they have among each other. To ensure that our protocol is also secure in those scenarios, we consider security against a general adversary that corrupts all players in $I_{\advA}$, by replacing their protocols by a single (interactive) algorithm $\advA$ that receives the registers $R^{\inreg}_{\advA} := R \sqcup \bigsqcup_{i \in I_{\advA}} R^{\inreg}_i$ as input, and after the protocol produces output in the register $R^{\outreg}_{\advA} := R \sqcup \bigsqcup_{i \in I_{\advA}} R^{\outreg}_i$. Here, $R$ is a side-information register in which the adversary may output extra information.

We will always consider protocols that fulfill the ideal functionality with respect to some gate set $\mathcal{G}$: the protocol should then mimic the ideal functionality only for circuits $C$ that consist of gates from $\mathcal{G}$. This security is captured by the definition below.

\begin{figure}
\centering
\begin{tikzpicture}
\node at (0,0) {
	\begin{tikzpicture}
	\node at (2,3.25) {$C$};
	\draw[->] (2,3)--(2,2.75);

	\draw (0,2.25) rectangle (4,2.75);
	\node at (2,2.5) {$\Pi$};

	\draw[->] (.1,1.75)--(.1,2.25);						\draw[->] (1.6,1.75)--(1.6,2.25);		\draw[->] (2.75,1.75)--(2.75,2.25);
	\draw[<-] (.4,1.75)--(.4,2.25);						\draw[<-] (1.9,1.75)--(1.9,2.25);		\draw[<-] (3.5,1.75)--(3.5,2.25);

	\draw (0,1.25) rectangle (.5,1.75);					\draw (1.5,1.25) rectangle (2,1.75);	\draw (2.25,1.25) rectangle (4,1.75);
	\node at (.25,1.5) {$P_1$};		\node at (1,1.5) {$\dots$};	\node at (1.75,1.5) {$P_{\ell}$};		\node at (3.125,1.5) {${\cal A}$};

	\draw[->] (.1,.75)--(.1,1.25);						\draw[->] (1.6,.75)--(1.6,1.25);		\draw[->] (2.75,.75)--(2.75,1.25);
	\draw[<-] (.4,.75)--(.4,1.25);						\draw[<-] (1.9,.75)--(1.9,1.25);		\draw[<-] (3.5,.75)--(3.5,1.25);

	\draw (0,.25) rectangle (4,.75);
	\node at (2,.5) {$\cal E$};

	\draw[dashed] (-.25,1)--(-.25,3.5)--(4.25,3.5)--(4.25,1)--(-.25,1);
	\node at (.4,3.15) {$\Pi^{\sf MPQC}_{C,{\cal A}}$};
	\end{tikzpicture}
};

\node at (6,0) {
	\begin{tikzpicture}
	\node at (2,3.25) {$C$};
	\draw[->] (2,3)--(2,2.75);

	\draw (0,2.25) rectangle (4,2.75);
	\node at (2,2.5) {$\mathfrak{I}$};

	\draw[->] (.1,.75)--(.1,2.25);	\draw[->] (1.6,.75)--(1.6,2.25);		\draw[->] (2.75,1.75)--(2.75,2.25);
	\draw[<-] (.4,.75)--(.4,2.25);	\draw[<-] (1.9,.75)--(1.9,2.25);		\draw[<-] (3.5,1.75)--(3.5,2.25);

														\draw (2.25,1.25) rectangle (4,1.75);
					\node at (1,1.5) {$\dots$};					\node at (3.125,1.5) {${\cal S}$};

														\draw[->] (2.75,.75)--(2.75,1.25);
														\draw[<-] (3.5,.75)--(3.5,1.25);

	\draw (0,.25) rectangle (4,.75);
	\node at (2,.5) {$\cal E$};

	\draw[dashed] (-.25,1)--(-.25,3.5)--(4.25,3.5)--(4.25,1)--(-.25,1);
	\node at (.4,3.15) {$\mathfrak{I}^{\sf MPQC}_{C,{\cal S}}$};
	\end{tikzpicture}
};

\node at (0,-2) {(1)};%

\node at (6,-2) {(2)};%

\end{tikzpicture}
\caption{(1) The environment interacting with the protocol as run by honest players $P_1,\dots,P_\ell$, and an adversary who has corrupted the remaining players. (2) The environment interacting with a simulator running the ideal functionality. }\label{fig:mpc-security}
\end{figure}
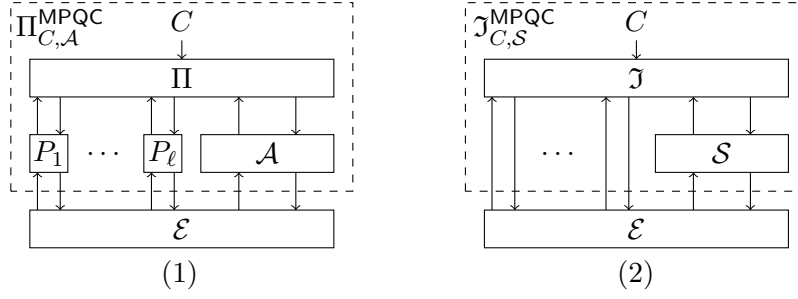

\begin{definition}[Computational security of quantum $k$-party computation with abort]
	Let $\mathcal{G}$ be a set of quantum gates. Let $\RealProtocol^{\MPQC}$ be a $k$-party quantum computation protocol, parameterized by a security parameter $n$. For any circuit $C$, set $I_{\advA} \subsetneq [k]$ of corrupted players, and adversarial (interactive) algorithm $\advA$ that performs all interactions of the players in $I_{\advA}$, define $\RealProtocol_{C, A}^{\MPQC} : R^{\inreg}_{\advA} \sqcup \bigsqcup_{i \not\in I_{\advA}} R^{\inreg}_i \to R^{\outreg}_{\advA} \sqcup \bigsqcup_{i \not\in I_{\advA}} R^{\outreg}_i$ to be the channel that executes the protocol $\RealProtocol^{\MPQC}$ for circuit $C$ by executing the honest interactions of the players in $[k] \setminus I_{\advA}$, and letting $\advA$ fulfill the role of the players in $I_{\advA}$ (See Figure~\ref{fig:mpc-security}, (1)).

	For a simulator $\simS$ that receives inputs in $R^{\inreg}_{\advA}$, then interacts with the ideal functionalities on all interfaces for players in $I_{\advA}$, and then produces output in $R^{\outreg}_{\advA}$, let $\IdealProtocol_{C,\simS}^{\MPQC}$ be the ideal functionality described in Definition~\ref{def:ideal-MPQC}, for circuit $C$, simulator $\simS$ for players $i \in I_{\advA}$, and honest executions (with $b_i = 0$) for players $i \not\in I_{\advA}$ (See Figure~\ref{fig:mpc-security}, (2)).
	We say that $\RealProtocol^{\MPQC}$ is a \emph{computationally $\eps$-secure quantum $k$-party computation protocol with abort}, if for all $I_{\advA} \subsetneq [k]$, for all quantum polynomial-time (QPT) adversaries $\advA$, and all circuits $C$ comprised of gates from $\mathcal{G}$, there exists a QPT simulator $\mathcal{S}$ such that for all QPT environments $\envE$,
	\[
	\left|
	\Pr\left[1 \leftarrow (\envE \leftrightarrows \RealProtocol_{C,\advA}^{\MPQC})\right]
	-
	\Pr\left[1 \leftarrow (\envE \leftrightarrows \IdealProtocol^{\MPQC}_{C,\simS})\right]
	\right| \leq \eps.
	\]
	Here, the notation $b \leftarrow (\envE \leftrightarrows (\cdot))$ represents the environment $\envE$, on input $1^n$, interacting with the (real or ideal) functionality $(\cdot)$, and producing a single bit $b$ as output.
\end{definition}

\begin{remark}
In the above definition, we assume that all QPT parties are polynomial in
the size of circuit $|C|$, and in the security parameter $n$.
\end{remark}

We show in \Cref{sec:universal-mpqc} the protocol $\RealProtocol^{\MPQC}$ implementing the ideal
functionality described in~\Cref{def:ideal-MPQC}, and we prove its security in
\Cref{lem:protocol-mpqc}.

\section{Setup and encoding}\label{sec:setup-encoding}

\subsection{Input encoding}
In the first phase of the protocol, all players encode their input registers qubit-by-qubit. For simplicity of presentation, we pretend that player 1 holds a single-qubit input state, and the other players do not have input. In the actual protocol, multiple players can hold multiple-qubit inputs: in that case, the initialization is run several times in parallel, using independent randomness. Any other player $i$ can trivially take on the role of player 1 by relabeling the player indices.
\begin{definition}[Ideal functionality for input encoding] \label{def:input-encoding} Without loss of generality, let $R_1^{\inreg}$ be a single-qubit input register, and let $\dim(R_i^{\inreg}) = 0$ for all $i \neq 1$. Let $I_{\advA} \subsetneq [k]$ be a set of corrupted players.
	\begin{enumerate}
		\item Player 1 sends register $R_1^{\inreg}$ to the trusted third party.
		\item The trusted third party initializes a register $T_1$ with $\kb{0^n}$,
      applies a random $(n+1)$-qubit Clifford $E$ to $MT_1$, and sends these
      registers to player~1.
		\item All players $i \in I_{\advA}$ send a bit $b_i$ to the trusted third party. If $b_i = 0$ for all $i$, then the trusted third party stores the key $E$ in the state register $S$ of the ideal functionality. Otherwise, it aborts by storing $\bot$ in $S$.
	\end{enumerate}
\end{definition}

The following protocol implements the ideal functionality. It uses, as a black box, an ideal functionality \MPC that implements a classical multi-party computation with memory.

\begin{protocol}(Input encoding) \label{protocol:input-encoding}Without loss of generality, let $M := R_1^{\inreg}$ be a single-qubit input register, and let $\dim(R_i^{\inreg}) = 0$ for all $i \neq 1$.
	\begin{enumerate}
		\item For every $i \in [k]$, \MPC samples a random $(2n+1)$-qubit Clifford $F_i$ and tells it to player~$i$.\label{step:encoding-first}
		\item Player 1 applies the map $\rho^{\reg{M}} \mapsto F_1\left(\rho^{\reg{M}} \otimes \kb{0^{2n}}^{\reg{T_1T_2}}\right)F_1^{\dagger}$ for two $n$-qubit (trap) registers $T_1$ and $T_2$, and sends the registers $MT_1T_2$ to player 2.
		\item Every player $i = 2, 3, ..., k$ applies $F_i$ to $MT_1T_2$, and forwards it to player $i+1$. Eventually, player $k$ sends the registers back to player 1.
		\item \MPC samples a random $(n+1)$-qubit Clifford $E$, random $n$-bit strings $r$ and $s$, and a random classical invertible linear operator $g \in GL(2n,\mathbb{F}_2)$. Let $U_g$ be the (Clifford) unitary that computes $g$ in-place, i.e., $U_g\ket{t} = \ket{g(t)}$ for all $t \in \{0,1\}^{2n}$.\label{step:encoding-Ug}
		\item \label{step:encoding-apply-T}\MPC gives\footnote{As described in~\Cref{sec:notation}, the MPC gives
			$V$ as a group element, and the adversary
			cannot decompose it into the different parts that appear in its
			definition.}
		\[
		V := (E^{\reg{MT_1}} \otimes (\X^r\Z^s)^{\reg{T_2}}) (\mathbb{I} \otimes (U_g)^{\reg{T_1T_2}})(F_k \cdots F_2 F_1)^{\dagger}
		\]
		to player 1, who applies it to $MT_1T_2$.\label{step:encoding-right-before-measurement}
		\item Player 1 measures $T_2$ in the computational basis, discarding the measured wires, and keeps the other $(n+1)$ qubits as its output in $R_1^{\outreg} = MT_1$.\label{step:encoding-measurement}
		\item Player 1 submits the measurement outcome $r'$ to \MPC, who checks whether $r = r'$. If so, \MPC stores the key $E$ in its memory-state register $S$. If not, it aborts by storing $\bot$ in $S$.\label{step:encoding-store-key}
	\end{enumerate}
\end{protocol}
If \MPC aborts the protocol in step~\ref{step:encoding-store-key}, the
information about the Clifford encoding key $E$ is erased. In that case, the
registers $MT_1$ will be fully mixed. Note that this result differs slightly
from the `reject' outcome of a quantum authentication code as in
Definition~\ref{def:auth-security}, where the message register $M$ is replaced by a dummy state $\kb{\bot}$. In our current setting, the register $M$ is in the hands of (the possibly malicious) player 1. We therefore cannot enforce the replacement of register $M$ with a dummy state: we can only make sure that all its information content is removed. Depending on the application or setting, the trusted \MPC can of course broadcast the fact that they aborted to all players, including the honest one(s).

To run Protocol~\ref{protocol:input-encoding} in parallel for multiple input qubits held by multiple players, \MPC samples a list of Cliffords $F_{i,q}$ for each player $i \in [k]$ and each qubit $q$. The $F_{i,q}$ operations can be applied in parallel for all qubits $q$: with $k$ rounds of communication, all qubits will have completed their round past all players.

We will show that Protocol~\ref{protocol:input-encoding} fulfills the ideal functionality for input encoding:

\begin{lemma}\label{lem:input-encoding}
	Let $\RealProtocol^{\Enc}$ be Protocol~\ref{protocol:input-encoding}, and $\IdealProtocol^{\Enc}$ be the ideal functionality described in Definition~\ref{def:input-encoding}. For all sets $I_{\advA} \subsetneq [k]$ of corrupted players and all adversaries $\advA$ that perform the interactions of players in $I_{\advA}$ with $\RealProtocol$, there exists a simulator $\simS$ (the complexity of which scales polynomially in that of the adversary) such that for all environments $\envE$,
	\[
	|\Pr[1 \leftarrow (\envE \leftrightarrows \RealProtocol^{\Enc}_{\advA})] - \Pr[1 \leftarrow (\envE \leftrightarrows \IdealProtocol^{\Enc}_{\simS})| \leq \negl{n}.
	\]
\end{lemma}
Note that the environment $\envE$ also receives the state register $S$, which acts as the ``output" register of the ideal functionality (in the simulated case) or of \MPC (in the real case). It is important that the environment cannot distinguish between the output states even given that state register $S$, because we want to be able to compose Protocol~\ref{protocol:CNOT} with other protocols that use the key information inside $S$. In other words, it is important that, unless the key is discarded, the states \emph{inside} the Clifford encoding are also indistinguishable for the environment.

We provide just a sketch of the proof
for~\Cref{lem:input-encoding}, and refer to~\Cref{ap:input-encoding} for its
full proof.

\begin{proof}[Proof sketch]
We divide our proof into two cases: when player 1 is honest, or when she is
  dishonest.

For the case when player $1$ is honest,  we know that she correctly prepares the
expected state before the state is given to the other players. That is, she appends $2n$ ancilla qubits in state $\ket{0}$ and applies the random
  Clifford instructed by the classical MPC.
  When the encoded state is returned to player 1, she performs the Clifford $V$ as instructed  by the MPC.
  By the properties of the Clifford encoding, if the other players acted dishonestly, the tested traps will be non-zero with probability exponentially
  close to $1$.

\smallskip
The second case is a bit more complicated: the first player has full control over
the state and, more importantly, the traps that will be used in the first
  encoding. In particular, she could start with nonzero traps, which could
  possibly give some advantage to the dishonest players later on the execution
  of the protocol.

  In order to prevent this type of attack,
  the MPC instructs the first player to apply a random linear
  function $U_g$ on the traps, which is hidden from the players inside the Clifford $V$.
  If the traps were initially zero, their value does not change,
  but otherwise, they will be mapped to a random value, unknown by the
  dishonest parties. As such, the map $U_g$ removes any advantage that the dishonest parties could have in
  step~\ref{step:encoding-store-key} by starting with non-zero traps. Because \emph{any} nonzero trap state in $T_1T_2$ is mapped to a random string, it suffices to measure only $T_2$ in order to be convinced that $T_1$ is also in the all-zero state (except with negligible probability). This intuition is formalized in Lemma~\ref{lem:GL(2n,F2)-twirl} in \Cref{ap:input-encoding}.
  
  Other
  possible attacks are dealt with in a way that is similar to the case where player 1 is honest (but
  from the perspective of another honest player).

  \medskip

  In the full proof (see~\Cref{ap:input-encoding}), we present two simulators,
  one for each case, that tests (using Pauli filters from Section~\ref{sec:pauli-filter}) whether the adversary performs any such
  attacks during the protocol, and chooses the input to the ideal functionality
  accordingly. See~\Cref{fig:enc} for a pictorial representation of the
  structure of the simulator for the case where player 1 is honest.
\end{proof}

\begin{figure}
	\resizebox{\textwidth}{!}{
		\centering
		\begin{tikzpicture}
		\node at (0,0){\begin{tikzpicture}
			\draw[line width=1pt] (-1.85,3) rectangle (5,3.5);
			\node at (1.5,3.25) {MPC};
			
			\draw[line width=1pt] (-1.85,2.5) rectangle (0,1);
			\node at (-.25,1.25) {$P_1$};
			\draw[->] (-.625,3)--(-.625,2.5);
			\node[anchor=east] at (-.545,2.75) {$E$};
			
			\draw[-] (-2.35,2.25)--(1,2.25);
			\node[anchor=south] at (-2.1,2.17) {$M$};
			\node[anchor=south] at (.5,2.17) {$M$};
			\draw[-] (-.96,1.75)--(1,1.75);
			\node[anchor=south] at (.5,1.67) {$T_1T_2$};
			\node[anchor=east] at (-.83,1.75) {$\ket{0^{2n}}$};
			
			\draw[fill=white] (-.75,2.375) rectangle (-.25,1.625);
			\node at (-.5,2) {$E$};

			\draw[line width=1pt] (1,2.5) rectangle (5,.5);
			
			\draw[-] (.5,.75)--(5.5,.75);
			\node[anchor=south] at (.75,.67) {$R$};
			\node[anchor=south] at (5.2,.67) {$R$};
			
			\draw[->] (5,2)--(5.75,2)--(5.75,.25)--(-.5,.25)--(-.5,1);
			
			\draw (1,1.75) -- (1.25,2);
			\draw (1,2.25) -- (1.25,2);
			\draw (1.25,2) -- (5,2);
			
			\draw[fill=white] (2.75,1.75) rectangle (4.75,2.25);
			\node at (3.75,2) {$F_k \cdots F_2$};
			
			\draw[fill=white] (1.5,2.25) rectangle (2.25,.625);
			\node at (1.875,1.5) {$A$};
			
			\draw[->] (2,3)--(2,2.5);
			\node[anchor=east] at (2.06,2.75) {$F_2$};
			\node at (2.75,2.75) {$\dots$};
			\draw[->] (4,3)--(4,2.5);
			\node[anchor=east] at (4.06,2.75) {$F_k$};
			\end{tikzpicture}};
		
		\node at (9,0){\begin{tikzpicture}
			
			\draw[line width=1pt] (-2.75,2.5) rectangle (4.8,2);
			\node at (1.25,2.25) {$\mathfrak{J}^{\mathsf{Enc}}$};
			
			\draw[line width=1pt] (-2.75,1.5) rectangle (-1,0);
			\node at (-1.25,.25) {$P_1$};
			\draw[->] (-3.3,1.25) -- (-2,1.25) -- (-2,2);
			\node[anchor=south] at (-3.05,1.17) {$M$};
			\node[anchor=east] at (-1.92,1.75) {$M$};
			\draw[<-] (-1.5,1.5)--(-1.5,2);
			\node[anchor=west] at (-1.58,1.75) {$MT_1$};
			
			\draw[dashed] (-.5,1.5) rectangle (4.8,-1);
			\node at (-.25,1.25) {$\cal S$};

			\draw[-] (-1,-.5)--(5.3,-.5);
			\node[anchor=south] at (-.75,-.58) {$R$};
			\node[anchor=south] at (5.05,-.58) {$R$};
			
			\draw (3,0) -- (4.8,0);
			\draw (3,-.1) -- (4.8,-.1);
			\node[anchor=south] at (3.25,-0.08) {$F$};
			
			\node at (1.5,0.6) {$F_2', \dots,F_k' \leftarrow \$$};
			
			\filldraw[fill=white] (0,-.75) rectangle (3,.25);
			\node at (1.5,-.25) {$\idfilter{MT_1T_2}(A)$};

			\filldraw (4,-.05) circle (.05);
			\draw[->] (4,0) -- (4,2);
			\node[anchor=west] at (3.92,1.75) {\small $b_k$};
			\draw[->] (4,1.25) -- (2.5,1.25) --(2.5,2);
			\node[anchor=west] at (2.42,1.75) {\small $b_2$};
			\node at (3.5,1.75) {$\cdots$};

			\end{tikzpicture}};

	\end{tikzpicture}
}
\caption{On the left, the adversary's interaction with the protocol $\Pi^{\mathsf{Enc}}$ in case player 1 is the only honest player. The $R$ register contains side information for the adversary. We may assume that the adversarial map consists of a unitary $A$ followed by the honest protocol $F_k \cdots F_2$ (see \Cref{ap:input-encoding}). On the right, the simulator's interaction with $\mathfrak{J}^{\mathsf{Enc}}$. It performs the Pauli filter $\idfilter{MT_1T_2}$ on the adversary's attack on the encoded state.}\label{fig:enc}
\end{figure}
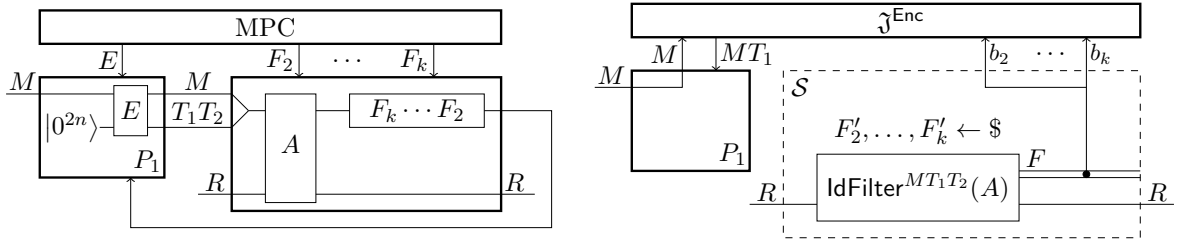

\subsection{Preparing ancilla qubits}\label{sec:ancilla}
Apart from encrypting the players' inputs, we also need a way to obtain encoded ancilla-zero states, which may be fed as additional input to the circuit. Since none of the players can be trusted to simply generate these states as part of their input, we need to treat them separately.

In~\cite{DNS12}, Alice generates an encoding of $\kb{0}$, and Bob tests it by entangling (with the help of the classical \MPC) the data qubit with a separate $\kb{0}$ qubit. Upon measuring that qubit, Bob then either detects a maliciously generated data qubit, or collapses it into the correct state. For details, see~\cite[Appendix E]{DNS12}.

Here, we take a similar approach, except with a public test on the shared traps. In order to guard against a player that may lie about the measurement outcomes during a test, we entangle the data qubits with \emph{all} traps. We do so using a random linear operator, similarly to the encoding described in the previous subsection.

Essentially, the protocol for preparing ancilla qubits is identical to Protocol~\ref{protocol:input-encoding} for input encoding, except that now we do not only test whether the $2n$ traps are in the $\kb{0}$ state, but also the data qubit: concretely, the linear operator $g$ acts on $2n+1$ elements instead of $2n$. That is,
\[
V := (E \otimes P) U_g (F_k \cdots F_2F_1)^{\dagger}.
\]
As a convention, Player 1 will always create the ancilla $\kb{0}$ states and encode them. In principle, the ancillas can be created by any other player, or by all players together.

Per the same proof as for Lemma~\ref{lem:input-encoding}, we have implemented the following ideal functionality, again making use of a classical \MPC as a black box.

\begin{definition}[Ideal functionality for encoding of $\kb{0}$]
	Let $I_{\advA} \subsetneq [k]$ be a set of corrupted players.
	\begin{enumerate}
			\item The trusted third party initializes a register $T_1$ with $\kb{0^n}$, applies a random $(n+1)$-qubit Clifford $E$ to $MT_1$, and sends these registers to player~1.
		\item All players $i \in I_{\advA}$ send a bit $b_i$ to the trusted third party. If $b_i = 0$ for all $i$, then the trusted third party stores the key $E$ in the state register $S$ of the ideal functionality. Otherwise, it aborts by storing $\bot$ in $S$.
	\end{enumerate}
\end{definition}

\section{Computation of Clifford and measurement}\label{sec:clifford-computation}
After all players have successfully encoded their inputs and sufficiently many ancillary qubits, they perform a quantum computation gate-by-gate on their joint inputs. In this section, we will present a protocol for circuits that consist only of Clifford gates and computational-basis measurements. The Clifford gates may be classically controlled (for example, on the measurement outcomes that appear earlier in the circuit). In Section~\ref{sec:T-gates}, we will discuss how to expand the protocol to general quantum circuits.

Concretely, we wish to achieve the functionality in Definition~\ref{def:ideal-MPQC} for all circuits $C$ that consist of Clifford gates and computational-basis measurements. As an intermediate step, we aim to achieve the following ideal functionality, where the players only receive an \emph{encoded} output, for all such circuits:

\begin{definition}[Ideal quantum $k$-party computation without decoding]\label{def:ideal-Clifford}
	Let $C$ be a quantum circuit on $W$ wires. Consider a partition of the wires into the players' input registers plus an ancillary register, as $[W] = R_1^{\inreg} \sqcup \cdots \sqcup R_k^{\inreg} \sqcup R^{\ancillareg}$, and a partition into the players' output registers plus a register that is discarded at the end of the computation, as $[W] = R_1^{\outreg} \sqcup \cdots \sqcup R_k^{\outreg} \sqcup R^{\discardreg}$. Let $I_A \subsetneq [k]$ be the set of corrupted players.
	\begin{enumerate}
		\item All players $i$ send their register $R_i^{\inreg}$ to the trusted third party.
		\item The trusted third party instantiates $R^{\ancillareg}$ with $\kb{0}$ states.
		\item The trusted third party applies $C$ to the wires $[W]$.
		\item For every player $i$ and every output wire $w \in R^{\outreg}_i$, the trusted third party samples a random $(n+1)$-qubit Clifford $E_w$, applies $ \rho \mapsto E_w(\rho \otimes \kb{0^n})E_w^{\dagger}$ to $w$, and sends the result to player i.
		\item All players $i \in I_A$ send a bit $b_{i}$ to the trusted third party.
		\begin{enumerate}
			\item If $b_{i} = 0$ for all $i$, all keys $E_w$ and all measurement outcomes are stored in the state register $S$.
			\item Otherwise, the trusted third party \texttt{abort}s by storing $\bot$ in $S$.
		\end{enumerate}
	\end{enumerate}
\end{definition}

To achieve the ideal functionality, we define several subprotocols. The subprotocols for encoding the players' inputs and ancillary qubits have already been described in Section~\ref{sec:setup-encoding}. It remains to describe the subprotocols for (classically-controlled) single-qubit Clifford gates (Section~\ref{sec:single-qubit-Cliffords}), (classically controlled) \CNOT gates (Section~\ref{sec:cnot}), and computational-basis measurements (Section~\ref{sec:measurement}).

In Section~\ref{sec:combining-clifford-gates}, we show how to combine the
subprotocols in order to compute any polynomial-sized Clifford+measurement
circuit. Our approach is inductive in the number of gates in the circuit. The
base case is the identity circuit, which is essentially covered in
Section~\ref{sec:setup-encoding}. In
Sections~\ref{sec:single-qubit-Cliffords}--\ref{sec:measurement}, we show that
the ideal functionality for any circuit $C$, followed by the subprotocol for a
gate $G$, results in the ideal functionality for the circuit $G \circ C$ ($C$
followed by $G$). As such, we can chain together the subprotocols to realize the
ideal functionality in Definition~\ref{def:ideal-Clifford} for any
polynomial-sized Clifford+measurement circuit. Combined with the decoding
subprotocol we present in Section~\ref{sec:decoding}, such a chain of
subprotocols satisfies Definition~\ref{def:ideal-MPQC} for ideal $k$-party
quantum Clifford+measurement computation with abort.

In Definition~\ref{def:ideal-Clifford}, all measurement outcomes are stored in the state register of the ideal functionality. We do so to ensure that the measurement results can be used as a classical control to gates that are applied after the circuit $C$, which can be technically required when building up to the ideal functionality for $C$ inductively. Our protocols can easily be altered to broadcast measurement results as they happen, but the functionality presented in Definition~\ref{def:ideal-Clifford} is the most general: if some player is supposed to learn a measurement outcome $m_{\ell}$, then the circuit can contain a gate $\X^{m_{\ell}}$ on an ancillary zero qubit that will be part of that player's output.

\subsection{Subprotocol: single-qubit Cliffords}\label{sec:single-qubit-Cliffords}
Due to the structure of the Clifford code, applying single-qubit Clifford is simple: the classical \MPC, who keeps track of the encoding keys, can simply update the key so that it includes the single-qubit Clifford on the data register. We describe the case of a single-qubit Clifford that is classically controlled on a previous measurement outcome stored in the \MPC's state.  The unconditional case can be trivially obtained by omitting the conditioning.

\begin{protocol}[Single-qubit Cliffords]\label{protocol:single-qubit-Clifford}
	Let $G^{m_\ell}$ be a single-qubit Clifford to be applied on a wire $w$ (held by a player $i$), conditioned on a measurement outcome $m_\ell$. Initially, player $i$ holds an encoding of the state on that wire, and the classical \MPC holds the encoding key $E$.
	\begin{enumerate}
		\item \MPC reads result $m_\ell$ from its state register $S$, and updates its internally stored key $E$ to $E(\left(G^{m_\ell}\right)^{\dagger} \otimes \I^{\otimes n})$.
	\end{enumerate}
\end{protocol}
If $m_{\ell} = 0$, nothing happens. To see that the protocol is correct for $m_{\ell} = 1$, consider what happens if the state $E(\rho \otimes \kb{0^n})E^{\dagger}$ is decoded using the updated key: the decoded output is
\begin{align*}
(E(G^{\dagger} \otimes \I^{\otimes n}))^{\dagger} E(\rho \otimes \kb{0^n})E^{\dagger} (E(G^{\dagger} \otimes \I^{\otimes n})) \ \ \ = \ \ \ G \rho G^{\dagger} \otimes \kb{0^n}.
\end{align*}
Protocol~\ref{protocol:single-qubit-Clifford} implements the ideal functionality
securely: given an ideal implementation $\IdealProtocol^{C}$ for some circuit
$C$, we can implement $G^{m_{\ell}} \circ C$ (i.e., the circuit $C$ followed by
the gate $G^{m_{\ell}}$) by performing Protocol~\ref{protocol:single-qubit-Clifford} right after the interaction with $\IdealProtocol^C$.

\begin{lemma}\label{lem:single-qubit-Clifford}
	Let $G^{m_\ell}$ be a single-qubit Clifford to be applied on a wire $w$ (held by a player $i$), conditioned on a measurement outcome $m_\ell$. Let $\RealProtocol^{G^{m_\ell}}$ be Protocol~\ref{protocol:single-qubit-Clifford} for the gate $G^{m_\ell}$, and $\IdealProtocol^{C}$ be the ideal functionality for a circuit $C$ as described in Definition~\ref{def:ideal-Clifford}. For all sets $I_{\advA} \subsetneq [k]$ of corrupted players and all adversaries $\advA$ that perform the interactions of players in $I_{\advA}$, there exists a simulator $\simS$ (the complexity of which scales polynomially in that of the adversary) such that for all environments $\envE$,
	\[
	\Pr[1 \leftarrow (\envE \leftrightarrows (\RealProtocol^{G^{m_\ell}} \diamond \IdealProtocol^{C})_{\advA})] = \Pr[1 \leftarrow (\envE \leftrightarrows \IdealProtocol^{G^{m_\ell} \circ C}_{\simS})].
	\]
\end{lemma}

\begin{proof}[Proof sketch]
	In the protocol $\RealProtocol^{G^{m_{\ell}}} \diamond \IdealProtocol^{C}$, an adversary has two opportunities to attack: once before its input state is submitted to $\IdealProtocol^{C}$, and once afterwards. We define a simulator that applies these same attacks, except that it interacts with the ideal functionality $\IdealProtocol^{G^{m_{\ell}} \circ C}$.
	
	Syntactically, the state register $S$ of $\IdealProtocol^{C}$ is provided as
  input to the \MPC in $\RealProtocol^{G^{m_{\ell}}}$, so that the \MPC can
  update the key as described by the protocol. As such, the output state of the
  adversary and the simulator are exactly equal. We provide a full proof in~\Cref{ap:proof-single-qubit-Clifford}.
\end{proof}

\subsection{Subprotocol: \CNOT gates}\label{sec:cnot}
The application of two-qubit Clifford gates (such as \CNOT) is more complicated than the single-qubit case, for two reasons.

First, a \CNOT is a \emph{joint} operation on two states that are encrypted with \emph{separate} keys. If we were to classically update two keys $E_1$ and $E_2$ in a similar fashion as in Protocol~\ref{protocol:single-qubit-Clifford}, we would end up with a new key $(E_1 \otimes E_2)(\CNOT_{1,n+2})$, which cannot be written as a product of two separate keys. The keys would become `entangled', which is undesirable for the rest of the computation.

Second, the input qubits might belong to separate players, who may not trust the authenticity of each other's qubits. In~\cite{DNS12}, authenticity of the output state is guaranteed by having both players test each state several times. In a multi-party setting, both players involved in the \CNOT are potentially dishonest, so it might seem necessary to involve all players in this extensive testing. However, because all our tests are publicly verified, our protocol requires less testing. Still, interaction with all other players is necessary to apply a fresh `joint' Clifford on the two ciphertexts.

\begin{protocol}[\CNOT]\label{protocol:CNOT}
	This protocol applies a \CNOT gate to wires $w_i$ (control) and $w_j$ (target), conditioned on a measurement outcome $m_{\ell}$. Suppose that player $i$ holds an encoding of the first wire, in register $M^iT^i_1$, and player $j$ of the second wire, in register $M^jT^j_1$. The classical \MPC holds the encoding keys $E_i$ and $E_j$.
	\begin{enumerate}
		\item If $i \neq j$, player $j$ sends their registers $M^jT^j_1$ to player $i$. Player $i$ now holds a ($2n+2$)-qubit state.
		\item Player $i$ initializes the registers $T_2^i$ and $T_2^j$ both in the state $\kb{0^n}$.\label{step:cnot-initialize-t2}
		\item \label{step:cnot-other-players} For all players $h$, \MPC samples random ($4n+2$)-qubit Cliffords $D_{h}$, and gives them to the respective players. Starting with player $i$, each player $h$ applies $D_{h}$ to $M^{ij}T^{ij}_{12}$,\footnote{We combine subscripts and superscripts to denote multiple registers: e.g., $T^{ij}_{12}$ is shorthand for $T^i_1T^i_2T^j_1T^j_2$.} and sends the state to player $h+1$. Eventually, player $i$ receives the state back from player $i-1$. \MPC remembers the applied Clifford
		\[
		D := D_{i-1}D_{i-2} \cdots D_1 D_k D_{k-1} \cdots D_i \, .
		\]
		\item \MPC samples random ($2n+1$)-qubit Cliffords $F_i$ and $F_j$, and tells player $i$ to apply
		\[
		V := (F_i \otimes F_j) \CNOT_{1,2n+2}^{m_{\ell}} (E_i^{\dagger} \otimes \I^{\otimes n} \otimes E_j^{\dagger} \otimes \I^{\otimes n})D^{\dagger}.
		\]
		Here, the \CNOT acts on the two data qubits inside the encodings.
		\item If $i \neq j$, player $i$ sends $M^jT^j_{12}$ to player $j$.
		\item Players $i$ and $j$ publicly test their encodings. The procedures are identical, we describe the steps for player $i$:
		\begin{enumerate}
			\item \MPC samples a random $(n+1)$-qubit Clifford $E_i'$, which will be the new encoding key. Furthermore, \MPC samples random $n$-bit strings $s_i$ and $r_i$, and a random classical invertible linear operator $g_i$ on $\mathbb{F}_2^{2n}$.
			\item \MPC tells player $i$ to apply
			\[
			W_i := (E_i' \otimes (\X^{r_i}\Z^{s_i})^{\reg{T_2^i}}) U_{g_i}^{\reg{T^i_{12}}} F_i^{\dagger}.
			\]
			Here, $U_{g_i}$ is as defined in Protocol~\ref{protocol:input-encoding}.\label{step:cnot-before-measurement}
			\item Player $i$ measures $T_2^i$ in the computational basis and reports the $n$-bit measurement outcome $r_i'$ to the \MPC.\label{step:cnot-measurement}
			\item \MPC checks whether $r_i' = r_i$. If it is not, \MPC sends \texttt{abort} to all players. If it is, the test has passed, and \MPC stores the new encoding key $E_i'$ in its internal memory.
		\end{enumerate}
	\end{enumerate}
\end{protocol}

\begin{lemma}\label{lem:cnot} Let $\RealProtocol^{\CNOT^{m_\ell}}$ be Protocol~\ref{protocol:CNOT}, to be executed on wires $w_i$ and $w_j$, held by players $i$ and $j$, respectively. Let $\IdealProtocol^{C}$ be the ideal functionality for a circuit $C$ as described in Definition~\ref{def:ideal-Clifford}. For all sets $I_{\advA} \subsetneq [k]$ of corrupted players and all adversaries $\advA$ that perform the interactions of players in $I_{\advA}$, there exists a simulator $\simS$ (the complexity of which scales polynomially in that of the adversary) such that for all environments $\envE$,
	\[
	\left| \Pr[1 \leftarrow (\envE \leftrightarrows (\RealProtocol^{\CNOT^{m_\ell}} \diamond \IdealProtocol^{C})_{\advA})] = \Pr[1 \leftarrow (\envE \leftrightarrows \IdealProtocol^{\CNOT^{m_\ell} \circ C}_{\simS})]\right| \leq \negl{n}.
	\]
\end{lemma}

\begin{proof}[Proof sketch]
	There are four different cases, depending on which of players $i$ and $j$ are dishonest.  In~\Cref{ap:cnot}, we provide a full proof by detailing the simulators for all four cases, but in this sketch, we only provide an intuition for the security in the case where both players are dishonest.
	
	It is crucial that the adversary does not learn any information about the keys ($E_i, E_j, E'_i, E'_j$), nor about the randomizing elements ($r_i$, $r_j$, $s_i$, $s_j$, $g_i$, $g_j$). Even though the adversary learns $W_i, W_j$, and $V$ explicitly during the protocol, all the secret information remains hidden by the randomizing Cliffords $F_i, F_j$, and $D$.
	
	We consider a few ways in which the adversary may attack. First, he may prepare a
  non-zero state in the registers $T_2^{i}$ (or $T_2^{j}$) in
  step~\ref{step:cnot-initialize-t2}, potentially intending to spread those errors into $M^iT_1^i$ (or $M^jT_1^j$). Doing so, however, will cause $U_{g_i}$
  (or $U_{g_j}$) to map the trap state to a random non-zero string, and the
  adversary would not know what measurement string $r_i'$ (or $r_j'$) to report.
  Since $g_i$ is unknown to the adversary,
  Lemma~\ref{lem:GL(2n,F2)-twirl} (see \Cref{ap:input-encoding}) is applicable in this case: it states that it suffices to measure $T^i_2$ in order to detect any errors in $T^i_{12}$. 
  
  Second, the adversary may fail to execute its instructions $V$ or $W_i \otimes W_j$ correctly. Doing so is equivalent to attacking the state right before or right after these instructions. In both cases, however, the state in $M^iT_1^i$ is Clifford-encoded (and the state in $T_2^i$ is Pauli-encoded) with keys unknown to the adversary, so the authentication property of the Clifford code prevents the adversary from altering the outcome.
	
	The simulator we define in~\Cref{ap:cnot} tests the adversary exactly for the types of attacks above. By using Pauli filters (see Definition~\ref{def:pauli-filter}), the simulator checks whether the attacker leaves the authenticated states and the trap states $T_2^i$ and $T_2^j$ (both at initialization and before measurement) unaltered. In the full proof, we show that the output state of the simulator approximates, up to an error negligible in $n$, the output state of the real protocol.
\end{proof}

\subsection{Subprotocol: Measurement}\label{sec:measurement}
Measurement of authenticated states introduces a new conceptual challenge. For a random key $E$, the result of measuring $E(\rho \otimes \kb{0^n})E^{\dagger}$ in a fixed basis is in no way correlated with the logical measurement outcome of the state $\rho$. However, the measuring player is also not allowed to learn the key $E$, so they cannot perform a measurement in a basis that depends meaningfully on $E$.

In~\cite[Appendix E]{DNS10}, this challenge is solved by entangling the state with an ancilla-zero state on a logical level. After this entanglement step, Alice gets the original state while Bob gets the ancilla state. They both decode their state (learning the key from the \MPC), and can measure it. Because those states are entangled, and at least one of Alice and Bob is honest, they can ensure that the measurement outcome was not altered, simply by checking that they both obtained the same outcome. The same strategy can in principle also be scaled up to $k$ players, by making all $k$ players hold part of a big (logically) entangled state. However, doing so requires the application of $k-1$ logical \CNOT operations, making it a relatively expensive procedure.

We take a different approach in our protocol. The player that performs the measurement essentially entangles, with the help of the \MPC, the data qubit with a random subset of the traps. The \MPC later checks the consistency of the outcomes: all entangled qubits should yield the same measurement result.

Our alternative approach has the additional benefit that the measurement outcome can be kept secret from some or all of the players. In the description of the protocol below, the \MPC stores the measurement outcome in its internal state. This allows the \MPC to classically control future gates on the outcome. If it is desired to instead reveal the outcome to one or more of the players, this can easily be done by performing a classically-controlled \X operation on some unused output qubit of those players.

\begin{protocol}[Computational-basis measurement]\label{protocol:measurement}
	Player $i$ holds an encoding of the state in a wire $w$ in the register $MT_1$. The classical \MPC holds the encoding key $E$ in the register $S$.
	\begin{enumerate}
		\item \MPC samples random strings $r,s \in \{0,1\}^{n+1}$ and $c \in \{0,1\}^n$.
		\item \MPC tells player $i$ to apply
		\[
		V := \X^r\Z^s \CNOT_{1,c} E^{\dagger}
		\]
		to the register $MT_1$,
		where $\CNOT_{1,c}$ denotes the unitary $\prod_{i \in [n]} \CNOT_{1,i}^{c_i}$ (that is, the string $c$ dictates with which of the qubits in $T_1$ the $M$ register will be entangled).
		\item Player $i$ measures the register $MT_1$ in the computational basis, reporting the result $r'$ to \MPC.
		\item \MPC checks whether $r' = r \oplus (m, m\cdot c)$ for some $m \in
      \{0,1\}$.\footnote{The $\cdot$ symbol represents scalar multiplication of
      the bit $m$ with the string $c$.} If so, it stores the measurement
      outcome $m$ in the state register $S$. Otherwise, it aborts by storing
      $\bot$ in $S$.\label{step:measurement-mpc-check}
		\item \MPC removes the key $E$ from the state register $S$.
	\end{enumerate}
\end{protocol}

\begin{lemma}\label{lem:measurement}
	Let $C$ be a circuit on $W$ wires that leaves some wire $w \leq W$ unmeasured. Let $\IdealProtocol^{C}$ be the ideal functionality for $C$, as described in Definition~\ref{def:ideal-Clifford}, and let $\RealProtocol^{\smeas{}}$ be Protocol~\ref{protocol:measurement} for a computational-basis measurement on $w$. For all sets $I_{\advA} \subsetneq [k]$ of corrupted players and all adversaries $\advA$ that perform the interactions of players in $I_{\advA}$, there exists a simulator $\simS$ (the complexity of which scales polynomially in that of the adversary) such that for all environments $\envE$,
	
	\[
	\left|\Pr[1 \leftarrow (\envE \leftrightarrows (\RealProtocol^{\smeas{}} \diamond \IdealProtocol^{C})_{\advA})] - \Pr[1 \leftarrow (\envE \leftrightarrows \IdealProtocol^{\smeas{} \circ C}_{\simS})]\right| \leq \negl{n}.
	\]
	
\end{lemma}

\begin{proof}[Proof sketch]
	The operation $\CNOT_{1,c}$ entangles the data qubit in register $M$ with a random subset of the trap qubits in register $T_1$, as dictated by $c$. In step~\ref{step:measurement-mpc-check} of Protocol~\ref{protocol:measurement}, the \MPC checks both for consistency of all the bits entangled by $c$ (they have to match the measured data) \emph{and} all the bits that are not entangled by $c$ (they have to remain zero).
	
	In~\Cref{lem:measurement-cnot-trick} in~\Cref{ap:measurement}, we show that checking the consistency of a measurement outcome after the application of $\CNOT_{1,c}$ is as good as measuring the logical state: any attacker that does not know $c$ will have a hard time influencing the measurement outcome, as he will  have to flip all qubits in positions $i$ for which $c_i = 1$ without accidentally flipping any of the qubits in positions $i$ for which $c_i = 0$. See~\Cref{ap:measurement} for a full proof that the output state in the real and simulated case are negligibly close.
\end{proof}

\subsection{Subprotocol: Decoding}\label{sec:decoding}
After the players run the computation subprotocols for all gates in the Clifford circuit, all they need to do is to decode their wires to recover their output. At this point, there is no need to check the authentication traps publicly: there is nothing to gain for a dishonest player by incorrectly measuring or lying about their measurement outcome. Hence, it is sufficient for all (honest) players to apply the regular decoding procedure for the Clifford code.

Below, we describe the decoding procedure for a single wire held by one of the players. If there are multiple output wires, then Protocol~\ref{protocol:decoding} can be run in parallel for all those wires.

\begin{protocol}
	[Decoding]\label{protocol:decoding}
	Player $i$ holds an encoding of the state $w$ in the register $MT_1$. The classical \MPC holds the encoding key $E$ in the state register $S$.
	\begin{enumerate}
		\item \MPC sends $E$ to player $i$, removing it from the state register $S$.
		\item Player $i$ applies $E$ to register $MT_1$.
		\item Player $i$ measures $T_1$ in the computational basis. If the outcome is not $0^n$, player $i$ discards $M$ and aborts the protocol.
	\end{enumerate}
\end{protocol}

\begin{lemma}\label{lem:decoding}
	Let $C$ be a circuit on $W$ wires that leaves a single wire $w \leq W$ (intended for player $i$) unmeasured. Let $\IdealProtocol^{C}$ be the ideal functionality for $C$, as described in Definition~\ref{def:ideal-Clifford}, and let $\IdealProtocol^{\MPQC}_{C}$ be the ideal MPQC functionality for $C$, as described in Definition~\ref{def:ideal-MPQC}. Let $\RealProtocol^{\Dec}$ be Protocol~\ref{protocol:decoding} for decoding wire $w$. For all sets $I_{\advA} \subsetneq [k]$ of corrupted players and all adversaries $\advA$ that perform the interactions of players in $I_{\advA}$, there exists a simulator $\simS$ (the complexity of which scales polynomially in that of the adversary) such that for all environments $\envE$,
	\[
	\Pr[1 \leftarrow (\envE \leftrightarrows (\RealProtocol^{\Dec} \diamond \IdealProtocol^{C})_{\advA})] = \Pr[1 \leftarrow (\envE \leftrightarrows \IdealProtocol^{\MPQC}_{C,\simS})].
	\]
\end{lemma}
\begin{proof}[Proof sketch]
	If player $i$ is honest, then he correctly decodes the state received from the ideal functionality $\IdealProtocol^{C}$. A simulator would only have to compute the adversary's abort bit for $\IdealProtocol^{\MPQC}_{C}$ based on whether the adversary decides to abort in either $\IdealProtocol^{C}$ or the \MPC computation in $\RealProtocol^{\Dec}$.
	
	If player $i$ is dishonest, a simulator $\simS$ runs the adversary on the input state received from the environment before inputting the resulting state into the ideal functionality $\IdealProtocol^{\MPQC}_{C}$. The simulator then samples a key for the Clifford code and encodes the output of $\IdealProtocol^{\MPQC}_C$, before handing it back to the adversary. It then simulates $\Pi^{\Dec}$ by handing the sampled key to the adversary. If the adversary aborts in one of the two simulated protocols, then the simulator sends abort to the ideal functionality $\IdealProtocol^{\MPQC}_C$.
\end{proof}

\subsection{Combining Subprotocols}\label{sec:combining-clifford-gates}
We show in this section how to combine the subprotocols of the previous sections
in order to perform multi-party quantum Clifford computation.

Recalling the notation defined in~\Cref{def:ideal-MPQC},
let $C$ be a quantum circuit on $W \in \mathbb{N}_{>0}$ wires, which are
partitioned into the players' input registers plus an ancillary register, as
$[W] = R_1^{\inreg} \sqcup \cdots \sqcup R_k^{\inreg} \sqcup R^{\ancillareg}$,
and a partition into the players' output registers plus a register that is discarded at the end of the computation, as $[W] = R_1^{\outreg} \sqcup \cdots
\sqcup R_k^{\outreg} \sqcup R^{\discardreg}$. We assume that $C$ is decomposed
in a sequence $G_1,...,G_m$ of operations where each $G_i$ is one of
the following operations:
\begin{itemize}
  \item a single-qubit Clifford on some wire $j \in [M]$;
  \item a CNOT on wires $j_1,j_2 \in [M]$ for $j_1 \ne j_2$;
  \item a measurement of the qubit on wire $j$ in the computational basis.
\end{itemize}
In Sections~\ref{sec:setup-encoding} and~\ref{sec:single-qubit-Cliffords}--\ref{sec:measurement}, we have presented subprotocols for encoding single qubits and perform these types of operations on single wires. The protocol for all players to jointly perform the bigger computation $C$ is simply a concatenation of those smaller subprotocols:

\begin{protocol}[Encoding and Clifford+measurement computation]\label{protocol:clifford}
	Let $C$ be a Clifford + measurement circuit composed of the gates $G_1, \dots, G_m$ on wires $[W]$ as described above.
  \begin{enumerate}
    \item For all $i \in [k]$ and $j \in R^{\inreg}_i$, run
      Protocol \ref{protocol:input-encoding} for the qubit in wire $j$.
    \item For all $j \in R^{\ancillareg}$, run  Protocol \ref{protocol:input-encoding}
      (with the differences described in~\Cref{sec:ancilla}).
    \item For all $j \in [m]$:
    \begin{enumerate}
      \item If $G_j$ is a single-qubit Clifford,
        run  Protocol \ref{protocol:single-qubit-Clifford} for $G_j$.
      \item If $G_j$ is a CNOT,
        run Protocol \ref{protocol:CNOT} for $G_j$.
      \item If $G_j$ is a computational-basis measurement,
        run Protocol \ref{protocol:measurement} for $G_j$.
    \end{enumerate}
    \item For all $i \in [k]$ and $j \in R^{\outreg}_i$, run
      Protocol \ref{protocol:decoding} for the qubit in wire $j$.
  \end{enumerate}
\end{protocol}

\begin{lemma}\label{lem:protocol-clifford}
	Let $\RealProtocol^{\Cliff}$ be Protocol~\ref{protocol:clifford}, and
  $\IdealProtocol^{\Cliff}$ be the ideal functionality described in
  Definition~\ref{def:ideal-MPQC} for the special case where the circuit consists of (a polynomial number of)
  Cliffords and measurements. For all sets $I_{\advA} \subsetneq [k]$ of corrupted players and all adversaries $\advA$ that perform the interactions of players in $I_{\advA}$ with $\RealProtocol$, there exists a simulator $\simS$ (the complexity of which scales polynomially in that of the adversary) such that for all environments $\envE$,
	\[
	|\Pr[1 \leftarrow (\envE \leftrightarrows \RealProtocol^{\Cliff}_{\advA})] -
  \Pr[1 \leftarrow (\envE \leftrightarrows \IdealProtocol^{\Cliff}_{\simS})| \leq \negl{n}.
	\]
\end{lemma}
\begin{proof}
  First notice that
  $\IdealProtocol^{\Cliff} = \IdealProtocol^{G_m \circ \dots \circ G_1 \circ \Enc}$ and
  $\RealProtocol^{\Cliff} = \RealProtocol^{\Dec} \diamond
  \RealProtocol^{G_m} \diamond ... \diamond \RealProtocol^{G_1}
  \diamond \RealProtocol^{\Enc}$. For simplicity, for some circuit $C'$ composed
  of gates $G'_{1},...G'_{m'}$, we write
  $\Pi^{C'} = \RealProtocol^{G'_{m'}} \diamond ... \diamond
  \RealProtocol^{G'_1}$. We also denote by $C'_{r,s}$ for $1 \leq r \leq s \leq
  m'$ the circuit composed by gates $G'_{r},...,G'_{s}$.

  We start by proving by induction that for all $i$, the following holds for all environments $\envE'$:
  \begin{align} 
    \Big|\Pr[1 \leftarrow (\envE' \leftrightarrows
  \RealProtocol^{C}
    \diamond \RealProtocol^{\Enc}_{\advA})]  \nonumber 
    -\Pr[1 \leftarrow (\envE' \leftrightarrows
    \RealProtocol^{C_{i,m}}
    \diamond \IdealProtocol^{C_{1,i-1} \circ \Enc}_{\simS_{\Enc}})\Big| 
    \leq  i \cdot \negl{n}. \label{eq:assumption-composition}
  \end{align}

  For the basis case $i = 1$, notice that from \Cref{lem:input-encoding}, there exists a simulator $\simS_{\Enc}$ such
  that for all $\envE'$
	\[
	|\Pr[1 \leftarrow (\envE' \leftrightarrows \RealProtocol^{\Enc}_{\advA})] -
  \Pr[1 \leftarrow (\envE' \leftrightarrows
  \IdealProtocol^{\Enc}_{\simS_{\Enc}})| \leq \negl{n},
	\]
  therefore, in particular for every $\envE''$ we have that
  \begin{align*} \label{eq:basis-composition}
    |\Pr[1 \leftarrow (\envE'' \leftrightarrows
  \RealProtocol^{C}
    \diamond \RealProtocol^{\Enc}_{\advA})] -
  \Pr[1 \leftarrow (\envE'' \leftrightarrows
  \RealProtocol^{C}
    \diamond \IdealProtocol^{\Enc}_{\simS_{\Enc}})|
    \leq \negl{n}.
  \end{align*}

  For the induction step, assume that our statement holds for some $i \geq
  1$. Then if $G_{i+1}$ is a single-qubit Clifford we have that
  \begin{align*}
    &|\Pr[1 \leftarrow (\envE' \leftrightarrows
    \RealProtocol^{C}
    \diamond \RealProtocol^{\Enc}_{\advA})] -
  \Pr[1 \leftarrow (\envE' \leftrightarrows
    \RealProtocol^{C_{i+1,m}}
    \diamond \IdealProtocol^{C_{1,i} \circ \Enc}_{\simS_{\Enc}})| \\
    &\leq |\Pr[1 \leftarrow (\envE' \leftrightarrows
    \RealProtocol^{C}
    \diamond \RealProtocol^{\Enc}_{\advA})] -
  \Pr[1 \leftarrow (\envE' \leftrightarrows
    \RealProtocol^{C_{i,m}}
    \diamond \IdealProtocol^{C_{1,i-1} \circ \Enc}_{\simS_{\Enc}})| \\
    &\; + |
  \Pr[1 \leftarrow (\envE' \leftrightarrows
    \RealProtocol^{C_{i,m}}
    \diamond \IdealProtocol^{C_{1,i-1} \circ \Enc}_{\simS_{\Enc}})
    -
  \Pr[1 \leftarrow (\envE' \leftrightarrows
    \RealProtocol^{C_{i+1,m}}
    \diamond \IdealProtocol^{C_{1,i} \circ \Enc}_{\simS_{\Enc}})| \\
    &\leq (i+1)\negl{n},
  \end{align*}
  where in the first step we use the triangle inequality and in the second step
  we use the induction hypothesis and  \Cref{lem:single-qubit-Clifford}.

  For the cases where $G_{i+1}$ is a CNOT or measurement, the same argument follows by using \Cref{lem:cnot,lem:measurement} accordingly.

  Finally, by~\Cref{lem:decoding}, we can also replace $\RealProtocol^{\Dec}$ by
  $\IdealProtocol^{\Dec}$, at the cost of $\negl{n}$.
 We note that the $m = \poly{n}$ negligible functions we accumulated by the use of \Cref{lem:single-qubit-Clifford,lem:cnot,lem:measurement,lem:decoding} only depend on the type of operation and not on the position $i$. Therefore, the result follows since $m \cdot \negl{n} = \negl{n}$.
\end{proof}

\section{Protocol: MPQC for general quantum circuits}\label{sec:T-gates}

In this section, we show how to lift the MPQC for Clifford operations (as laid out in \Cref{sec:setup-encoding,sec:clifford-computation}) to MPQC
for general quantum circuits.

The main idea is to use magic states for $\T$ gates, as described in
\Cref{sec:t-magic-states}. Our main difficulty here is that the magic states
must be supplied by the possibly dishonest players themselves.
We solve this problem in \Cref{sec:distillation} and then 
in \Cref{sec:universal-mpqc}, we describe the MPQC protocol for universal computation combining the results
from \Cref{sec:clifford-computation,sec:distillation}.

\subsection{Magic-state distillation}
\label{sec:distillation}
We now describe a subprotocol that allows the players to create the 
encoding of exponentially
good magic states, if the players do not abort.

\medskip

Our subprotocol can be divided into two parts. In the first part, player 1 is
asked to create many magic states, which the other players will test. After this step, if none of the players abort during the testing, then with high
probability the resource states created by player 1 are at least somewhat good. In the second part of the subprotocol, the
players run a
distillation procedure to further increase the quality of the magic states.

\begin{protocol}
	[Magic-state creation]\label{protocol:ms-creation}
	Let $t$ be the number of magic states we wish to create. Let $\ell := (t+k)n$.
\begin{enumerate}
  \item Player 1 creates $\ell$ copies of $\ket{\T}$ and encodes them separately
    using Protocol \ref{protocol:input-encoding} (jointly with the other players).
  \item \MPC picks random disjoint sets $S_2,\dots,S_k \subseteq [\ell]$ of size $n$
    each.
  \item For each $i \in 2,\dots,k$, player $i$ decodes the magic states indicated
    by $S_i$ (see Protocol \ref{protocol:decoding}), measures in the $\{\ket{\T}, \ket{\T^\perp}\}$-basis and aborts if any outcome is different from $\ket{\T}$
  \label{step:abort-magic}.
\item On the remaining encoded states, the players run Protocol
  \ref{protocol:clifford} for multi-party computation of Clifford circuits (but
    skipping the input-encoding step) to perform the magic-state distillation
    protocol described in Protocol~\ref{protocol:distillation}. Any randomness required in that protocol is sampled by the classical \MPC.    \label{step:distillation}
\end{enumerate}
\end{protocol}

We claim that Protocol \ref{protocol:ms-creation} implements the following ideal functionality for creating $t$ magic states, up to a negligible error:

\begin{definition}[Ideal functionality for magic-state
	creation]\label{def:ms-creation}
	Let $t$ be the number of magic states we wish to create. Let $I_{\advA} \subsetneq [k]$ be a set of corrupted players.
	\begin{enumerate}
		\item For every $i \in I_{\advA}$, player $i$ sends a bit $b_{i}$ to the trusted third party.
		\begin{enumerate}
			\item If $b_{i} = 0$ for all $i$, the trusted third party samples $t$
			random $(n+1)$-qubit Clifford $E_j$ for $1 \leq j \leq t$,
			and sends $E_j(\ket{\mathsf{T}} \otimes \ket{0^n})$ to Player $1$.
			\item Otherwise, the trusted third party sends \texttt{abort} to all players.
		\end{enumerate}
		\item Store the keys $E_j$, for $1 \leq j \leq t$ in the state register $S$ of the ideal functionality.
	\end{enumerate}
\end{definition}

\begin{lemma}\label{lem:protocol-ms-states}
 Let $\RealProtocol^{MS}$ be Protocol~\ref{protocol:ms-creation}, and
  $\IdealProtocol^{MS}$ be the ideal functionality described in
  Definition~\ref{def:ms-creation}. For all sets $I_{\advA} \subsetneq [k]$ of corrupted players and all adversaries $\advA$ that perform the interactions of players in $I_{\advA}$ with $\RealProtocol$, there exists a simulator $\simS$ (the complexity of which scales polynomially in that of the adversary) such that for all environments $\envE$,
	\[
	\left|\Pr[1 \leftarrow (\envE \leftrightarrows \RealProtocol^{MS}_{\advA})] - \Pr[1
  \leftarrow (\envE \leftrightarrows \IdealProtocol^{MS}_{\simS})\right| \leq
  \negl{n}.
	\]
\end{lemma}

We prove this lemma in \Cref{ap:protocol-ms-states}.

\subsection{MPQC protocol for universal quantum computation}
\label{sec:universal-mpqc}
Finally, we present our protocol for some arbitrary quantum computation.
For this setting, we extend the setup of~\Cref{sec:combining-clifford-gates} by considering quantum circuits $C = G_m...G_1$ where $G_i$ can be single-qubit
Cliffords, \CNOT{}s, measurements or, additionally, \T gates.

For that, we will consider a circuit $C'$ where each gate $G_i = \T$ acting on
qubit $j$ is then
replaced by the $\T$-gadget presented in~\Cref{fig:magic-state}, acting on the
qubit $j$ and a fresh new $\T$ magic state.

\begin{protocol}[Protocol for universal MPQC]\label{protocol:mpqc} Let $C$ be a polynomial-sized quantum circuit, and $t$ be the number of \T-gates in $C$.
\begin{enumerate}
  \item Run Protocol \ref{protocol:ms-creation}  to create $t$ magic states.
  \item Run Protocol \ref{protocol:clifford} for the circuit $C'$, which is equal to the circuit $C$, except each \T gate is replaced with the \T-gadget from \Cref{fig:magic-state}.
\end{enumerate}
\end{protocol}

\begin{theorem}\label{lem:protocol-mpqc}
	Let $\RealProtocol^{\MPQC}$ be Protocol~\ref{protocol:mpqc}, and
  $\IdealProtocol^{\MPQC}$ be the ideal functionality described in
  Definition~\ref{def:ideal-MPQC}. For all sets $I_{\advA} \subsetneq [k]$ of corrupted players and all adversaries $\advA$ that perform the interactions of players in $I_{\advA}$ with $\RealProtocol$, there exists a simulator $\simS$ (the complexity of which scales polynomially in that of the adversary) such that for all environments $\envE$,
	\[
	|\Pr[1 \leftarrow (\envE \leftrightarrows \RealProtocol^{\MPQC}_{\advA})] -
  \Pr[1 \leftarrow (\envE \leftrightarrows \IdealProtocol^{\MPQC}_{\simS})| \leq \negl{n}.
	\]
\end{theorem}
  \begin{proof}
    Direct from~\Cref{lem:protocol-clifford,lem:protocol-ms-states}.
  \end{proof}

\subsection{Round Complexity and MPC Calls}
\label{sec:roundcomplexity}
Recall that we are assuming access to an ideal (classical) MPC functionality defined in \Cref{def:classicalMPC}. One MPC call can produce outputs to all players simultaneously. In this section, we analyze the number of rounds of quantum communication, and the number of calls to the classical MPC. The actual implementation of the classical MPC is likely to result in additional rounds of classical communication. 

In the way we describe it, \Cref{protocol:input-encoding} encodes a single-qubit input (or an ancilla $\ket{0}$ state) using $k$ rounds of quantum communication and $O(1)$ MPC calls. Note that this protocol can be run in parallel for all input qubits per player, simultaneously for all players. Hence, the overall number of communication rounds for the encoding phase remains $k$, and the total number of calls to the MPC is $O(w )$ where $w$ is the total number of qubits.

\Cref{protocol:single-qubit-Clifford} for single-qubit Cliffords, \Cref{protocol:measurement} for measuring in the computational basis and \Cref{protocol:decoding} for decoding do not require quantum communication and use $O(1)$ MPC calls each, whereas \Cref{protocol:CNOT} for \CNOT requires at most $k+2$ rounds of quantum communication, and makes $O(1)$ MPC calls. Overall, \Cref{protocol:clifford} for encoding and Clifford+measurement computation require $O(d k)$ rounds of quantum communication and $O(w + g)$ calls to the MPC, where $d$ is the \CNOT-depth of the quantum circuit, and $g$ is the total number of gates in the circuit.

\smallskip

\Cref{protocol:ms-creation} for magic-state creation encodes $\ell := (t+k)n$ qubits in parallel using $k$ rounds of quantum communication (which can be done in parallel with the actual input encoding) and $O((t+k)n)$ MPC calls. Then a circuit of size $\pdistill(n)$ and \CNOT-depth $\ddistill(n)$ classically controlled Cliffords and measurements is run on each of the $t$ blocks of $n$ qubits each, which can be done in parallel for the $t$ blocks, requiring  $O(k \cdot \ddistill(n))$ rounds of quantum communication and $O(tn\cdot \pdistill(n))$ calls to the MPC.

Eventually, all \T-gate operations in the original circuit $C$ are replaced by the \T-gadget from \Cref{fig:magic-state}, resulting in one \CNOT and classically controlled Cliffords. Overall, our \Cref{protocol:mpqc} for universal MPQC requires $O(k \cdot (\ddistill(n) + d))$ rounds of quantum communication and $O(tn\cdot \pdistill(n) + w + g)$ calls to the classical MPC, where $d$ is the $\{\CNOT, \T\}$-depth of the circuit, $w$ is the total number of qubits and $g$ is the total number of gates in the circuit.

We notice that instead of evaluating each Clifford operation gate-by-gate, we could evaluate a general $w$-qubit Clifford using $O(k)$ rounds of quantum communication, similarly to the \CNOT protocol.
This could improve the parameter $d$ to be the $\T$ depth of the circuit, at the cost of requiring significantly more communication per round.

\bibliographystyle{alphaarxiv}
\bibliography{references}

\newpage
\appendix
  \section*{Appendices}
\section{Twirling}\label{ap:twirling}
One of the techniques we use in this work is the twirl over a group $\mathcal{G}$ of unitary operators, which maps a state (or channel) to its `$\mathcal G$-averaged' version. Specifically, the twirl of a state $\rho$ is defined
\begin{align*}
\twirl{\mathcal{G}}(\rho) := \frac 1 {|\mathcal G|}\sum_{U \in \mathcal{G}} U\rho U^{\dagger},
\end{align*}
and the twirl of a channel $\Lambda$ is defined as
\begin{align*}
\Twirl{\mathcal{G}}(\Lambda(\cdot)) := \frac 1 {|\mathcal G|}\sum_{U \in \mathcal{G}} U^{\dagger} (\Lambda (U (\cdot) U^{\dagger})) U.
\end{align*}
We sometimes abuse notation for non-unitary groups: for example, throughout this work we use $\twirl{GL(2n,\mathbb{F}_2)}(\cdot)$ to denote a twirl over the unitary group $\{U_g \mid g \in GL(2n, \mathbb{F}_2)\}$, where $U_g$ is defined as the unitary that applies $g$ in-place, i.e., $U_g\ket{t} = \ket{g(t)}$ for all $t \in \{0,1\}^{2n}$.

Twirling a state over the $n$-qubit Pauli group is equivalent to encrypting the state under the quantum one-time pad, and tracing out the encryption key. From the point of view of someone without that encryption key, the resulting state is fully mixed:

\begin{lemma}[Pauli twirl of a state]\label{lem:pauli-twirl-state}
	For all $n$-dimensional quantum states $\rho$,
	\begin{align*}
	\twirl{\Pauli_n}\left(\rho\right) = \tau.
	\end{align*}
\end{lemma}

\begin{proof}
	Write $\rho = \sum_{i,j\in\{0,1\}^n} \alpha_{ij} \ketbra{i}{j}$. For every $i,j$ we have
	\begin{align*}
	\twirl{\Pauli_n}\left(\ketbra{i}{j}\right) &= \E_{x,z \in \{0,1\}^n} \X^x\Z^z\ketbra{i}{j}\Z^z\X^x\\
	&= \E_{x,z \in \{0,1\}^n} (-1)^{z(i \oplus j)}\ketbra{i\oplus x}{j \oplus x}.
	\end{align*}
	Note that $\E_{z \in \{0,1\}^n} (-1)^{z(i \oplus j)} = 0$ whenever $i \neq j$ (i.e., $i \oplus j \neq 0$), and that the term evaluates to 1 whenever $i = j$. So
	\begin{align*}
	\twirl{\Pauli_n}\left(\ketbra{i}{j}\right) &= \left\{
	\begin{array}{l l}
	\E_x \kb{i \oplus x} = \tau &\text{if } i = j\\
	0 &\text{otherwise.}
	\end{array}
	\right.
	\end{align*}
	To conclude the proof, sum all terms of $\rho$ and use $\Tr(\rho)=\sum_i \alpha_{ii} = 1$ to get
	\begin{align*}
	\twirl{\Pauli_n}(\rho) = \sum_{i,j} \alpha_{ij} \twirl{\Pauli_n}\left(\ketbra{i}{j}\right) = \sum_i \alpha_{ii} \tau = \tau.
	\end{align*}
\end{proof}

Since the Pauli group is a subgroup of the Clifford group, Lemma~\ref{lem:pauli-twirl-state} also holds when twirling over the $n$-qubit Clifford group.

On a channel, the Pauli twirl has a similar effect of transforming a superposition of attack maps into a classical mixture of Pauli attacks. This transformation greatly simplifies analysis:

\begin{lemma}[Pauli twirl of a channel~{\cite[Lemma 5.1]{AB+17}}]\label{lem:pauli-twirl-channel}
	For all $V^{\reg{AB}} = \sum_{P \in \Pauli_{n}} P^{\reg{A}} \otimes V_P^{\reg{B}}$,
	\[
	\Twirl{\Pauli_n}^{\reg{A}}(V(\cdot)V^{\dagger}) = \sum_{P \in \Pauli_n} (P \otimes V_P)(\cdot)(P \otimes V_P)^{\dagger},
	\]
	where the twirl is applied on the $2^n$-dimensional register $A$.
\end{lemma}

\section{Proof of~\Cref{lem:pauli-filter}}\label{ap:pauli-filter}
\begin{proof}
	Given an arbitrary state $\rho^{TR}$ (for some reference system $R$), we calculate the result of applying $\paulifilter{\pauliset}{S}(U)$ to $\rho$. The state in $TR$ corresponding to $\kb{0}$ in the flag register $F$ is:
	\begin{align*}
	&\phantom{= \sum_{x,z,x',z'}} \ \Tr_{SS'} \left[\Pi U^{\reg{ST}}\left(\kb{\Phi}^{\reg{SS'}} \otimes \rho^{\reg{TR}}\right)U^{\dagger} \right]\\
	&= \sum_{\substack{(a,b) \in \pauliset \\ x,z,x',z'}} \Tr_{SS'} \left[\X^a\Z^b\kb{\Phi}\Z^b\X^a\X^x\Z^z \kb{\Phi}\Z^{z'}\X^{x'} \right]  \otimes U_{x,z}^{\reg{T}}\rho^{\reg{TR}}U_{x',z'}^{\dagger}\\
	&= \sum_{\substack{(a,b) \in \pauliset \\ x,z,x',z'}} \Tr_{SS'} \left[\kb{\Phi}\X^{a \oplus x}\Z^{b \oplus z} \kb{\Phi}\Z^{b \oplus z'}\X^{a \oplus x'}\right]  \otimes U_{x,z}^{\reg{T}}\rho^{\reg{TR}}U_{x',z'}^{\dagger} \cdot (-1)^{b \cdot (x\oplus x')}\\
	&= \sum_{(a,b) \in \pauliset} U_{a,b}^{\reg{T}} \rho^{\reg{TR}} U^{\dagger}_{a,b}.
	\end{align*}
	The calculation for the $\kb{1}$-flag is very similar, after observing that
	\[\Id - \sum_{(a,b) \in \pauliset} \X^a\Z^b\kb{\Phi}\Z^b\X^a = \sum_{(a,b) \not\in \pauliset} \X^a\Z^b\kb{\Phi}\Z^b\X^a.\]
\end{proof}

\section{Security of Clifford code}
\label{sec:proof-clifford-code}
\begin{lemma}[Variation on~{\cite[Theorem 3.7]{AM17}}]\label{lem:clifford-code-property}
	Let $M$, $T$, and $R$ be registers with $\log|M| = 1$ and $\log|T| = n > 0$. Let $U^{\reg{MTR}}$ be a unitary, and write $U = \sum_{x,z \in \{0,1\}^{n+1}} (\X^x\Z^z)^{\reg{MT}} \otimes U_{x,z}^{\reg{R}}$. Then for any state $\rho^{\reg{MR}}$,
  \begin{align*}
    &\left\|\Twirl{\Clifford_{n+1}}^{\reg{MT}}(U)\left(\rho^{\reg{MR}} \otimes \kb{0^n}^{\reg{T}}\right)
	- \left(U_{0,0}^{\reg{R}}\rho U_{0,0}^{\dagger} \otimes \kb{0^n}^{\reg{T}} +
	\Tr_{M}\left[ \sum_{(x,z) \neq (0,0)} U^{\reg{R}}_{x,z}\rho U^{\dagger}_{x,z}
    \right] \otimes \tau^{\reg{MT}}\right)\right\|_1 \\
    &\leq \negl{n}.
  \end{align*}
\end{lemma}
\begin{proof}
	The proof is a straightforward application of the Clifford twirl~\cite[Lemma 3.6]{AB+17}, which is similar to Lemma~\ref{lem:pauli-twirl-channel}, but for the Clifford group. Using this Clifford twirl in the first step, and writing $U = \sum_{x,z} (\X^x \Z^z)^{\reg{MT}} \otimes U^{\reg{R}}_{x,z}$, we derive
	\begin{align}
	&\Twirl{\Clifford_{n+1}}^{\reg{MT}}(U)(\rho \otimes \kb{0^n}) \nonumber \\
	= \ &\sum_{x,z} \E_C (C\X^x\Z^zC^{\dagger} \otimes U_{x,z}) (\rho \otimes \kb{0^n}) ({C^\dagger}\X^x\Z^zC \otimes U_{x,z}^{\dagger})\nonumber\\
	= \ &U_{0,0}^{\reg{R}}\rho U_{0,0}^{\dagger} + \sum_{(x,z) \neq (0,0)} \E_C (C\X^x\Z^zC^{\dagger} \otimes U_{x,z}) (\rho \otimes \kb{0^n}) ({C^\dagger}\X^x\Z^zC \otimes U_{x,z}^{\dagger})\label{eq:clifford-property-before}\\
	= \ &U_{0,0}^{\reg{R}}\rho U_{0,0}^{\dagger} + \sum_{(x,z) \neq (0,0)} \E_{(x',z') \neq (0,0)} (\X^{x'}\Z^{z'} \otimes U_{x,z}) (\rho \otimes \kb{0^n}) (\X^{x'}\Z^{z'} \otimes U_{x,z}^{\dagger})\label{eq:clifford-property-after}\\
	\approx_{\negl{n}} \ &U_{0,0}^{\reg{R}}\rho U_{0,0}^{\dagger} + \sum_{(x,z) \neq (0,0)} U_{x,z}^{\reg{R}} \left(\twirl{\Pauli_{n+1}}^{\reg{MT}} (\rho \otimes \kb{0^n})\right) U_{x,z}^{\dagger}\nonumber\\
	= \ &U_{0,0}^{\reg{R}}\rho U_{0,0}^{\dagger} + \sum_{(x,z) \neq (0,0)} \Tr_M \left[U_{x,z}^{\reg{R}} \rho U_{x,z}^{\dagger}\right] \otimes \tau^{\reg{MT}}.\nonumber
	\end{align}
	In the step from Equation~\eqref{eq:clifford-property-before} to Equation~\eqref{eq:clifford-property-after}, we used the fact that any non-identity Pauli is mapped to a random non-identity Pauli by expectation over the Clifford group.
\end{proof}

\begin{corollary}
	The Clifford authentication code with $n$ trap qubits is $\negl{n}$-secure.
\end{corollary}

\begin{proof}
	In the decoding procedure for the $n$-trap Clifford code, the register $T$ is measured using the two-outcome measurement defined by the projector $\Pi := \kb{0^n}$.
	Note that, given an attack $\advA$,
	\begin{align*}
	\E_{k \in \KeySet} \left[\Dec_k \left( \advA \left( \Enc_k \left(\rho \right)\right)\right)\right] = \mathcal{L}^{\Pi}\left(\Twirl{\Clifford_{n+1}}(\advA)\left(\rho^{\reg{MR}} \otimes \kb{0^n}^{\reg{T}}\right)\right),
	\end{align*}
	where $\mathcal{L}^{\Pi}(X) := \Tr_T [\Pi X \Pi ] + \kb{\bot}^M \otimes \Tr_{MT}[\overline{\Pi} X \overline{\Pi}]$.
	 Then apply Lemma~\ref{lem:clifford-code-property}, and use the fact that $\Tr[\kb{0}^{\reg{T}} \tau^{\reg{MT}}] = 2^{-n}$. In the terminology of Definition~\ref{def:auth-security}, we may explicitly describe $\Attack_{\acc} := U_{0,0}(\cdot)U_{0,0}^{\dagger}$ and $\Attack_{\rej} := \sum_{(x,z) \neq (0,0)} U_{x,z}(\cdot)U_{x,z}^{\dagger}$ for $\mathcal A=U(\cdot)U^\dagger$ and $U$ decomposed as in the proof of Lemma \ref{lem:clifford-code-property}.
\end{proof}

\section{Proof of~\Cref{lem:distillation-works}}\label{sec:proof-distillation}
For the rest of this section, fix a basis $\ket{\hat{0}}:=\ket{\mathsf{T}}$ and $\ket{\hat{1}}:=\ket{\mathsf{T}^\bot}$. For $w\in\{0,1\}^m$, we will let
$\ket{\hat w} :=\ket{\hat w_1}\dots\ket{\hat w_m}.$
Then the all-0s string in this basis represents $m$ copies of $\ket{\mathsf{T}}$.
We analyze Circuit~\ref{protocol:distillation}.
It is simple to verify that
$\hat{\mathsf{Z}} =  \ket{\hat 0}\bra{\hat 0} - \ket{\hat 1}\bra{\hat 1}$
(up to a global phase). The first step of the circuit is to apply $\hat{\mathsf{Z}}$ with probability $\frac{1}{2}$ to each qubit, which has the effect of \emph{dephasing} the qubit, or equivalently, making the state diagonal, in the $\{\ket{\hat{0}},\ket{\hat{1}}\}$ basis. More precisely, if we let $\rho=\sum_{w,w'\in \{0,1\}^m} \alpha_{w,w'}\ket{\hat w}\bra{\hat w'}$:
\begin{align}
\rho &\mapsto \sum_{w,w'\in\{0,1\}^m}\alpha_{w,w'}\bigotimes_{i=1}^m\frac{1}{2}\left(\ket{\hat w_i}\bra{\hat w_i'}+\hat{\mathsf{Z}}\ket{\hat w_i}\bra{\hat w_i'}\hat{\mathsf{Z}}\right)\\
&= \sum_{w,w'\in\{0,1\}^m}\alpha_{w,w'}\bigotimes_{i=1}^m\frac{1}{2}\left(\ket{\hat w_i}\bra{\hat w_i'}+(-1)^{w_i+w_i'}\ket{\hat w_i}\bra{\hat w_i'}\right)\\
&= \sum_{w\in\{0,1\}^m}\alpha_{w,w}\ket{\hat w}\bra{\hat w} = : \rho'.\label{eq:rho-prime}
\end{align}

Let $\Xi'$ denote the quantum channel given by steps 2--3 of Circuit~\ref{protocol:distillation}. Note that $\Xi'$ is symmetric: Given two inputs $\rho$ and $\rho'=\pi\rho\pi^\dagger$, after step 2, either state will be mapped to $\sum_{\pi'\in S_m}\frac{1}{m!}\pi'\rho\pi'^\dagger$.
Thus, we can apply Theorem D.1 of \cite{DNS12}, which states the following:
\begin{theorem}[\cite{DNS12}]\label{thm:DNS-distillation}
Let $\sigma$ be an $m$-qubit state, diagonal in the basis $\{\ket{\hat w}:w\in\{0,1\}^m\}$, and suppose $\Pi_{LW}\sigma=\sigma$. Let $\Xi'$ be any CPTP map from $m$ qubits to $t$ qubits such that $\Xi'(\pi\omega\pi^\dagger)=\Xi'(\omega)$ for any $n$-qubit state $\omega$ and any $\pi\in S_m$. Then, letting $\delta_s=\frac{s}{m}$:
$$\norm{\Xi'(\sigma)-(\ket{\hat 0}\bra{\hat 0})^{\otimes t}}_1\leq (m+1)\max_{s\leq \ell}\norm{\Xi'\left(\left((1-\delta_s)\ket{\hat 0}\bra{\hat 0}+\delta_s\ket{\hat 1}\bra{\hat 1}\right)^{\otimes m}\right)-(\ket{\hat 0}\bra{\hat 0})^{\otimes t}}_1.$$
\end{theorem}

We can put everything together to prove Lemma~\ref{lem:distillation-works}.
\begin{proof}[Proof of Lemma~\ref{lem:distillation-works}]
Let $\rho'$ be as in \eqref{eq:rho-prime}. We have
$$\Tr(\Pi_{LW}\rho')=\Tr(\Pi_{LW}\rho)\geq 1-\eps.$$
Thus, write $\rho'=(1-\eps)\sigma+\eps\sigma'$ for $\sigma=\frac{1}{1-\eps}\Pi_{LW}\rho'$. Applying Theorem~\ref{thm:DNS-distillation}, we get
$$\norm{\Xi'(\sigma)-\left(\ket{\hat 0}\bra{\hat 0}\right)^{\otimes t}}_1 \leq (m+1) \max_{s\leq \ell}\norm{\Xi'\left(\left((1-\delta_s)\ket{\hat 0}\bra{\hat 0}+\delta_s\ket{\hat 1}\bra{\hat 1}\right)^{\otimes m}\right)-(\ket{\hat 0}\bra{\hat 0})^{\otimes t}}_1.$$
On a symmetric state, $\Xi'$ is simply the state distillation protocol of \cite{BK04}, applied $t$ times in parallel to $m/t$ qubits each time. Let $\Phi$ be one state distillation protocol distilling one qubit from $m/t$ (so $\Xi'$ acts as $\Phi^{\otimes t}$ on symmetric states). By Theorem~\ref{thm:distillation}, using $\delta_s=\frac{s}{m}\leq \frac{\ell}{m}$, and
$$\tau=\left((1-\delta_s)\ket{\hat 0}\bra{\hat 0}+\delta_s\ket{\hat 1}\bra{\hat 1}\right)^{\otimes m/t},$$
we have
$$1-\bra{\hat 0}\Phi(\tau)\ket{\hat 0}\leq O\left({\left(5\delta_s\right)^{(m/t)^c}}\right)\leq O\left(\left(5\frac{\ell}{m}\right)^{(m/t)^c}\right)$$
for $c\approx 0.2$. Let $\delta = (5\ell/m)^{(m/t)^c}$. Using the inequality
  between trace distance and fidelity, $\norm{\rho-\ket{\psi}\bra{\psi}}_1\leq
  2\sqrt{1-\bra{\psi}\rho\ket{\psi}}$, we have %
\begin{align*}
&\norm{\Xi'\left(\left((1-\delta_s)\ket{\hat 0}\bra{\hat 0}+\delta_s\ket{\hat 1}\bra{\hat 1}\right)^{\otimes m}\right)-(\ket{\hat 0}\bra{\hat 0})^{\otimes t}}_1\\
={}& \norm{\Phi(\tau)^{\otimes t}-(\ket{\hat 0}\bra{\hat 0})^{\otimes t}}_1 \\
\leq{}& 2\sqrt{1-(\bra{\hat 0}\Phi(\tau)\ket{\hat 0})^{ t}} \leq 2\sqrt{1-(1-\delta)^t}.
\end{align*}
Since $1-x\leq e^{-x}$ for all $x$, and $e^{-2x}\leq 1-x$ whenever $x\leq 1/2$,
  it follows that:
$$1-2\delta t \leq e^{-2\delta t}\leq (1-\delta)^t.$$
Thus
$$\norm{\Xi'\left(\left((1-\delta_s)\ket{\hat 0}\bra{\hat 0}+\delta_s\ket{\hat 1}\bra{\hat 1}\right)^{\otimes m}\right)-(\ket{\hat 0}\bra{\hat 0})^{\otimes t}}_1\leq 2\sqrt{2\delta t}.$$
Thus, we have:
\begin{align*}
\norm{\Xi(\rho)-(\ket{\hat 0}\bra{\hat 0})^{\otimes t}}_1 &= \norm{\Xi'(\rho')-(\ket{\hat 0}\bra{\hat 0})^{\otimes t}}_1
= \norm{(1-\eps)\Xi'(\sigma)+\eps\Xi'(\sigma')-(\ket{\hat 0}\bra{\hat 0})^{\otimes t}}_1\\
&\leq (1-\eps)(m+1)2\sqrt{2\delta t}+\eps = O\left(m\sqrt{t} (5\ell/m)^{(m/t)^c/2}+\eps\right).
\end{align*}
\end{proof}

\section{Proof of \Cref{lem:input-encoding}}
\label{ap:input-encoding}
Before we prove Lemma~\ref{lem:input-encoding}, let us begin by zooming in on the test phase (steps~\ref{step:encoding-Ug}--\ref{step:encoding-measurement}): we show in a separate lemma that, with high probability, it only checks out if the resulting state is a correctly-encoded one.

For a projector $\Pi$ on two $n$-qubit quantum registers $T_1T_2$, define the quantum channel $ \mathcal L^{\Pi}$ by
\begin{align*}
\mathcal{L}^{\Pi}(X)&:=\Pi(X)\Pi+\proj{\bot}\Tr\left[\bar\Pi X\right]
\end{align*}
where $\ket\bot$ is a distinguished state on $T_1T_2$ with $\Pi\ket\bot=0$.
	Furthermore, for $s\in\{0,1\}^n$, define the ``full" and ``half" projectors
	\begin{align}
		\Pi_{s,F}&:=\begin{cases}
			\proj{0^{2n}}& \text{if $s=0^n$}\\
			0&\text{else}
		\end{cases} \label{eqn:pifull}\\
	 	\Pi_{s,H}&:=\mathbb I\otimes \proj{s}. \label{eqn:pihalf}
	 \end{align}
	  The following lemma shows that measuring $\Pi_{s,F}$ on the one hand, and applying a twirl $\mathcal T_{\mathrm{GL}(2n, \mathbb F_2)}$ followed by a measurement of $\Pi_{s,H}$ on the other hand are equivalent as tests in the above protocol.

\begin{lemma}\label{lem:GL(2n,F2)-twirl}
For any $s \in \{0,1\}^n$, applying a random element of $\mathrm{GL}(2n, \mathbb F_2)$ followed by $\mathcal{L}^{\Pi_{s,H}}$ is essentially equivalent to applying $\mathcal{L}^{\Pi_{s,F}}$:
\begin{equation*}
	\left\|\mathcal{L}^{\Pi_{s,F}}-\mathcal{L}^{\Pi_{s,H}}\circ \mathcal
  T_{\mathrm{GL}(2n, \mathbb F_2)}\right\|_\diamond\leq 12\cdot 2^{-\frac n 2} = \negl{n}.
\end{equation*}
\end{lemma}
\begin{proof}
	First, observe the following facts about a random $g\in\mathrm{GL}(2n, \mathbb F_2)$. Of course, $g 0=0$ by linearity. On the other hand, $gx$ is uniformly random on $\F_2^{2n}\setminus \{0\}$ for $x\neq 0$. More generally, $x$ and $y$ with $x\neq 0\neq y$ are linearly independent if and only if $x\neq y$, and therefore $(gx,gy)$ is uniformly random on $\left\{(x,y)\in\left(\F_2^{2n}\setminus \{0\}\right)^2|x\neq y\right\}$. In the following we abbreviate $\mathcal T:=\mathcal T_{\mathrm{GL}(2n, \mathbb F_2)}$. We calculate for $x,y\in \F_2^{2n}\setminus \{0\}$ with $x\neq y$,
\begin{align*}
\mathcal T(\proj 0)&=\proj 0\\
	\mathcal T(\proj x)&=\frac{\mathbb I-\proj 0}{2^{2n}-1}\\
	\mathcal T(\ketbra{x}{0})&=\left(2^{2n}-1\right)^{-\frac 1 2}\ketbra{+'}{0}\\
	\mathcal T(\ketbra{x}{y})&=\left(2^{2n}-1\right)^{-1}\left(2^{2n}-2\right)^{-1}\sum_{\substack{z,t\in \F_2^{2n}\setminus \{0\}\\z\neq t}}\ketbra{z}{t}\\
	&=\left(2^{2n}-2\right)^{-1}\left(\proj{+'}-\frac{\mathbb I-\proj 0}{2^{2n}-1}\right)=:S.
\end{align*}
Here we have defined the unit vector
\begin{equation*}
\ket{+'} :=\left(2^{2n}-1\right)^{-\frac 1 2}\sum_{x\in\F_2^{2n}\setminus \{0\}}\ket x.
\end{equation*}
We can now evaluate $T$ on an arbitrary input density matrix,
\begin{align}
	\mathcal T(\rho_{T_1T_2E})=&\sum_{x,y\in\{0,1\}^{2n}}\mathcal T(\ketbra{x}{y})_{T_1T_2}\otimes \bra x_{T_1T_2}\rho_{T_1T_2E}\ket y_{T_1T_2}\nonumber\\
	=&\proj 0_{T_1T_2}\otimes\bra 0_{T_1T_2}\rho_{T_1T_2E}\ket 0_{T_1T_2}+\ketbra{0}{+'}_{T_1T_2}\otimes\bra 0_{T_1T_2}\rho_{T_1T_2E}\ket{+'}_{T_1T_2}\nonumber\\
	&+\ketbra{+'}{0}_{T_1T_2}\otimes \bra{+'}_{T_1T_2}\rho_{T_1T_2E}\ket{0}_{T_1T_2}+\frac{\mathbb I-\proj 0}{2^{2n}-1}\otimes \Tr_{T_1T_2}[(\mathbb I-\proj 0)_{T_1T_2}\rho_{T_1T_2E}]\nonumber\\
	&+S\otimes \sum_{\substack{x,y\in \F_2^{2n}\setminus \{0\}\\x\neq y}}\bra x_{T_1T_2}\rho_{T_1T_2E}\ket y_{T_1T_2}.\label{eq:Trho}
\end{align}
We calculate
\begin{align*}
  \|\Pi_{s,H}\ket{+'}\|_2\leq&\sqrt{\frac{2^n}{2^{2n}-1}} \quad \text{ and  }
  \quad
	\left\|\Pi_{s,H}\frac{\mathbb I-\proj
  0}{2^{2n}-1}\Pi_{s,H}\right\|_1\leq\frac{2^n-1}{2^{2n}-1},
\end{align*}
  where the first inequality is tight for $s = 0$ and the second inequality
  for $s \ne 0$. We also have
  \begin{align*} 
	\|S\|_1=&2\frac{1-\left(2^{2n}-1\right)^{-1}}{2^{2n}-2}
\end{align*}
Using these quantities, we analyze $\Pi_{s,H}\mathcal T(\rho_{T_1T_2E})\Pi_{s,H}$ term by term according to equation \eqref{eq:Trho}. We have
\begin{align}
&	\Pi_{s,H}\proj 0_{T_1T_2}\Pi_{s,H}=\begin{cases}
\proj 0_{T_1T_2} & s=0\\
0 & \mathrm{else},
\end{cases}
\label{eq:firstone}\\
&	\left\|\Pi_{s,H}\ketbra{0}{+'}_{T_1T_2}\Pi_{s,H}\otimes\bra 0_{T_1T_2}\rho_{T_1T_2E}\ket{+'}_{T_1T_2}\right\|_1\le \sqrt{\frac{2^n}{2^{2n}-1}},\\
&	\left\|\Pi_{s,H}\frac{\mathbb I-\proj 0}{2^{2n}-1}\Pi_{s,H}\otimes \Tr_{T_1T_2}[(\mathbb I-\proj 0)_{T_1T_2}\rho_{T_1T_2E}]\right\|_1\le\frac{2^n}{2^{2n}-1}\text{ , and}\\
\end{align}
  \begin{align}
&	\left\|\Pi_{s,H}S\Pi_{s,H}\otimes \sum_{\substack{x,y\in \F_2^{2n}\setminus \{0\}\\x\neq y}}\bra x_{T_1T_2}\rho_{T_1T_2E}\ket y_{T_1T_2}\right\|_1\le\left\|\Pi_{s,H}S\Pi_{s,H}\otimes \sum_{x,y\in \F_2^{2n}\setminus \{0\}}\bra x_{T_1T_2}\rho_{T_1T_2E}\ket y_{T_1T_2}\right\|_1\nonumber\\
	&+\left\|\Pi_{s,H}S\Pi_{s,H}\otimes \sum_{x\in \F_2^{2n}\setminus \{0\}}\bra x_{T_1T_2}\rho_{T_1T_2E}\ket x_{T_1T_2}\right\|_1\nonumber\\
	&\le 2\left(2^{2n}-2\right)^{-1}\left(\left\|\Pi_{s,H}\proj{+'}\Pi_{s,H}\right\|_1+\left\|\Pi_{s,H}\frac{\mathbb I-\proj 0}{2^{2n}-1}\Pi_{s,H}\right\|_1\right)\nonumber\\
	&\qquad \Tr\left[\left(\left(2^{2n}-1\right)\proj{+'}+\mathbb I-\proj 0\right)_{T_1T_2}\rho_{T_1T_2 E}\right]\nonumber\\
  &\le 2\left(2^{2n}-2\right)^{-1}\left(\frac{2^n}{2^{2n} -1} + \frac{2^n}{2^{2n} -1}\right)
	 \left(2\left(2^{2n}-1\right)\right]\nonumber\\
  &\le 8\frac{2^{n}\left(2^{2n}-1\right)}{\left(2^{2n}-1\right)\left(2^{2n}-2\right)}\nonumber\\
	&\le 2\frac{2^{n}}{\left(2^{2n}-2\right)} \label{eq:lastone}.
\end{align}
The first, second and fifth inequality use  normalization of $\rho$ and the third and fourth inequality use the triangle inequality for the trace norm. The fifth inequality additionally uses H\"older's inequality. Using the observation that
\begin{equation}
	\Pi_{s,F}\mathcal T(\rho_{T_1T_2E})\Pi_{s,F}=\begin{cases}
		\proj 0_{T_1T_2}\otimes\bra 0_{T_1T_2}\rho_{T_1T_2E}\ket 0_{T_1T_2}& s=0\\
		0&\text{else}
	\end{cases},
\end{equation}
together with Equations \eqref{eq:Trho} and \eqref{eq:firstone}-\eqref{eq:lastone}, we get that there exists a $\rho_{T_1T_2E}$  such that
\begin{align*}
	\left\|\mathcal{L}^{\Pi_{s,F}}-\mathcal{L}^{\Pi_{s,H}}\circ \mathcal T_{\mathrm{GL}(2n, \mathbb F_2)}\right\|_\diamond=&\left\|\mathcal{L}^{\Pi_{s,F}}(\rho)-\mathcal{L}^{\Pi_{s,H}}\circ \mathcal T_{\mathrm{GL}(2n, \mathbb F_2)}(\rho)\right\|_1\\
	\le&\left\|\Pi_{s,H}\ketbra{0}{+'}_{T_1T_2}\Pi_{s,H}\otimes\bra 0_{T_1T_2}\rho_{T_1T_2E}\ket{+'}_{T_1T_2}\right\|_1\\
	&+\left\|\Pi_{s,H}\ketbra{+'}{0}_{T_1T_2}\Pi_{s,H}\otimes \bra{+'}_{T_1T_2}\rho_{T_1T_2E}\ket{0}_{T_1T_2}\right\|\\
	&+\left\|\Pi_{s,H}\frac{\mathbb I-\proj 0}{2^{2n}-1}\Pi_{s,H}\otimes \Tr_{T_1T_2}[(\mathbb I-\proj 0)_{T_1T_2}\rho_{T_1T_2E}]\right\|_1\\
	&+\left\|\Pi_{s,H}S\Pi_{s,H}\otimes \sum_{\substack{x,y\in \F_2^{2n}\setminus \{0\}\\x\neq y}}\bra x_{T_1T_2}\rho_{T_1T_2E}\ket y_{T_1T_2}\right\|_1\\
	\le&2\sqrt{\frac{2^n}{2^{2n}-1}}+\frac{2^n}{2^{2n}-1}+8\frac{2^n}{2^{2n}-2}\\
	\le& 12\cdot 2^{-\frac n 2},
\end{align*}
which is negligible.
\end{proof}

Now that we have established that it suffices to measure only the $T_2$ register (after applying a random $g \in GL(2n,\mathbb{F}_2)$), we are ready to prove the security of Protocol~\ref{protocol:input-encoding}:

\begin{proof}[Proof of Lemma~\ref{lem:input-encoding}]
  
	We consider two cases: either player 1 is honest, or she is corrupted.

	\paragraph{Case 1: player 1 is honest.} See Figure~\ref{fig:enc}.
  For the setting where player 1 is honest, we prove security in the worst case, where \emph{all} other players are corrupted: $I_{\advA} = \{2,3,\dots,k\}$. If, instead, some of these players are not corrupted, a simulator can simulate the actions of every honest player $h \neq 1$ (by applying a random Clifford), and interleave these honest actions with the adversarial maps of the corrupted players. The resulting map is a special case of the adversarial map we consider below. Since the only task of the honest players $h \neq 1$ is to apply a random Clifford, it is sufficient if the simulator samples this Clifford itself.

	The corrupted players act as one entity whose honest action is to apply $U_H := F_k F_{k-1} \cdots F_3 F_2$, and return the state to player 1. Without loss of generality, assume that $\advA$ is unitary by expanding the side-information register $R$ as necessary. Then, define an attack unitary $A := U_H^{\dagger}\advA$, so that we may write $\advA = U_HA$. In other words, we establish that $\advA$ consists of a unitary attack $A$, followed by the honest unitary $U_H$. Note that $A$ may depend arbitrarily on its instructions $F_2$ through $F_k$

	The simulator $\simS$ has access to the ideal functionality only through the ability to submit the bits $b_i$ for players $i \ne 1$. It does not receive any input from the environment, except for a side-information register $R$. Define the simulator as follows (in terms of an adversarial map $A$):
	\begin{simulator}[see right side of Figure~\ref{fig:enc}]
          
		On input register $R$ received from the environment, do:
		\begin{enumerate}
		\item Sample random $F_2', \dots, F_k' \in \Clifford_{2n+1}$.\footnote{Whenever a simulator samples random elements, it does so by running the ideal functionality for classical MPC with the adversary it is currently simulating. If that ideal functionality aborts, the simulator will also abort by setting $b_i = 1$ for the adversarial players $i$. In that case, the simulated output state and the real output state will be indistinguishable by security of the classical MPC. To avoid clutter in the exposition of our simulators and proofs, we will ignore this technicality, and pretend that the simulator generates the randomness itself.}
		\item Run $\idfilter{MT_1T_2}(A)$ on the register $R$, using the instructions $F_2', F_3', \dots, F_k'$ to determine $A$. (See \Cref{sec:pauli-filter}.)
		\item If the flag register is 0, set $b_i = 0$ for all $i \neq 1$. Otherwise, set $b_i = 1$ for all $i \neq 1$. Submit the bits $b_i$ to the ideal functionality.
		\end{enumerate}
	\end{simulator}
    We will consider the joint state in the output register $R_1^{\outreg} = MT_1$, the state register $S$, and the attacker's side-information register $R$ in both the real and the ideal (simulated) case. In both cases, it will be useful to decompose the attack map $A$ as
    \begin{align*}
    A = \sum_{a,c \in \{0,1\}^{n+1}} \left(\X^a\Z^c\right)^{\reg{MT_1}} \otimes A_{a,c}^{\reg{R}}.
    \end{align*}
    We start by analyzing the ideal case. By \Cref{def:input-encoding} of the ideal encoding functionality and Lemma~\ref{lem:pauli-filter}, and using $\pauliset = \{(0,0)\}$ with $0$ as an abbreviation for $0^n$, the output state in $MT_1RS$ in case of accept (setting all $b_i = 0$) is
    \begin{align}
    \E_{E} E^{\reg{MT_1}}
    \left(
    A_{0,0}^{\reg{R}}
    \rho^{\reg{MR}}
    A_{0,0}^{\dagger}
    \otimes \kb{0^n}^{\reg{T_1}}
    \right)E^{\dagger} \otimes \kb{E}^\reg{S}.\label{eq:encoding-1-ideal-accept} 
    \end{align}
    The output state in $MT_1RS$ in case of reject is (again by \Cref{def:input-encoding} and Lemma~\ref{lem:pauli-filter})
    \begin{align}
    &\twirl{\Clifford_{n+1}}^{\reg{MT_1}}
    \left( \sum_{(x,z) \neq (0,0)}
    A_{x,z}^{\reg{R}}
    \rho^{\reg{MR}}
    A_{x,z}^{\dagger}
    \otimes \kb{0^n}^{\reg{T_1}}
    \right)\otimes \kb{\bot}^\reg{S}\\
    = \ &\tau^{\reg{MT_1}}
    \otimes
    \sum_{(x,z) \neq (0,0)}
    A_{x,z}
    \rho_{R}
    A_{x,z}^{\dagger}
    \otimes \kb{\bot}^\reg{S}.\label{eq:encoding-1-ideal-reject}
    \end{align}

    Next, we consider the state in $MT_1RS$ after the real protocol is executed, and argue that it is negligibly close to Equations~\eqref{eq:encoding-1-ideal-accept}+\eqref{eq:encoding-1-ideal-reject}. Again, we first consider the accept case. Following the steps in Protocol~\ref{protocol:input-encoding} on an input state $\rho^{\reg{MR}}$, and noting that
    \begin{align*}
    \left(F_1^{\dagger}\right)^{\reg{MT_1T_2}}A^{\reg{MT_1T_2R}}F_1\left(\rho^{\reg{MR}} \otimes \kb{0^{2n}}^{\reg{T_1T_2}}\right)F_1^{\dagger}A^{\dagger}F_1
    = \left(\Twirl{\Clifford_{2n+1}}^{\reg{MT_1T_2}}(A)\right)\left(\rho \otimes \kb{0^{2n}}\right),
    \end{align*}
    the output state in the case of accept is
    \begin{align}
    &\phantom{=}\E_{E,r,s} \bra{r}^{\reg{T_2}}
    \left(E \otimes \X^r\Z^s\right)
    \twirl{\mathrm{GL}_{(2n,\F_2)}}^{\reg{T_1T_2}}\left(
    \left(\Twirl{\Clifford_{2n+1}}^{\reg{MT_1T_2}}(A)\right)\left(\rho \otimes \kb{0^{2n}}\right)
    \right)
    \left(E \otimes \X^r\Z^s\right)^{\dagger}
    \ket{r} \otimes \kb{E}^{\reg{S}} \nonumber \\
    &=\E_{E} E^{\reg{MT_1}} \bra{0^n}^{\reg{T_2}}
    \twirl{\mathrm{GL}_{(2n,\F_2)}}^{\reg{T_1T_2}}\left(
    \left(\Twirl{\Clifford_{2n+1}}^{\reg{MT_1T_2}}(A)\right)\left(\rho \otimes \kb{0^{2n}}\right)
    \right)
    \ket{0^n} E^{\dagger}\otimes \kb{E}^{\reg{S}} \nonumber \\
    &\approx_{\negl{n}} \E_{E} E^{\reg{MT_1}} \bra{0^{2n}}^{\reg{T_1T_2}}
    \left(
    \left(\Twirl{\Clifford_{2n+1}}^{\reg{MT_1T_2}}(A)\right)\left(\rho \otimes \kb{0^{2n}}\right)
    \right)
    \ket{0^{2n}} E^{\dagger}\otimes \kb{E}^{\reg{S}},\label{eq:encoding-1-real-accept}
    \end{align}
    where the approximation follows from Lemma~\ref{lem:GL(2n,F2)-twirl}. This is where the authentication property of the Clifford code comes in: by Lemma~\ref{lem:clifford-code-property}, only the part of $A$ that acts trivially on $MT_1T_2$ remains after the measurement of $T_1T_2$. Thus, Eq.~\eqref{eq:encoding-1-real-accept} $\approx_{\negl{n}}$ Eq.~\eqref{eq:encoding-1-ideal-accept}.

    The reject case of the real protocol is similar: again using Lemmas~\ref{lem:GL(2n,F2)-twirl} and~\ref{lem:clifford-code-property}, we can see that it approximates (up to a negligible factor in $n$) Eq.~\eqref{eq:encoding-1-ideal-reject}.

     We conclude that, from the point of view of any environment, the real output state in registers $MT_1SR$ (encoding, memory state, and side information) is indistinguishable from the simulated state.

	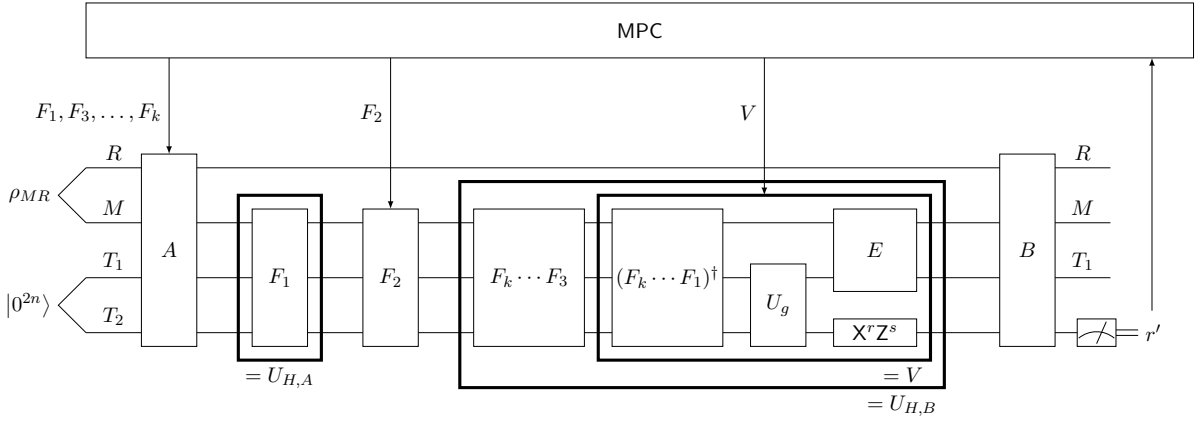
\begin{figure}
		\resizebox{\textwidth}{!}{
			\centering
			\begin{tikzpicture}
				\draw (0,3) -- (-.5,2.5) -- (0,2);
				\node[anchor=east] at (-.5,2.5) {$\rho_{MR}$};
				\draw (0,3) -- (18.5,3);
				\node[anchor=south] at (.5,3) {$R$};
				\draw (0,2) -- (18.5,2);
				\node[anchor=south] at (.5,2) {$M$};
				
				\draw (0,1) -- (-.5,.5) -- (0,0);
				\node[anchor=east] at (-.5,.5) {$\ket{0^{2n}}$};
				\draw (0,1) -- (18.5,1);
				\node[anchor=south] at (.5,1) {$T_1$};
				\draw (0,0) -- (18.5,0);
				\node[anchor=south] at (.5,0) {$T_2$};
				
				\draw[fill=white] (1,-.25) rectangle (2,3.25);
				\node at (1.5,1.5) {$A$};

				\draw[ultra thick] (2.75,-.5) rectangle (4.25,2.5);
				\node[anchor=north east] at (4.25,-.5) {$=U_{H,A}$};
				\draw[fill=white] (3,-.25) rectangle (4,2.25);
				\node at (3.5,1) {$F_1$};
				
				\draw[fill=white] (5,-.25) rectangle (6,2.25);
				\node at (5.5,1) {$F_2$};
				
				\draw[fill=white] (7,-.25) rectangle (9,2.25);
				\node at (8,1) {$F_k \cdots F_3$};
				
				\draw[ultra thick] (9.25,-.5) rectangle (15.25,2.5);
				\node[anchor=north east] at (15.25,-.5) {$=V$};
				
				\draw[fill=white] (9.5,-.25) rectangle (11.5,2.25);
				\node at (10.5,1) {$(F_k \cdots F_1)^{\dagger}$};
				
				\draw[fill=white] (12,-.25) rectangle (13,1.25);
				\node at (12.5,.5) {$U_g$};
				
				\draw[fill=white] (13.5,.75) rectangle (15,2.25);
				\node at (14.25,1.5) {$E$};
				\draw[fill=white] (13.5,-.25) rectangle (15,.25);
				\node at (14.25,0) {$\X^r\Z^s$};
				
				\draw[ultra thick] (6.75,-1) rectangle (15.5,2.75);
				\node[anchor=north east] at (15.5,-1) {$=U_{H,B}$};
				
				\draw[fill=white] (16.5,-.25) rectangle (17.5,3.25);
				\node at (17,1.5) {$B$};
				
				\node[anchor=south] at (18,3) {$R$};
				\node[anchor=south] at (18,2) {$M$};
				\node[anchor=south] at (18,1) {$T_1$};
				
				\draw (18,-.05) -- (19,-.05);
				\draw (18,.05) -- (19,.05);
				\node[anchor=west] at (19,0) {$r'$};
				
				\node at (18.25,0) {$\meas$};
				
				\draw[-latex] (19.25,.4) -- (19.25,5);
				
				\draw[fill=white] (0,5) rectangle (20,6);
				\node at (10,5.5) {\MPC};
				
				\draw[-latex] (1.5,5) -- (1.5,3.25);
				\node[anchor=east] at (1.5,4) {$F_1,F_3,\dots,F_k$};
				
				\draw[-latex] (5.5,5) -- (5.5,2.25);
				\node[anchor=east] at (5.5,4) {$F_2$};
				
				\draw[-latex] (12.25,5) -- (12.25,2.5);
				\node[anchor=east] at (12.25,4) {$V$};
				
			\end{tikzpicture}
	}
	\caption{Execution of the input-encoding protocol $\RealProtocol^{\Enc}$ (see \Cref{protocol:input-encoding}), where player 2 is the only honest player (case 2).}  \label{fig:enc-case-2-real}
\end{figure}

	\paragraph{Case 2: player 1 is dishonest.} For the same reason as in the first case, we assume that the only honest player is player 2, i.e., $I_A = \{1, 3, 4, ..., k\}$.

	In the real protocol, the adversary interacts with the honest player 2, and has two opportunities to attack: before player 2 applies its Clifford operation, and after.

	The adversaries' actions before the interaction with player 2 can, without loss of generality, be described by a unitary $U_{H,A} \cdot A$, that acts on the input state $\rho^{\reg{MR}}$, plus the registers $T_1T_2$ that are initialized to zero. The unitary $U_{H,A}$ is player 1's honest operation $F_1^{\reg{MT_1T_2}}$.

	Similarly, the adversaries' actions after the interaction with player 2 can be described by a unitary $B \cdot U_{H,B}$, followed by a computational-basis measurement on $T_2$ which results in an $n$-bit string $r'$. Again, $U_{H,B}$ is the honest unitary $VF_k\cdots F_4F_3$ that should be applied jointly by players $3, 4, \dots,k, 1$. See \Cref{fig:enc-case-2-real}.

	For any adversary, described by such unitaries $A$ and $B$, define a simulator as follows (see \Cref{fig:enc-case-2-ideal}):

	\begin{simulator} \label{sim:enc-case-2} On input register $MR$ received from the environment, do:
			\begin{enumerate}
			\item Initialize $b_i = 0$ for all $i \in I_A$.
			\item Sample random $F_1, F_2, \dots, F_k \in \Clifford_{2n+1}$. Run $\zerofilter{T_1T_2}(A)$ on $MR$, using the instructions $F_i$ (for all $i \in I_{\advA}$) to determine $A$. If the filter flag is 1, abort by setting $b_1 = 1$.\label{sim:encoding-2-early-abort}
			\item Input the $M$ register into the ideal functionality, and receive a state in the register $MT_1$.
			\item Run $\xfilter{T_2}(B)$ on $MT_1R$, using the instruction $V:= F_3^{\dagger}F_4^{\dagger} \cdots F_k^{\dagger}$ to determine $B$. (This choice of $V$ ensures that the honest action $U_{H,B}$ is identity.) If the filter flag is 1, abort by setting $b_1 = 1$. \label{sim:encoding-2-late-abort}
			\item Submit all the $b_i$ to the ideal functionality.
		\end{enumerate}
	\end{simulator}

	\begin{figure}
	\resizebox{\textwidth}{!}{
		\centering
		\begin{tikzpicture}
		\draw (0,3) -- (-.5,2.5) -- (0,2);
		\node[anchor=east] at (-.5,2.5) {$\rho_{MR}$};
		\draw (0,3) -- (19.5,3);
		\node[anchor=south] at (.5,3) {$R$};
		\node[anchor=south] at (19,3) {$R$};
		\draw[-latex] (0,2) -- (5.25,2) -- (5.25,5);
		\node[anchor=south] at (.5,2) {$M$};
		\node[anchor=east] at (5.25,4.5) {$M$};
		
		\draw[fill=white] (1,.75) rectangle (4.5,3.25);
		\node at (2.75,2) {$\zerofilter{T_1T_2}(A)$};
		
		\draw (4.5,.95) -- (5,.95) -- (6.5,-.55) -- (15.5,-.55);
		\draw (4.5,1.05) -- (5.07,1.05) -- (6.57,-.45) -- (15.5,-.45);
		\node[anchor=south] at (4.75,1) {$F$};

		\draw[-latex] (9,5) -- (9,1) -- (10.5,1);
		\draw (11,1) -- (19.5,1);
		\node[anchor=east] at (9,4.5) {$T_1$};
		\node[anchor=south] at (19,1) {$T_1$};
		
		\draw[-latex] (9.5,5) -- (9.5,2) -- (10.5,2);
		\draw (11,2) -- (19.5,2);
		\node[anchor=west] at (9.5,4.5) {$M$};
		\node[anchor=south] at (19,2) {$M$};
		
		\draw[fill=white] (10.5,3.25) rectangle (14.5,-.25);
		\node at (12.5,1.5) {$\xfilter{T_2}(B)$};
		
		\draw (14.5,-.05) -- (15.5,-.05);
		\draw (14.5,.05) -- (15.5,.05);
		\node[anchor=south] at (14.75,0) {$F$};
		
		\draw[fill=white] (15.5,-.75) rectangle (16.5,.25);
		\node at (16,-.25) {or};
		\draw (16.5,-.2) -- (18,-.2);
		\draw (16.5,-.3) -- (18,-.3);
		\draw[fill=black] (17.5,-.25) circle (.05);
		\draw[-latex] (17.5,-.25) -- (17.5,5);
		\node[anchor=east] at (17.5,4.5) {$b_i$};
		
		\draw[fill=white] (0,5) rectangle (20,6);
		\node at (10,5.5) {$\IdealProtocol^{\Enc}$};
		\node at (2.75,3.5) {$F_1,F_3,\dots,F_k \leftarrow \$$};
		
		\draw[dashed] (.75,4) rectangle (18,-1);
		\node[anchor=south west] at (.75,-1) {$\simS$};
		
		\end{tikzpicture}
	}
	\caption{Interaction between the ideal functionality and the simulator $\simS$ (see \Cref{sim:enc-case-2}) for the case in which only player 2 is honest (case 2). The simulator performs two filters, and sets the abort bit to 1 whenever at least one of the flags $F$ is set to 1.}  \label{fig:enc-case-2-ideal}
\end{figure}
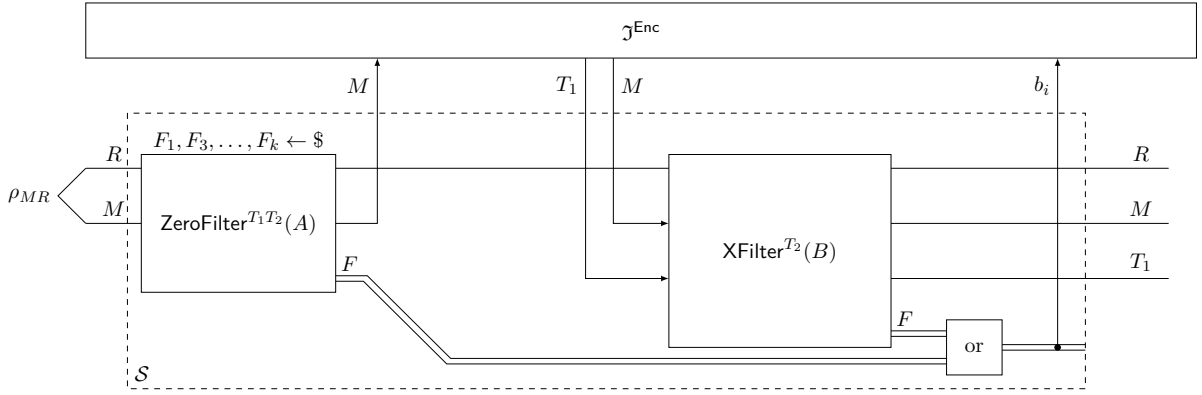

	Similarly to the previous case, we consider the output state in the registers $MT_1RS$ in both the ideal (simulated) case, and the real case, as computed on an input state $\rho^{\reg{MR}}$.

	Again, we decompose the attack maps $A$ and $B$ as

	\begin{align}
	A &= \sum_{a,c \in \{0,1\}^{2n}} \left(\X^a\Z^c\right)^{\reg{T_1T_2}} \otimes A_{a,c}^{\reg{MR}} \ \ \ ,\label{eq:encoding-2-decomposition-A}\\
	B &= \sum_{b,d \in \{0,1\}^n} \left(\X^b\Z^d\right)^{\reg{T_2}} \otimes B_{b,d}^{\reg{MT_1R}}\ \ \ \label{eq:encoding-2-decomposition-B}.
	\end{align}
	Note that the decompositions are taken over different registers for $A$ and $B$. In the derivations below, we will often abbreviate $A_a := \sum_c A_{a,c}$, and, in the subscripts, we will abbreviate $0$ for the all-zero string.

	In the ideal or simulated case, one of three things may happen: the simulator sets $b_1$ to 0 (signaling accept to the ideal functionality), or sets $b_1$ to 1 in step~\ref{sim:encoding-2-early-abort}, or sets $b_1$ to 1 in step~\ref{sim:encoding-2-late-abort} (both signaling reject to the ideal functionality). The ideal output state is thus  the sum of three separate terms, which we will analyze separately.
	
	We start with the the accept case, where both filters result in a $\kb{0}$ flag. Using the decompositions from Equations~\eqref{eq:encoding-2-decomposition-A} and~\eqref{eq:encoding-2-decomposition-B}, we apply Lemma~\ref{lem:pauli-filter} to see that the resulting state is
	\begin{align}
	\E_E \sum_d B_{0,d}^{\reg{MT_1R}} E\left(
	A_{0}\rho A_{0}^{\dagger} \otimes \kb{0^n}^{\reg{T_1}}
	\right)E^{\dagger}B_{0,d}^{\dagger} \otimes \kb{E}^{\reg{S}}.\label{eq:encoding-2-ideal-accept}
	\end{align}
	Here, $E$ is the key that the ideal functionality samples (and stores in the register $S$) when it is called to encode $M$.
	
	Next, we consider the simulator choosing $b_1 = 1$ already in step~\ref{sim:encoding-2-early-abort}, the zero filter has failed. In this case, the ideal functionality does not store the encoding key $E$ in the register $S$. This allows us to view the Clifford encoding as a twirl on the Clifford group. The output state is (by~\Cref{lem:pauli-filter})
	\begin{align}
	&\sum_{a \neq 0^{2n},b,d} B_{b,d} \twirl{\Clifford_{n+1}}^{\reg{MT_1}}
	\left(
	A_{a}
	\rho^{\reg{MR}}
	A_{a}^{\dagger}
	\otimes \kb{0^n}^{\reg{T_1}}
	\right) B_{b,d}^{\dagger}
	\otimes \kb{\bot}^{\reg{S}}\nonumber\\
	= &\sum_{a \neq 0^{2n},b,d} B_{b,d}
	\left(
	\Tr_M\left[A_a
	\rho^{\reg{MR}}
	A_a^{\dagger}\right] \otimes \tau^{\reg{MT_1}}
	\right) B_{b,d}^{\dagger}
	\otimes \kb{\bot}^{\reg{S}}.\label{eq:encoding-2-ideal-reject-1}
    \end{align}
	Note that in this case, the flag in the \X filter does not influence the bit $b_1$ (it is already set to 1). Therefore, both terms in~\Cref{lem:pauli-filter} survive, and all pairs $(b,d)$ are included in the sum.
	
	Finally, we look at the case where the zero filter does not result in changing $b_1$, but the \X filter does, in step~\ref{sim:encoding-2-late-abort}. If this happens, the key $E$ is erased so we can again apply a Clifford twirl, and the output state is (by~\Cref{lem:pauli-filter})
	\begin{align}
	&\sum_{b \neq 0^n, d} B_{b,d}
	\twirl{\Clifford_{n+1}}^{\reg{MT_1}}\left(
	A_{0}
	\rho^{\reg{MR}} A_{0}^{\dagger}
	\otimes \kb{0^n}^{\reg{T_1}}
	\right)
	B_{b,d}^{\dagger} \otimes \kb{\bot}^{\reg{S}}\nonumber\\
	=&\sum_{b \neq 0^n} B_{b,d}
	\left(\Tr_M\left[A_0
	\rho^{\reg{MR}}
	A_0^{\dagger}\right] \otimes \tau^{\reg{MT_1}}
	\right)B_{b,d}^{\dagger} \otimes \kb{\bot}^{\reg{S}}.\label{eq:encoding-2-ideal-reject-2}
	\end{align}
	In summary, the output state in the ideal case is
	\begin{align*}
	\text{Eq. } \eqref{eq:encoding-2-ideal-accept} + \text{Eq. } \eqref{eq:encoding-2-ideal-reject-1} + \text{Eq. } \eqref{eq:encoding-2-ideal-reject-2}.
	\end{align*}
	
	In the real protocol, only one measurement is performed at the end. The output state in the real case is thus a sum of only two terms: an accept and reject case. We will again analyze these separately, and will show that the accept state is approximately equal to Equation~\eqref{eq:encoding-2-ideal-accept}, while the reject state approximates Equations~$\eqref{eq:encoding-2-ideal-reject-1}+\eqref{eq:encoding-2-ideal-reject-2}$.

	Following Protocol~\ref{protocol:input-encoding} on an input state
  $\rho^{\reg{MR}}$, and canceling out the $F_i$ and $F_i^{\dagger}$ terms that
  are part of the honest actions, we first consider the state in case of accept.
  We abbreviate \begin{align*}
	\sigma := &\E_g E^{\reg{MT_1}}U_g^{\reg{T_1T_2}}(A(\rho \otimes \kb{0^{2n}})A^{\dagger})U_g^{\dagger}E^{\dagger}\\
	= &\phantom{\E_g\ } E^{\reg{MT_1}}\twirl{GL(2n,\F_2)}^{\reg{T_1T_2}}(A(\rho \otimes \kb{0^{2n}})A^{\dagger})E^{\dagger},
	\end{align*}
	where we are allowed to view $\E_g U_g(\cdot)U_g^{\dagger}$ as a Twirling operation, since $A$ and $B$ are independent of $g$. We decompose the attack $B$ as in Equation~\eqref{eq:encoding-2-decomposition-B}, and derive the accept case

	\begin{align}
	&\E_{E,r,s} \bra{r}^{\reg{T_2}}
	B
	\left(\X^r\Z^s\right)^{\reg{T_2}}
	\sigma
	\left(\X^r\Z^s\right)^{\dagger}
	B^{\dagger}
	\ket{r}
    \otimes \kb{E}^{\reg{S}}\label{eq:encoding-2-real-accept-first-step}\\
	=&\E_{E,r,s} \bra{0}^{\reg{T_2}}\left(\X^r\Z^s\right)^{\dagger \reg{T_2}}
	B
	\left(\X^r\Z^s\right)^{\reg{T_2}}
	\sigma
	\left(\X^r\Z^s\right)^{\dagger}
	B^{\dagger}
	\left(\X^r\Z^s\right)\ket{0}
	\otimes \kb{E}^{\reg{S}} \nonumber \\
	=&\E_{E,r,s} \sum_{b,d,b',d'} \bra{0}^{\reg{T_2}}\left(\left(\X^r\Z^s\X^b\Z^d\X^r\Z^s\right)
	\otimes B_{b,d}\right)
	\sigma
	\left(\left(\X^r\Z^s\X^{b'}\Z^{d'}\X^r\Z^s\right)
	\otimes B_{b',d'}^{\dagger}\right)
	\ket{0}
	\otimes \kb{E}^{\reg{S}}\label{eq:encoding-2-real-accept-before-twirl}\\
	=&\E_{E} \sum_{b,d} \bra{b}^{\reg{T_2}}
	B_{b,d}
	\sigma
	B_{b,d}^{\dagger}
	\ket{b}
	\otimes \kb{E}^{\reg{S}}\label{eq:encoding-2-real-accept-after-twirl}\\
	=&\E_{E} \sum_{b,d} B_{b,d} \Tr_{T_2} \left[\Pi_{b,H}^{\reg{T_2}}
	\sigma
	\Pi_{b,H}^{\dagger}
	\right]
	B_{b,d}^{\dagger}
	\otimes \kb{E}^{\reg{S}} ,\label{eq:encoding-2-real-accept-before-sigma}
	\end{align}
        where $\Pi_{b,H}$ is defined in~\Cref{eqn:pihalf}.
	From Equation~\eqref{eq:encoding-2-real-accept-before-twirl} to~\eqref{eq:encoding-2-real-accept-after-twirl}, we used the Pauli twirl to remove all terms for which $(b,d) \neq (b',d')$. This application of the Pauli twirl is possible, because neither $A$ nor $B$ depends on $r,s$.

	We continue with the accept case by expanding $\sigma$ in
  Equation~\eqref{eq:encoding-2-real-accept-before-sigma}, and evaluate the
  effect of the random $\mathrm{GL}_{2n,\F_2}$ element on $T_1T_2$ using Lemma \ref{lem:GL(2n,F2)-twirl}. It ensures
  that, if $A$ altered the $T_1T_2$ register, then $B$ cannot successfully reset
  the register $T_2$ to the correct value $r$. It follows that
	\begin{align}
	\text{Eq.}\ \eqref{eq:encoding-2-real-accept-before-sigma} &\approx \E_{E} \sum_{b,d}
	B_{b,d}^{\reg{MT_1R}}
	E^{\reg{MT_1}}
	\Tr_{T_2}\left[
	\Pi_{b,F}^{\reg{T_1T_2}}
	A\left(\rho \otimes \kb{0^{2n}}\right)A^{\dagger}
	\Pi_{b,F}^{\dagger}
	\right]
	E^{\dagger}
	B_{b,d}^{\dagger}
	\otimes \kb{E}^{\reg{S}}\label{eq:encoding-2-real-accept-reject}\\
	&=\E_{E}
	\sum_{d}
	B_{0,d}^{\reg{MT_1R}}
	E^{\reg{MT_1}}
	\Tr_{T_2}\left[
	\kb{0^{2n}}
	A\left(\rho \otimes \kb{0^{2n}}\right)A^{\dagger}
	\kb{0^{2n}}
	\right]
	E^{\dagger}
	B_{0,d}^{\dagger}
	\otimes \kb{E}^{\reg{S}} \nonumber \\
	&=\E_E
	\sum_{d}
	B_{0,d}^{\reg{MT_1R}}
	E^{\reg{MT_1}}
	\left(A_0\rho A_0^{\dagger} \otimes \kb{0^{n}}\right)
	E^{\dagger}
	B_{0,d}^{\dagger}\\
	&= \text{Eq.}\ \eqref{eq:encoding-2-ideal-accept}.
	\end{align}
	The difference in the approximation is bound by $\negl{n}$, since for each $b$ we
  can use Lemma~\ref{lem:GL(2n,F2)-twirl} (and there is an implicit average over
  the $b$s because of the normalization factor induced by the $B_{b,d}$
  operator). Essentially, the only way to pass the measurement test successfully is for $A$ not to alter the all-zero state in $T_1T_2$, and for $B$ to leave $T_2$ unaltered in the computational basis. This is reflected in the simulator's zero filter and \X filter, respectively.

	If the real protocol rejects, the \MPC stores a dummy $\bot$ in the key register $S$. The resulting state can be derived in a similar way, up to Equation~\eqref{eq:encoding-2-real-accept-reject}, after which the derivation becomes slightly different. The output state in the case of reject approximates (up to a difference of $\negl{n}$)
	\begin{align}
	\E_{E} &\sum_{b,d}
	B_{b,d}^{\reg{MT_1R}}
	E^{\reg{MT_1}}
	\Tr_{T_2}\left[
	\left(\Id - \Pi_{b,F}\right)^{\reg{T_1T_2}}
	A\left(\rho \otimes \kb{0^{2n}}\right)A^{\dagger}
	\left(\Id - \Pi_{b,F}\right)^{\dagger}
	\right]
	E^{\dagger}
	B_{b,d}^{\dagger}
	\otimes \kb{\bot}^{\reg{S}} \nonumber \\
	=\phantom{\E}\  &\sum_{b,d}
	B_{b,d}^{\reg{MT_1R}}
	\twirl{\Clifford_{n+1}}^{\reg{MT_1}}\left(
	\Tr_{T_2}\left[
	\left(\Id - \Pi_{b,F}\right)^{\reg{T_1T_2}}
	A\left(\rho \otimes \kb{0^{2n}}\right)A^{\dagger}
	\left(\Id - \Pi_{b,F}\right)^{\dagger}
	\right]
	\right)
	B_{b,d}^{\dagger}
	\otimes \kb{\bot}^{\reg{S}}\label{eq:encoding-2-reference-point-for-cnot}\\
	=\phantom{\E} \ &\sum_{\substack{b \neq 0^{n} \\ d, a, a'}} B_{b,d}^{\reg{MT_1R}}
	\twirl{\Clifford_{n+1}}^{\reg{MT_1}}\left(
	\Tr_{T_2}\left[
	A_{a} \rho^{\reg{MR}} A_{a'}^{\dagger} \otimes \ketbra{a}{a'}
	\right]
	\right)
	B_{b,d}^{\dagger}
	\otimes \kb{\bot}^{\reg{S}}\nonumber\\
	+ &\sum_{\substack{a,a' \neq 0^{2n}\\d}} B_{0,d}^{\reg{MT_1R}}
	\twirl{\Clifford_{n+1}}^{\reg{MT_1}}\left(
	\Tr_{T_2}\left[
	A_{a} \rho^{\reg{MR}} A_{a'}^{\dagger} \otimes \ketbra{a}{a'}
	\right]
	\right)
	B_{0,d}^{\dagger}
	\otimes \kb{\bot}^{\reg{S}} \nonumber\\
	= \phantom{\E} \ &\sum_{\substack{(b,a) \neq (0^{n},0^{2n})\\d}} B_{b,d}^{\reg{MT_1R}}
	\left(
	\Tr_M\left[A_{a}\rho^{\reg{MR}} A_{a}^{\dagger}\right] \otimes \tau^{\reg{MT_1}}
	\right)
	B_{b,d}^{\dagger}
	\otimes \kb{\bot}^{\reg{S}}\nonumber\\
	= \phantom{E} &\text{Eq.} \ \eqref{eq:encoding-2-ideal-reject-1} + \text{Eq.}
    \ \eqref{eq:encoding-2-ideal-reject-2}. \nonumber
	\end{align}
	Tracing out register $T_2$ ensures that the second half of $a$ and $a'$ have to be equal; Twirling over the Clifford group ensures that the first half (acting on register $T_1$) of $a$ and $a'$ have to be equal (see the proof of Lemma~\ref{lem:pauli-twirl-state}).

	These derivations show that the output state that the environment sees (in registers $MT_1RS$) in the real protocol are negligibly close to the output state in the ideal protocol. This concludes our proof for the second case, where player 1 is dishonest.
\end{proof}

\section{Proof of \Cref{lem:single-qubit-Clifford}}
\label{ap:proof-single-qubit-Clifford}
\begin{proof}
	For the sake of clarity, assume again that there is only one wire, held by player $1$ (who might be honest or dishonest). Generalizing the proof to multiple wires does not require any new technical ingredients, but simply requires a lot more (cluttering) notation.

	In the protocol $\RealProtocol^{G^{m_\ell}} \diamond \IdealProtocol^{C}$, an
  adversary $\advA$ receives a state $\rho^{\reg{MR}}$ from the environment
  (where again, $M := R^{\inreg}_1$). It potentially alters this state with a
  unitary map $A$, submits the result to the ideal functionality, and
  receives the register $MT_1 = R^{\outreg}_1$. The adversary may again act on
  the state, say with a map $B$, and then gets a chance to submit (for all
  players $i \in I_{\advA}$) bits $b_i$ to $\IdealProtocol^{C}$, and $b_i'$ to $\RealProtocol^{G^{m_\ell}}$). If one or more of those bits are 1,
  the ideal functionality (or the \MPC) aborts by overwriting the state register $S$ with $\bot$.

	In case all bits are 0, the output register $MT_1RS$ contains %
	\begin{align}
		&\E_E
    B^{\reg{MT_1R}}E^{\reg{MT_1}}\left(\mathcal{C}^{\reg{M}}\left(A\rho^{\reg{MR}}A^{\dagger}\right)
    \otimes \kb{0^n}^{\reg{T_1}}\right)E^{\dagger}B^{\dagger} \otimes
    \kb{E(\left(G^{m_\ell}\right)^{\dagger} \otimes \I^{\otimes n})}^{\reg{S}}
    \nonumber\\
		=&\E_E B^{\reg{MT_1R}}E^{\reg{MT_1}}\left(G^{m_\ell}\right)^{\reg{M}}\left(\mathcal{C}^{\reg{M}}\left(A\rho^{\reg{MR}}A^{\dagger}\right) \otimes \kb{0^n}^{\reg{T_1}}\right)\left(G^{m_\ell}\right)^{\dagger}E^{\dagger}B^{\dagger} \otimes \kb{E}^{\reg{S}},\label{eq:computation-single-clifford-accept}
\end{align}
	where $\mathcal{C}(\cdot)$ is the map induced by the circuit $C$.

	In case not all bits are 0, the output register $MT_1RS$ contains
	\begin{align}
	B'A'\rho_R(B'A')^{\dagger} \otimes \tau^{\reg{MT_1}} \otimes \kb{\bot}^{\reg{S}},\label{eq:computation-single-clifford-reject}
	\end{align}
	where $A'$ and $B'$ are the reduced maps $A$ and $B$ on register $R$.

	Define a simulator $\simS$ as follows:
	\begin{simulator} On input $\rho^{\reg{MR}}$ from the environment, do:
		\begin{itemize}
			\item Run $A$ on $MR$.
			\item Submit $M$ to the ideal functionality for $G^{m_\ell} \circ C$, and receive $MT_1$.
			\item Run $B$ on $MT_1R$, and note its output bits $(b_i,b'_i)$ for all $i \in I_A$. Submit $\max\{b_i,b_i'\}$ to the ideal functionality for $G^{m_\ell}\circ C$.
		\end{itemize}
	\end{simulator}
    From the point of view of the adversary, the state it receives from the ideal functionality is the same: a Clifford-encoded state. Thus, the bits $b_i$ and $b_i'$ will not be different in this simulated scenario. In fact, the output state is exactly Eq.~\eqref{eq:computation-single-clifford-accept} + Eq.~\eqref{eq:computation-single-clifford-reject}.
\end{proof}

\section{Proof of \Cref{lem:cnot}}
\label{ap:cnot}
\begin{proof}
	There are four different cases, for which we construct simulators separately: both players involved in the \CNOT are honest ($i,j \not\in I_{\advA}$), both players are dishonest ($i,j \in I_{\advA}$), only player $i$ is honest ($i \not\in I_{\advA}$, $j \in I_{\advA}$), or only player $j$ is honest ($i \in I_{\advA}$, $j \not\in I_{\advA}$). Without loss of generality, we will assume that all other players are dishonest (except in the second case, where at least one of the other players has to be honest), and that they have no inputs themselves: their encoded inputs can be regarded as part of the adversary's side information $R$. Note that these four cases also cover the possibility that $i = j$.

	\paragraph{Case 1: player $i$ and $j$ are honest.} In this case, the adversarial players in $I_{\advA}$ only have influence on the execution of step~\ref{step:cnot-other-players} of the protocol, where the state is sent around in order for the players to jointly apply the random Clifford $D$.

	As in the first case of the proof of Lemma~\ref{lem:input-encoding} for the encoding protocol $\RealProtocol^{\Enc}$ (where the encoding player is honest), define a simulator that performs a Pauli filter $\idfilter{}$ on the attack of the adversary. The simulator and proof are almost identical to those in Lemma~\ref{lem:input-encoding}, so we omit the details here.

	\paragraph{Case 2: player $i$ and $j$ are dishonest.} Without loss of
  generality, we can break up the attack of the adversary (acting jointly for
  players $i$, $j$ and any other players in $I_{\advA}$) into three unitary
  operations, acting on the relevant register plus a side-information register. As in the proof of Lemma~\ref{lem:input-encoding}, we may assume that the honest actions are executed as well, since each attack may start or end with undoing that honest action. The first attack $A^{\reg{M^{ij}R}}$ is executed on the plaintexts, before any protocol starts. The second attack $\tilde{A}^{\reg{M^{ij}T^{ij}_{12}R}}$ happens after step~\ref{step:cnot-initialize-t2} of Protocol~\ref{protocol:CNOT}, on the output of $\IdealProtocol^{C}$ and the initialized registers $T_2^{ij}$. Finally, the third attack $\dbtilde{A}^{\reg{M^{ij}T^{ij}_{12}R}}$ happens toward the end of the protocol, right before the $T_2^{ij}$ registers are measured in step~\ref{step:cnot-measurement} of Protocol~\ref{protocol:CNOT}. Note that $\dbtilde{A}$ may depend on the instructions $V$, $W_i$ and $W_j$.

	It will be useful to decompose the second and third attacks as follows:
	\begin{align}
	\tilde{A} &= \sum_{a^i_1,a^i_2,a^j_1,a^j_2, c^i_1,c^i_2,c^j_1,c^j_2} (\X^{a^i_1} \Z^{c^i_1})^{\reg{M^iT^i_1}} \otimes (\X^{a^i_2} \Z^{c^i_2})^{\reg{T^i_2}} \otimes (\X^{a^j_1} \Z^{c^j_1})^{\reg{M^jT^j_1}} \otimes (\X^{a^j_2} \Z^{c^j_2})^{\reg{T^j_2}} \otimes \tilde{A}_{a^{ij}_{12}, c^{ij}_{12}}^{\reg{R}}\label{eq:decomposition-A-tilde}\\
	\dbtilde{A} &= \sum_{b,d} (\X^b \Z^d)^{\reg{T_2^{ij}}} \otimes \dbtilde{A}_{b,d}^{\reg{M^{ij}T_1^{ij}R}}\label{eq:decomposition-A-double-tilde}
	\end{align}
	Whenever the order is clear from the context, we will abbreviate, for example, $a_{12}^{ij}$ for the concatenation $a^i_1a^i_2a^j_1a^j_2$, as we have done in the last term of Equation~\eqref{eq:decomposition-A-tilde}.

	In terms of an arbitrary attack $A, \tilde{A}, \dbtilde{A}$, define the simulator $\simS$ as follows:
	\begin{simulator}
		On input $\rho^{\reg{M^iM^jR}}$ from the environment, do:
		\begin{enumerate}
			\item Initialize $b_i = 0$.
			\item Run $A$ on $M^{ij}R$. %
			\item Submit $M^{ij}$ to $\IdealProtocol^{\CNOT^{m_{\ell}} \circ C}$, and receive $M^{ij}T_1^{ij}$, containing an encodings of the $M^i$ and $M^j$ registers of $\CNOT^{m_{\ell}}(C(\rho))$, under some (secret) keys $E_i, E_j$.
			\item Run $\zerofilter{T^{ij}_2}(\idfilter{M^{ij}T^{ij}_{1}}(\tilde{A}))$ on $R$ (see \Cref{sec:pauli-filter}). If one of the filter flags is 1, set $ b_i=1$.
			\item Sample random $V' \in \Clifford_{4n+2}$ and $W_i', W_j' \in \Clifford_{2n+1}$, and run $\xfilter{T_{2}^{ij}}(\dbtilde{A})$ on $M^{ij}T_1^{ij}R$, where $\dbtilde{A}$ may depend on $V', W_i', W_j'$. If the filter flag is 1, set $b_i = 1$. 
			\item Submit $b_i$ to the ideal functionality, along with all other $b_{\ell} = 0$ for $\ell \in I_{\advA}\setminus\{b_i\}$.
		\end{enumerate}
	\end{simulator}

	The simulator should also abort whenever the adversary signals abort during an interaction with \MPC. For simplicity, we leave out these abort bits in the simulator and proof. They are dealt with in the same way as in the proof of Lemma~\ref{lem:single-qubit-Clifford}.

	As before, we derive the real and ideal output states in the registers $R_i^{\outreg} = M^iT_1^i$ and $R_j^{\outreg} = M^jT^j_i$, the state register $S$, and the attacker's side information $R$, and aim to show that they are negligibly close in terms of the security parameter $n$.

	In the ideal (simulated) case, there are two points at which cheating may be detected by the simulator: once during the zero/identity filter of $\tilde{A}$, and during the \X filter of $\dbtilde{A}$. Thus, there are three possible outcome scenarios: both tests are passed, the first test is passed but the second is not, or the first test fails (in which case it does not matter whether the second test is passed or not).
	
	If both tests pass, then by three applications of~\Cref{lem:pauli-filter}, the simulated output state is
	\begin{align}
	\E_{E_i,E_j} \dbtilde{A}_{0}^{\reg{M^{ij}T_1^{ij}R}}
	\tilde{A}_{0,0}^{\reg{R}}
	&\left(E_i \otimes E_j\right)^{\reg{M^{ij}T_1^{ij}}}
	\left(
	\CNOT^{m_{\ell}}
	\mathcal{C} \left(
	A\rho A^{\dagger}
	\right)
	\CNOT^{m_{\ell}\dagger}
	\otimes \kb{0^{2n}}^{\reg{T_1^{ij}}}
	\right)\nonumber\\
	&\left(E_i \otimes E_j\right)^{\dagger}
	\tilde{A}_{0,0}^{\dagger}\dbtilde{A}^{\dagger}_{0}
	\otimes \kb{E_i,E_j}^{\reg{S}},\label{eq:cnot-ideal-accept}
	\end{align}
	where we write $\tilde{A}_{0,0}$ to denote the attack $\sum_{c_2^{ij}} \tilde{A}_{0000,0c_2^i0c_2^j}$ that passes through the zero/identity filter, and $\dbtilde{A}_0$ to denote the attack $\sum_d \dbtilde{A}_{0,d}$ that passes through the \X filter.
	
	If the first test is passed but the second test is not, then the storage register $S$ gets erased, so that we may view the $E_i$ and $E_j$ operations as Clifford twirls of the registers they encode. In that case, the (simulated) output state is
	\begin{align}
	&\sum_{b \neq 0} \dbtilde{A}_{b}%
	\tilde{A}_{0,0}%
	\twirl{\Clifford_{n+1}}^{\reg{M^iT_1^i}}\left(
	\twirl{\Clifford_{n+1}}^{\reg{M^jT_1^j}}\left(
	\CNOT^{m_{\ell}}
	\mathcal{C} \left(
	A\rho A^{\dagger}
	\right)
	\CNOT^{m_{\ell}\dagger}
	\otimes \kb{0^{2n}}%
	\right)\right)
	\tilde{A}_{0,0}^{\dagger}
	\dbtilde{A}^{\dagger}_{b} \otimes \kb{\bot}^{\reg{S}}\\
	= &\sum_{b \neq 0} \dbtilde{A}_{b}%
	\tilde{A}_{0,0}%
    \left(\Tr_{M^{ij}}\left[A\rho^{\reg{M^{ij}R}} A^{\dagger}\right]  \otimes \tau^{\reg{MT^{ij}_1}} \right)
	\tilde{A}_{0,0}^{\dagger}
	\dbtilde{A}^{\dagger}_{b} \otimes \kb{\bot}^{\reg{S}}.\label{eq:cnot-ideal-reject-1}
	\end{align}
	The Clifford twirls cause the data and trap registers to become fully mixed, thereby also nullifying the effect of the \CNOT and circuit $C$ on the data.
	
	Finally, we consider the third scenario, where already the first test (the zero / identity filter) fails. As in the previous scenario, the storage register $S$ is erased, allowing us to apply the Clifford twirl again. By~\Cref{lem:pauli-filter}, the output state in this case is
	\begin{align}
	\sum_b \sum_{\substack{(a^{ij}_{12},c^{ij}_1) \neq \\(0^{4n+2},0^{2n+2})}}
	\dbtilde{A}_b%
	\tilde{A}_{a^{ij}_{12},c^{ij}_1}%
	\left(\Tr_{M^{ij}}\left[A\rho^{\reg{M^{ij}R}} A^{\dagger}\right]  \otimes \tau^{\reg{MT^{ij}_1}} \right)
	\tilde{A}_{a^{ij}_{12},c^{ij}_1}^{\dagger}
	\dbtilde{A}_b^{\dagger} \otimes \kb{\bot}^{\reg{S}},\label{eq:cnot-ideal-reject-2}
	\end{align}
	writing $\tilde{A}_{a_{12}^{ij},c_1^{ij}} := \sum_{c_2^{ij}} \tilde{A}_{a_{12}^{ij},c_{12}^{ij}}$. Note that for the second test (the \X filter), the terms for both possible flag values remain: the cheating bit $b_i$ is already set to 1, regardless of the outcome of this second test.

	We move on to the analysis of the real protocol
  $\RealProtocol^{\CNOT^{m_{\ell}}} \diamond \IdealProtocol^C$, and aim to show
  that the output state is equal to Eq.~\eqref{eq:cnot-ideal-accept} +
  Eq.~\eqref{eq:cnot-ideal-reject-1} + Eq.~\eqref{eq:cnot-ideal-reject-2}. To do
  so, consider the output state of the real protocol, right before the final
  measurement.

  We continue to argue why the attacks are independent of $E_{ij}$, $E'_{ij}$,
  $g_{ij}$, $r_{ij}$ and $s_{ij}$. The intuition for this fact is that $D$ is
  uniformly random and independent of $E_{ij}$ from the perspective of the
  adversary. Therefore it ``hides" all other information that is used to compile $V$, including $F_{ij}$. Therefore $F_{ij}$ are as random and independent as $D$ from the perspective of the adversary, i.e. given $V$. This allows for a similar argument for the Cliffords $W_{ij}$, where now $F$ hides all the other information, i.e. $E'_{ij}, g_{ij}, r_{ij}$ and $s_{ij}$.

  For the following more formal argument, we treat all the mentioned quantities as random variables. Initially, $E_{ij}$ are uniformly random. $D$ is the product of a number of Clifford group elements, at least one of which is generated honestly and therefore sampled uniformly at random. But for any group $G$, given two independent random variables $\zeta$ and $\eta$ on $G$, where $\zeta$ is uniformly random, we have that $\eta\zeta$ is uniformly random and $\eta\zeta\bot \eta$, where $\bot$ denotes independence. This implies that $D$ is indeed a uniformly random Clifford itself. Using the same argument, $V$ is uniformly random and $V\bot (E_{ij}, F_{ij})$. The quantities $E'_{ij}$, $g_{ij}$, $r_{ij}$ and $s_{ij}$ are sampled independently and uniformly after $V$ is handed to player $i$, so we even have $V\bot (E_{ij}, F_{ij},E'_{ij}, g_{ij}, r_{ij}, s_{ij})$. After step 4. in Protocol \ref{protocol:CNOT}, the adversary has a description of $V$, so when analyzing $W_{ij}$, we have to derive independence statements \emph{given $V$}. But as shown before $F_{ij}$ are independent of $V$, so the the group random variable property above we have $W_{ij}\bot (E'_{ij}, g_{ij}, r_{ij}, s_{ij})|V$. Clearly, $E_{ij}$ is independent of all the random variables used in $W_{ij}$, and we have shown that $E_{ij}\bot V$, so $W_{ij}\bot (E_{ij}, E'_{ij}, g_{ij}, r_{ij}, s_{ij})|V$. In summary, we have
  \begin{equation}\label{eq:indy}
  	(V, W_{ij})\bot (E_{ij}, E'_{ij}, g_{ij}, r_{ij}, s_{ij}).
  \end{equation}
  According to the decomposition of the attack into attack maps $A, \tilde A$ and $\dbtilde A$, that we made without loss of generality, the Clifford operations $F_i$, $F_j$, and $D$ cancel again after having fulfilled their task of hiding information, which allows us to utilize Equation \eqref{eq:indy} to carry out the expectation values over various variables from the right hand side of that equation.

The output state of the real protocol is
	\begin{align}
	\E_{\substack{E_i', E_j', g_i, g_j \\ r_i, s_i, r_j, s_j}}
	\dbtilde{A}^{\reg{M^{ij}T^{ij}_{12}R}}\left(E_i' \otimes \left(\X^{r_i}\Z^{s_i}\right)^{\reg{T^i_2}} \otimes E_j' \otimes \left(\X^{r_j}\Z^{s_j}\right)^{\reg{T^j_2}}\right)
	\left(U_{g_i}^{\reg{T_{12}^i}} \otimes U_{g_j}^{\reg{T_{12}^j}}\right)
	\CNOT^{m_{\ell}}
	 \sigma \ \CNOT^{m_{\ell}\dagger}\nonumber\\
	\left(U_{g_i}^{\dagger} \otimes U_{g_j}^{\dagger}\right)
	\left(E_i' \otimes \left(\X^{r_i}\Z^{s_i}\right) \otimes E_j' \otimes \left(\X^{r_j}\Z^{s_j}\right)\right)^{\dagger}
	\dbtilde{A}^{\dagger}
	\otimes \kb{E_i',E_j'}^{\reg{S}},\label{eq:cnot-real-1}
	\end{align}

	where (again writing $\mathcal{C}(\cdot)$ for the map induced by the circuit $C$)

	\begin{align}
	\sigma &:= \E_{E_i, E_j \in \Clifford_{n+1}}
	\left(E_i^{\dagger} \otimes E_j^{\dagger}\right)
	\tilde{A}^{\reg{M^{ij}T_{12}^{ij}R}}
	\Bigg(
	\left(E_i^{\reg{M^iT_1^i}} \otimes E_j^{\reg{M^jT_1^j}}\right)
	\left(
	\mathcal{C}\left(A\rho^{\reg{M^{ij}R}}A^{\dagger}\right)
	\otimes
	\kb{0^{2n}}^{\reg{T_1^{ij}}}
	\right)\nonumber\\
	&\qquad \qquad \left(E_i^{\dagger} \otimes E_j^{\dagger}\right)
	\otimes \kb{0^{2n}}^{\reg{T_2^{ij}}}
	\Bigg)
	\tilde{A}^{\dagger}
	\left(E_i \otimes E_j\right)\nonumber\\
	&=\Twirl{\Clifford_{n+1}}^{\reg{M^iT^i_{1}}}\left(
	  \Twirl{\Clifford_{n+1}}^{\reg{M^jT^j_{1}}}\left(
		  \tilde{A}
	  \right)\right)
	  \left(
	  \mathcal{C}\left(A\rho^{\reg{M^{ij}R}}A^{\dagger}\right)
	  \otimes \kb{0^{4n}}^{\reg{T_{12}^{ij}}}
	  \right)\nonumber\\
	&\approx_{\negl{n}} \tilde{A}_{00}^{\reg{T_2^{ij}R}}
	\left(
	\mathcal{C}\left(A\rho^{\reg{M^{ij}R}}A^{\dagger}\right)
	\otimes \kb{0^{2n}}^{\reg{T_{2}^{ij}}}
	\right) \tilde{A}_{00}^{\dagger} \otimes \kb{0^{2n}}^{\reg{T_1^{ij}}} \nonumber \\
	&\quad + \Tr_{M^i} \left[
	\tilde{A}_{01}^{\reg{T_2^{ij}R}}
	\left(
	\mathcal{C}\left(A\rho^{\reg{M^{ij}R}}A^{\dagger}\right)
	\otimes \kb{0^{2n}}^{\reg{T_{2}^{ij}}}
	\right) \tilde{A}_{01}^{\dagger}
	\right] \otimes \tau^{\reg{M^iT_1^i}} \otimes \kb{0^n}^{\reg{T_1^j}}\nonumber \\
	&\quad + \Tr_{M^j} \left[
	\tilde{A}_{1 0}^{\reg{T_2^{ij}R}}
	\left(
	\mathcal{C}\left(A\rho^{\reg{M^{ij}R}}A^{\dagger}\right)
	\otimes \kb{0^{2n}}^{\reg{T_{2}^{ij}}}
	\right) \tilde{A}_{1 0}^{\dagger}
	\right] \otimes \kb{0^n}^{\reg{T_1^i}} \otimes \tau^{\reg{M^jT_1^j}}\nonumber \\
	&\quad + \Tr_{M^{ij}} \left[
	\tilde{A}_{1 1}^{\reg{T_2^{ij}R}}
	\left(
	\mathcal{C}\left(A\rho^{\reg{M^{ij}R}}A^{\dagger}\right)
	\otimes \kb{0^{2n}}^{\reg{T_{2}^{ij}}}
	\right) \tilde{A}_{1 1}^{\dagger}
	\right] \otimes \tau^{\reg{M^{ij}T_1^{ij}}},\label{eq:cnot-sigma-approximate-form}
	\end{align}
	and
	\begin{align*}
	\tilde{A}_{pq} &:= \sum_{\substack{a^{ij}_2, c^{ij}_2 \\ a_1^ic_1^i \in S_p \\ a_1^jc_1^j \in S_q}} (\X^{a_2^{ij}}\Z^{c_2^{ij}})^{\reg{T_2^{ij}}} \otimes \tilde{A}^{\reg{R}}_{a_{12}^ia_{12}^j,c_{12}^ic^j_{12}}.
	\end{align*}
	for $p,q \in \{0,1\}$ and $S_0 := \{0^{2n+2}\}$, $S_{1} := \{0,1\}^{2n+2} \setminus S_0$.
	The approximation follows by a double application of Lemma~\ref{lem:clifford-code-property}. We can twirl with the keys $E_i$ and $E_j$, since none of the attacks can depend on the secret encoding keys $E_i, E_j$, and the keys have been removed from the storage register $S$, and replaced by the new keys $E_i', E_j'$.

	Having rewritten the state $\sigma$ in this form, we consider the state in Equation~\eqref{eq:cnot-real-1} after the $T_2^{ij}$ registers are measured in the computational basis, as in step~\ref{step:cnot-measurement} of Protocol~\ref{protocol:CNOT}. We first consider the case where the measurement outcome is accepted by the MPC (i.e., the measurement outcome is $r_ir_j$). Using the same derivation steps as in Equations~\eqref{eq:encoding-2-real-accept-first-step}--\eqref{eq:encoding-2-reference-point-for-cnot}, we see that the real accept state approximates (up to a negligible error in $n$)
	\begin{align}
	\E_{E_i', E_j' \in \Clifford_{n+1}}
	\dbtilde{A}^{\reg{M^{ij}T^{ij}_{1}R}}_0\left(E_i' \otimes E_j'\right)^{\reg{M^{ij}T^{ij}_1}}
	\Tr_{T_2^{ij}} \left[
	\kb{0^{4n}}^{\reg{T^{ij}_{12}}}
	\CNOT^{m_{\ell}}
	\sigma
	\CNOT^{m_{\ell}\dagger}
	\kb{0^{4n}}
	\right]\nonumber\\
	\qquad\qquad
	\left(E_i' \otimes E_j'\right)^{\dagger}
	\dbtilde{A}_0^{\dagger}
	\otimes \kb{E_i',E_j'}^{\reg{S}}.\label{eq:cnot-real-accept}
	\end{align}
	To derive the above expression, we applied a Pauli twirl, which relies on the
  fact that the adversary cannot learn $r_i, r_j, s_i, s_j$. Furthermore, the
  derivation contains an application of Lemma~\ref{lem:GL(2n,F2)-twirl} to
  expand the effect of measuring $T_2^{ij}$ to measuring both registers
  $T_{12}^{ij}$. To apply this lemma, we use the aforementioned fact that $g_i$ and $g_j$ remain hidden from the adversary.

	The second, third, and fourth terms of the sum in the approximation of $\sigma$ (see Equation~\eqref{eq:cnot-sigma-approximate-form}) have negligible weight inside Equation~\eqref{eq:cnot-real-accept}, since the probability of measuring an all-zero string in the $T_1^{ij}$ registers is negligible in $n$ whenever one or both are in the fully mixed state $\tau$. Additionally, the only components in $\tilde{A}_{00}$ that survive are those that act trivially in the computational basis on $T_2^{ij}$. Hence,
	\begin{align*}
	\text{Eq.} \ \eqref{eq:cnot-real-accept} \approx_{\negl{n}} \text{Eq.} \ \eqref{eq:cnot-ideal-accept}.
	\end{align*}
	In case the measurement outcome is rejected by the MPC (i.e., it is anything other than $r_ir_j$), the output state can be derived using the same steps that were used to obtain Equation~\eqref{eq:encoding-2-reference-point-for-cnot} in the proof of Lemma~\ref{lem:input-encoding}. Up to an error negligible in $n$, it approximates
	\begin{align*}
	\sum_b \dbtilde{A}^{\reg{M^{ij}T^{ij}_{1}R}}_b\twirl{\Clifford_{n+1}}^{\reg{M^iT_1^i}}\Big(
	\twirl{\Clifford_{n+1}}^{\reg{M^jT_1^j}}\left(
	\Tr_{T_2^{ij}} \left[
	(\Id - \Pi_{b,F})^{\reg{T^{ij}_{12}}}
	\CNOT^{m_{\ell}}
	\sigma
	\CNOT^{m_{\ell}\dagger}
	(\Id - \Pi_{b,F})^{\dagger}
	\right]
	\right)
	\Big)
	\dbtilde{A}_b^{\dagger}
	\otimes \kb{\bot}^{\reg{S}}.
	\end{align*}
	The encoding under the keys $E_i', E_j'$ in Equation~\eqref{eq:cnot-real-1}
  can be regarded as two Clifford twirls, because these keys are removed from
  the storage register $S$, and because the attack maps also cannot depend on
  them, since they are unknown by the adversary.

	The next step is to substitute the expression for $\sigma$ that was derived in Equation~\eqref{eq:cnot-sigma-approximate-form}. We distinguish between the case $b \neq 0$, where $\Id - \Pi_{b,F} = \Id$ and thus all terms of Equation~\eqref{eq:cnot-sigma-approximate-form} remain, and the case $b = 0$, where one has to more carefully count which (parts of the) terms remain. To do so, observe that the first term is projected to non-zero in $T_{12}^{ij}$ whenever $a_{2}{^{ij}}$ is nonzero. The other three terms are always projected to non-zero, up to a negligible contribution of the all-zero string in the fully mixed state $\tau$. In summary, exactly those terms $\tilde{A}^{\reg{R}}_{a^{ij}_{12},c^{ij}_{12}}$ remain for which $(a^{ij}_{12},c^{ij}_{1}) \neq (0^{4n+2},0^{2n+1})$.

	Because the two Clifford twirls map the $M^{ij}$ registers to a fully mixed state, the four terms in Equation~\eqref{eq:cnot-sigma-approximate-form} can be combined, resulting in the following output state in the reject case %
	\begin{align*}
	\sum_{b \neq 0} \sum_{a_{12}^{ij},c_1^{ij}} &\dbtilde{A}_{b} \tilde{A}_{a_{12}^{ij},c_1^{ij}} \left(\Tr_{M^{ij}}\left[A\rho^{\reg{M^{ij}R}} A^{\dagger}\right]  \otimes \tau^{\reg{MT^{ij}_1}} \right)\tilde{A}_{a_{12}^{ij},c_1^{ij}}^{\dagger} \dbtilde{A}_b^{\dagger} \otimes \kb{\bot}^{\reg{S}}\nonumber\\
	+  \sum_{\substack{(a_{12}^{ij},c_1^{ij}) \neq \\(0^{4n+2},0^{2n+2})}} &\dbtilde{A}_{0} \tilde{A}_{a_{12}^{ij},c_1^{ij}} \left(\Tr_{M^{ij}}\left[A\rho^{\reg{M^{ij}R}} A^{\dagger}\right]  \otimes \tau^{\reg{MT^{ij}_1}} \right)\tilde{A}_{a_{12}^{ij},c_1^{ij}}^{\dagger} \dbtilde{A}_0^{\dagger} \otimes \kb{\bot}^{\reg{S}}\\
	= \ &\text{Eq.} \ \eqref{eq:cnot-ideal-reject-1} + \text{Eq.} \ \eqref{eq:cnot-ideal-reject-2}.
	\end{align*}
	We have shown that the sum of the three terms of the output state in the simulated case (both tests accept, the first test accepts but the second rejects, and the first test rejects) is approximately equal to the sum of the two terms of the output state in the real case (the MPC accepts the measurement outcome, or the MPC rejects the measurement outcome).
	
	\paragraph{Case 3: only player $i$ is honest.} At first, it may seem that this is just a special case of the previous one, where both players are dishonest. While this is true in spirit, we cannot directly use the simulator from the previous case. The reason is syntactical: a simulator would not have access to the registers $M^iT^i_{12}$, because they are held by honest player $i$. Thus, the simulator needs to differ slightly from the previous case. However, it is very similar, as is the derivation of the real/ideal output states. We therefore omit the full proof, and instead only define the simulator.
	
	The adversary again has three opportunities to attack: an attack $A$ on the plaintext and side-information register $M^jR$, which happens before the ideal functionality $\IdealProtocol^C$ is called; an attack $\tilde{A}$ on the output of $\IdealProtocol^{C}$ in registers $M^jT_1^jR$ (right before player $j$ sends their state to player $i$); and an attack $\dbtilde{A}$ on $M^jT^j_{12}R$, after an honest application of $W_j$ (which we may assume to happen without loss of generality), but before the computational-basis measurement of $T_2$. Given these attacks, define the simulator as follows.
	
	\begin{simulator}
		On input $\rho^{\reg{M^jR}}$ from the environment, do:
		\begin{enumerate}
			\item Initialize $b_j = 0$.
			\item Run $A$ on $M^jR$.
			\item Submit $M^j$ to the ideal functionality $\IdealProtocol^{\CNOT^{m_{\ell}} \circ C}$, and receive $M^jT_1^j$, containing an encoding under a secret key $E_j$. (Honest player $i$ holds the other output, encoded under $E_i$.)
			\item Run $\idfilter{M^jT_1^j}(\tilde{A})$ on $R$. If the filter flag is 1, then set $b_j = 1$.\label{step:cnot-case-3-filter-1}
			\item Sample random $W_j' \in \Clifford_{2n+1}$, and run $\xfilter{T_2^j}(\dbtilde{A})$ on $M^jT_1^jR$, where $\dbtilde{A}$ may depend on $W_j'$. If the filter flag is 1, then set $b_j = 1$.\label{step:cnot-case-3-filter-2}
			\item Submit $b_j$ to the ideal functionality, along with all other $b_{\ell} = 0$ for $\ell \in I_{\advA} \setminus \{b_j\}$.
		\end{enumerate}
	\end{simulator}
	Intuitively, the simulator tests whether player $j$ sent the actual outcome of $\IdealProtocol^C$ without altering it (step~\ref{step:cnot-case-3-filter-1} of the simulator), and whether player $j$ left the computational basis of $T_2$ invariant before measuring it (step~\ref{step:cnot-case-3-filter-2} of the simulator).

	\paragraph{Case 4: only player $j$ is honest.} Similarly to the previous case, we need to provide a separate simulator for the case where player $i$ is dishonest, player $j$ is honest, and (without loss of generality) all other players are dishonest. 
	
	The adversary has three opportunities to attack: an attack $A$ on the plaintext and side-information register $M^iR$, which happens before the ideal functionality $\IdealProtocol^C$ is called; an attack $\tilde{A}$ on registers $M^{ij}T^{ij}_{12}R$ that is applied on the outputs of the ideal functionality and on the extra registers $T_2$, right before $D$ is applied; and an attack $\dbtilde{A}$ on $M^{ij}T_{12}^{ij}R$, right before the measurement on $T_2^i$ (as part of player $i$'s test) and the application of $W_j$ (so right before sending the appropriate registers to player $j$). Given these attacks, define the simulator as follows.
	
	\begin{simulator}
		On input $\rho^{\reg{M^iR}}$ from the environment, do:
		\begin{enumerate}
			\item Initialize $b_i = 0$.
			\item Run $A$ on $M^iR$.
			\item Submit $M^i$ to the ideal functionality $\IdealProtocol^{\CNOT^{m_{\ell}} \circ C}$, and receive $M^iT_1^i$, containing an encoding under a secret key $E_i$. (Honest player $j$ holds the other output, encoded under $E_j$.)
			\item Run $\zerofilter{T_2^{ij}}\left(\idfilter{M^{ij}T_1^{ij}}(\tilde{A})\right)$ on $R$. If the filter flag is 1, then set $b_i = 1$.\label{step:cnot-case-4-filter-1}
			\item Sample random $V,W_i' \in \Clifford_{2n+1}$, and run $\xfilter{T_2^i}\left(\idfilter{M^jT_{12}^j}(\dbtilde{A})\right)$ on $M^iT_1^jR$, where $\dbtilde{A}$ may depend on $V$ and $W_i'$. If the filter flag is 1, then set $b_i = 1$.\label{step:cnot-case-4-filter-2}
			\item Submit $b_i$ to the ideal functionality, along with all other $b_{\ell} = 0$ for $\ell \in I_{\advA} \setminus \{b_i\}$.
		\end{enumerate}
	\end{simulator}
	Intuitively, the simulator tests (in step~\ref{step:cnot-case-4-filter-1}) whether player $i$ leaves the states received from the ideal functionality and player $j$ intact, as well as the traps in $T_{2}^{ij}$ that are initialized to $\kb{0^{2n}}$. In step~\ref{step:cnot-case-4-filter-2}, it tests both whether player $i$ executes the test honestly by not altering the computational-basis value of $T_2^i$, and whether he would give the correct (uncorrupted) state to player $j$.

\end{proof}

\section{Proof of \Cref{lem:measurement}}
\label{ap:measurement}

The following lemma captures the fact that $\CNOT_{1,c}$ makes it hard to alter the outcome of a computational-basis measurement with a (Pauli) attack $\X^b$ if $b$ does not depend on $c$. We will use this lemma later in the security proof of the measurement protocol.

\begin{lemma}\label{lem:measurement-cnot-trick}
	Let $m \in \{0,1\}$, and let $\rho$ be a single-qubit state. Let $p:\{0,1\}^{n+1}\to[0,1]$ be a probability distribution, and $u_n$ the uniform distribution on $\{0,1\}^n$. Then
	\[
	\left\|
	\E_{\substack{b \sim p \\ c \sim u_n}}
	\bra{m,m\cdot c}
	\X^b
	\CNOT_{1,c}
	\left(
	\rho \otimes \kb{0^n}
	\right)
	\CNOT_{1,c}^{\dagger}
	\X^b
	\ket{m,m \cdot c}
	-
	p(0^{n+1})
	\bra{m}\rho\ket{m}
	\right\|_1 \leq 2^{-n}.
	\]
\end{lemma}

\begin{proof}
	By commutation relations between $\CNOT$ and $\X$, we have that for all $b$ and $c$,
	\begin{align*}
	\X^b \CNOT_{1,c} = \CNOT_{1,c}\X^{b \oplus (0,b_1 \cdot c)},
	\end{align*}
	where $b_1$ denotes the first bit of $b$. Furthermore, $\CNOT_{1,c}\ket{m,m \cdot c} = \ket{m,0^n}$. Using these two equalities, we have
	\begin{align}
	&\E_{\substack{b \sim p \\ c \sim u_{n}}}
	\bra{m,m\cdot c}
	\X^b
	\CNOT_{1,c}
	\left(
	\rho \otimes \kb{0^n}
	\right)
	\CNOT_{1,c}^{\dagger}
	\X^b
	\ket{m,m \cdot c} \nonumber\\
	=&\E_{\substack{b \sim p \\ c \sim u_{n}}}
	\bra{m,0^n}
	\X^{b \oplus(0,b_1 \cdot c)}
	\left(
	\rho \otimes \kb{0^n}
	\right)
	\X^{b \oplus(0,b_1 \cdot c)}
	\ket{m,0^n}.\label{eq:measurement-lemma-1}
	\end{align}
	Let us consider which values of $b$ result in a non-zero term. In order for the last $n$ qubits to be in the $\kb{0^n}$ state after $\X^{b \oplus (0,b_1 \cdot c)}$, it is necessary that $b \oplus (0,b_1 \cdot c) \in \{(0,0^{n}),(1,0^n)\}$. By considering the two possible cases $b_1 = 0$ and $b_1 = 1$, we see that the only two values of $b$ for which this is the case are $b = (0,0^n)$ and $b = (1,c)$. Thus Equation~\eqref{eq:measurement-lemma-1} equals
	\begin{align*}
	\E_{c \sim u_n}
	&p(0^{n+1})
	\bra{m,0^n}
	\left(
	\rho \otimes \kb{0^n}
	\right)
	\ket{m,0^n} &+ &p(c) \bra{m,0^n}
	\X^{1,0^n}
	\left(
	\rho \otimes \kb{0^n}
	\right)
	\X^{1,0^n}
	\ket{m,0^n}\\
	= \qquad
	&p(0^{n+1})
	\bra{m}
	\rho
	\ket{m} &+
	&\E_{c \sim u_n}
	p(c) \bra{m+1}
	\rho
	\ket{m+1}\\
	\approx_{2^{-n}} \quad \ \
	&p(0^{n+1})
	\bra{m}
	\rho
	\ket{m}.
	\end{align*}
	The last step follows from the fact that $\E_c p(c) = 2^{-n}$.
\end{proof}

We now move on to proving the security of Protocol~\ref{protocol:measurement} by showing that its outcome resembles that of the ideal functionality.

\begin{proof}[Proof of Lemma~\ref{lem:measurement}]
	Let player $i$ be the player holding (the encoding of) the state in wire $w$
  (assume, for simplicity, that $w$ is the only wire in the computation). If
  player $i$ is honest, then it is simple to check that the outcome is correct:
  the unitary $V$ is designed so that, whatever the first (data) qubit collapses
  to, all other qubits that appear in $s$ measure to the same value. In step 4,
  the \MPC checks that this is indeed the case, and stores the measured value
	in the state register%
	.

	For the rest of this proof, we will assume that player $i$ is dishonest. The other players do not play a role, except for their power to abort the ideal functionalities and/or \MPC. We do not fix which players in $[k] \setminus \{i\}$ are honest: as long as at least one of them is, the encoding key $E$ will be unknown to the adversary.

	In an execution of $\RealProtocol^{\smeas{}} \diamond \IdealProtocol^{C}$, an adversary has two opportunities to influence the outcome: before and after interacting with the ideal functionality for $C$. Before the adversary submits the register $M = R_i^{\inreg}$ to $\IdealProtocol^{C}$, it applies an arbitrary attack unitary $A$ to the register $MR$ it receives from the environment. (Recall that $R$ is a side-information register.) Afterwards, it can act on $MT_1 = R_i^{\outreg}$ and $R$, and produces two bits ($b_i$ to signal cheating to $\IdealProtocol^{C}$, and $b_i'$ to signal cheating to the \MPC which is part of $\RealProtocol^{\smeas{}}$), plus a bit string. We may assume, without loss of generality, that the adversary first applies the honest unitary $V$, followed by an arbitrary (unitary) attack $B$ and subsequently by an honest computational-basis measurement of the registers $MT_1$.

	For any adversary, specified by the unitaries $A$ and $B$, define a simulator $\simS$ as follows:
	\begin{simulator}
		On input $\rho^{MR}$ from the environment, do:
		\begin{enumerate}
			\item Run $A$ on registers $MR$.
			\item Sample a random $F \in \Clifford_{n+1}$ and a random $r \in \{0,1\}^{n+1}$.
			\item Prepare the state $F\kb{r}F^{\dagger}$ in a separate register $XT_1$, and apply the map $B$ to $XT_1R$, using the instruction $F^{\dagger}$ instead of $T$.
			\item Measure $XT_1$ in the computational basis, and check that the outcome is $r$. If so, submit $M$ to $\IdealProtocol^{\smeas{} \circ C}$, along with a bit $b = 0$ (no cheating). Otherwise, submit $M$ and $b = 1$.
		\end{enumerate}
	\end{simulator}
	Throughout this proof, we decompose the attack $B$ as
	\begin{align}
	B = \sum_{b,d \in \{0,1\}^{n+1}} \left(\X^b\Z^d\right)^{\reg{MT_1}} \otimes B_{b,d}^{\reg{R}},\label{eq:measurement-decompose-B}
	\end{align}
	and similarly as before, we abbreviate $B_b := \sum_d B_{b,d}$ (and $B_0$ for $B_{0^{n+1}}$).

	We analyze the output state in registers $RS$ (note that the $MT_1$ registers are destroyed by the measurement) in both the ideal and the real case, and aim to show that they are indistinguishable, whatever the input $\rho^{\reg{MR}}$ was.

	In the ideal (simulated) case, first consider the output state in case of accept. Write $\mathcal{C}\left(\cdot\right)$ for the map induced by the circuit $C$. Following the steps of the simulator, abbreviating $\sigma = A^{\reg{MR}}\rho^{\reg{MR}}A^{\dagger}$, and decomposing $B$ as in Equation~\eqref{eq:measurement-decompose-B}, we see that the output in $RS$ in case of accept is

	\begin{align}
	&\sum_{m\in\{0,1\}^m}
	\E_r
	\bra{m}^{\reg{M}}
	\mathcal{C}^{\reg{M}}\left(
	\bra{r}^{\reg{XT_1}}
	B^{\reg{XT_1R}}
	\left(
    \sigma
	\otimes \left(F^{\dagger}F\kb{r}F^{\dagger}F\right)^{\reg{XT_1}}
	\right) B^{\dagger}
	\ket{r}
	\right)
	\ket{m}
	\otimes \kb{m}^{\reg{S}}\nonumber\\
	=&\sum_{m \in \{0,1\}}\sum_{b,d,b',d'}
	\E_r
	\bra{m}^{\reg{M}}\left(
	\mathcal{C}^{\reg{M}}
	\left(B_{b,d}^{\reg{R}}\sigma B_{b',d'}^{\dagger}\right)
	\otimes \bra{r}\X^b\Z^d\kb{r}\Z^{d'}\X^{b'}\ket{r}^{\reg{XT_1}}
	\right)
	\ket{m}
	\otimes \kb{m}^{\reg{S}}\nonumber\\
	=&\sum_{m \in \{0,1\}}
	\bra{m}^{\reg{M}}
	\mathcal{C}^{\reg{M}}
	\left(B_0^{\reg{R}}A^{\reg{MR}}\rho^{\reg{MR}} A^{\dagger}B_{0}^{\dagger}\right)
	\ket{m}
	\otimes \kb{m}^{\reg{S}}.\label{eq:measurement-ideal-accept}
	\end{align}
	The ideal reject case is similar, except we project onto $\Id - \kb{r}$ instead of onto $\kb{r}$. The output state is
	\begin{align}
	&\sum_{m \in \{0,1\}} \sum_{b \neq 0^{n+1}}
	\bra{m}^{\reg{M}}
	\mathcal{C}^{\reg{M}}
	\left(B_b^{\reg{R}}A^{\reg{MR}}\rho^{\reg{MR}} A^{\dagger}B_{b}^{\dagger}\right)
	\ket{m}
	\otimes \kb{\bot}^{\reg{S}}\nonumber\\
	=&\sum_{b \neq 0^{n+1}}
	Tr_M \left[ B_b^{\reg{R}}A^{\reg{MR}}\rho^{\reg{MR}} A^{\dagger}B_{b}^{\dagger}\right]
	\otimes \kb{\bot}^{\reg{S}}.\label{eq:measurement-ideal-reject}
	\end{align}

	In the real protocol, the unitary $T$ does not reveal any information about $c$, so the attack $B$ is independent of it. This allows us to apply Lemma~\ref{lem:measurement-cnot-trick}, after performing a Pauli twirl to decompose the attack $B$. Again abbreviating $\sigma = A\rho A^{\dagger}$, the state in the accept case is

	\begin{align*}
	&\phantomsection{=} \E_c \sum_m \bra{r \oplus (m,m\cdot c)}^{\reg{MT_1}} B^{\reg{MT_1R}}
	\X^r\Z^s
	\CNOT_{1,c}
	E^{\dagger} E
	\left(
	\mathcal{C}^{\reg{M}}\left(\sigma\right)
	\otimes \kb{0^n}^{\reg{T_1}}
	\right)\nonumber\\
	&\qquad \qquad \qquad
	E^{\dagger} E
	\CNOT_{1,c}^{\dagger}
	\Z^s\X^r
	 B^{\dagger}\ket{r \oplus (m,m\cdot c)}  \otimes \kb{m}^{\reg{S}}\\
	 &= \E_c \sum_{m,b} \bra{m,m\cdot c}
	 \X^b
	 \CNOT_{1,c}
	 \left(
	 \mathcal{C}^{\reg{M}}\left(B_b^{\reg{R}}\sigma B_b^{\dagger}\right)
	 \otimes \kb{0^n}
	 \right)
	 \CNOT_{1,c}^{\dagger}
	 \X^b
	 \ket{m,m\cdot c}
	 \otimes \kb{m}^{\reg{S}}\\
	 &\approx_{2^{-n}} \text{Eq.} \ \eqref{eq:measurement-ideal-accept}.
	\end{align*}
	For the last step, observe that the probabilities $p(b)$ in the statement of Lemma~\ref{lem:measurement-cnot-trick} are part of $B_b$.

	Similarly, the real reject state is
	\begin{align*}
	&\phantom{=} \E_c \sum_{m,b} \sum_{x \neq (m,m\cdot c)}\bra{x}^{\reg{MT_1}}
	\X^b
	\CNOT_{1,c}
	\left(
	\mathcal{C}^{\reg{M}}\left(B_b^{\reg{R}}\sigma B_b^{\dagger}\right)
	\otimes \kb{0^n}^{\reg{T_1}}
	\right)
	\CNOT_{1,c}^{\dagger}
	\X^b
	\ket{x}  \otimes \kb{\bot}^{\reg{S}}\\
	&\approx_{2^{-n}} \text{Eq.} \ \eqref{eq:measurement-ideal-reject}.
	\end{align*}
	In summary, we have shown that the output state in the real case is close to Eq.~\eqref{eq:measurement-ideal-accept} + Eq.~\eqref{eq:measurement-ideal-reject}, for any input state $\rho^{\reg{MR}}$ provided by the environment $\envE$.
\end{proof}

\section{Proof of \Cref{lem:protocol-ms-states}}
\label{ap:protocol-ms-states}
Before proving~\Cref{lem:protocol-ms-states}, we discuss the task of sampling in
the quantum world.

Classically, some properties of a bit-string can be
estimated just by querying a small fraction of it. For instance, in order to
estimate the Hamming weight $w$ of a $n$-bit string $x$, one could calculate the
Hamming weight $w_S$ of a subset $S$ of the bits of $x$ and we have that
$w \in [w_S-\delta, w_S+\delta]$ except with probability $O(2^{-2\delta^2|S|})$.

Such a result does not follow directly in the quantum setting,
since the tested quantum state could, for instance, be entangled with the
environment. However, Bouman and Fehr~\cite{BoumanF10} have studied this problem in the quantum setting,
and they showed that such sampling arguments also hold in the quantum setting,
but with a quadratic loss in the error probability. A corollary of their result
that will be important in this work is the following.

\begin{lemma}[Application of Theorem $3$ of \cite{BoumanF10}]\label{lem:sampling}
  Let $\ket{\phi_{AE}} \in (\mathbb{C}^2)^{\otimes n} \otimes \mathcal{H}_E$ be
  a quantum state and let $B = \{\ket{v_0}, \ket{v_1}\}$ be
  a fixed single-qubit basis. If we measure $k$ random qubits of
  $Tr_E(\kb{\phi_{AE}})$ in the $B$-basis and all of the outcomes are
  $\ket{v_0}$, then with probability $1- O(2^{-\delta^2 k})$, we have that
  \[\ket{\phi_{AE}} \in \mathrm{span}\left((P_\pi \ket{v_0}^{\otimes n-t}\ket{v_1}^{\otimes
  t}) \otimes \ket{\psi}
  : 0 \leq t \leq \delta n, \pi \in S_n, \ket{\psi} \in \mathcal{H}_E\right)).\]
\end{lemma}

\medskip
We now proceed to the proof of~\Cref{lem:protocol-ms-states}.

\begin{proof}[Proof of \Cref{lem:protocol-ms-states}]
  The simulator for $\RealProtocol^{MS}$ is similar to the composed simulator for
  $\RealProtocol^{\Dec} \diamond \RealProtocol^{C} \diamond \RealProtocol^{\Enc}$,
  where $C$ is the Clifford circuit of Protocol \ref{protocol:distillation}. The
  difference is that the input is now chosen by player 1 instead of being given
  by the environment, and that each player tests if the decoded qubit is correct. We make a small modification for each of the following cases:

  \paragraph{Case 1: player 1 is honest.} In this case, the simulator only needs to
  also set $b_i = 1$ whenever the adversary aborts after it receives the output
  of the ideal quantum computation in Step~\ref{step:abort-magic} of Protocol \ref{protocol:ms-creation}. Otherwise, the simulator is exactly the same as the composed one.

  \paragraph{Case 2: player 1 is dishonest.}
  In this case, the simulator also tests if the decoded qubits by the
  (simulated) honest players in $[k] \setminus I_{\advA}$ are indeed magic states of the correct form.
  More concretely, the
  simulator also measures all the qubits that the simulated players receive
  in the $\{\ket{\T},
  \ket{\T^\perp}\}$ basis, and sets $b_i = 1$ if
  any of the outcomes is $\ket{\T^\perp}$.
  Otherwise, the simulator replaces the qubits in $[\ell] \setminus
  \left(\bigcup_{2 \leq i \leq k} S_i\right)$ by true magic-states
  $\ket{\T}$, re-encodes them, and continues the composed simulation.
  Notice that this change makes the simulator abort with the same
  probability that an honest player would abort in
  Step~\ref{step:abort-magic}.

  \medskip

  We now argue that when there is no abort, the output of $\RealProtocol^{MS}$ is
  exponentially close to that of $\IdealProtocol^{MS}$.
  Notice that picking the disjoint $S_2,...,S_k \subseteq [\ell]$ uniformly at
  random is equivalent to first picking $\{S_i\}_{i \in I_{\advA}}$ from $[\ell]$,
  and then picking $\{S_i\}_{i \not\in \advA}$ from the
  remaining $[\ell] \setminus \left(\bigcup_{i \in \advA} S_i\right)$ elements. From this perspective, if the honest players do not abort in
  Step~\ref{step:abort-magic}, then \Cref{lem:sampling}
  implies that the state created by player $1$ in the other positions $[\ell] \setminus
  \left(\bigcup_{i \not\in \advA} S_i \right)$
  is
  $O(2^{\eps^2 (k-|I_{\advA}|)n})$-close to the the subspace
  $
    \mathrm{span}\left((P_\pi \ket{\T}^{\otimes tn-j}\ket{\T^\bot}^{\otimes
  j})
  : 0 \leq j \leq \eps tn, \pi \in S_{tn} \right)$.
  If we choose $\eps \leq \frac{1}{2}\left(1-\sqrt{3/7}\right)$, by
  \Cref{lem:distillation-works} and the union bound, the output of the distillation procedure is
  $O\left(t \eps\right)^{n^c}$-close to $\ket{\T}^{\otimes t}$. In this
  case, the output of $\RealProtocol^{MS}$ will be $\negl{n}$ close to
  encodings of $\ket{T}^{\otimes t}$, which is the output of
  $\IdealProtocol^{MS}$ in the no-abort case.
\end{proof}

\end{document}